\documentclass[pra,amsfonts, amssymb, amsmath,  showkeys, superscriptaddress, twocolumn,longbibliography]{revtex4-2}

\usepackage{float}
\makeatletter
\let\newfloat\newfloat@ltx
\makeatother
\usepackage[english]{babel}
\usepackage[utf8]{inputenc}
\usepackage{graphics}
\usepackage{selinput}

\usepackage{ytableau}

\usepackage{braket}
\usepackage{amsthm}
\usepackage{mathtools}
\usepackage{physics}
\usepackage{xcolor}
\usepackage{graphicx}
\usepackage[left=16mm,right=16mm,top=35mm,columnsep=15pt]{geometry} 
\usepackage{adjustbox}
\usepackage{placeins}
\usepackage[T1]{fontenc}
\usepackage{lipsum}
\usepackage{csquotes}
\usepackage{tikz}

\usepackage{tcolorbox}

\usepackage{upgreek}

\newtcolorbox[auto counter]{pabox}[2][]{fonttitle=\bfseries,
title=Example~\thetcbcounter: #2,#1,colframe=gray}

\usepackage[normalem]{ulem}
\usepackage{braket}
\usepackage{algorithm}
\usepackage[noend]{algpseudocode}

\def\HC{\mathcal{H}}

\def\LC{\mathcal{L}}

\def\ad{^{\dagger}}

\usepackage[makeroom]{cancel}
\usepackage[toc,page]{appendix}
\usepackage[colorlinks=true,citecolor=blue,linkcolor=magenta]{hyperref}

\newcommand{\poly}{\operatorname{poly}}

\newcommand{\BC}{\mathcal{B}}

\newcommand{\DC}{\mathcal{D}}

\newcommand{\GC}{\mathcal{G}}
\newcommand{\MC}{\mathcal{M}}

\newcommand{\OC}{\mathcal{O}}

\newcommand{\SC}{\mathcal{S}}
\newcommand{\TC}{\mathcal{T}}
\newcommand{\UC}{\mathcal{U}}

\newcommand{\Var}{{\rm Var}}

\renewcommand{\leq}{\leqslant}

\renewcommand{\Re}{\text{Re}}

\DeclareMathOperator*{\argmin}{arg\,min}
\renewcommand{\vec}[1]{\boldsymbol{#1}}  

\newcommand*{\id}{\openone}

\newcommand{\bs}{\textsf{BS}}

\renewcommand{\th}{\theta } 

\newcommand{\lm}{\lambda }

\newcommand{\thv}{\boldsymbol{\uptheta}}

\def\be{\begin{equation}}
\def\ee{\end{equation}}
\def\bs{\begin{split}}
\def\e{\end{split}}
\def\ba{\begin{eqnarray}}
\def\bea{\begin{eqnarray}}

\def\tea{\end{eqnarray}}
\def\ea{\end{eqnarray}}
\def\eea{\end{eqnarray}}

\def\l{\lambda}

\def\tn{^{\otimes n}}
\def\tn{^{\otimes n}}

\newcommand\mf[1]{\mathfrak{#1}}
\newcommand\mbb[1]{\mathbb{#1}}

\newcommand\spn{\text{span}}
\newcommand\X{\text{X}}

\newtheorem{theorem}{Theorem}

\newtheorem{lemma}{Lemma}
\newtheorem{corollary}{Corollary}

\newtheorem{definition}{Definition}

\usepackage{amssymb}
\usepackage{dsfont}

\def\be{\begin{equation}}
\def\te{\end{equation}}
\def\ee{\end{equation}}
\def\ba{\begin{eqnarray}}
\def\bea{\begin{eqnarray}}

\def\tea{\end{eqnarray}}
\def\ea{\end{eqnarray}}
\def\eea{\end{eqnarray}}

\definecolor{myblue}{RGB}{0,163,243}
\definecolor{myred}{RGB}{255,100,100}
\definecolor{mygreen}{RGB}{0,153,0}
\definecolor{mypurple}{RGB}{153,51,253}

\newtcolorbox{blue_boxed_example}[1]{colback=myblue!5!white,colframe=myblue!75!black,fonttitle=\bfseries,title=#1}
\newtcolorbox{red_boxed_example}[1]{colback=myred!5!white,colframe=myred!75!black,fonttitle=\bfseries,title=#1}
\newtcolorbox{green_boxed_example}[1]{colback=mygreen!5!white,colframe=mygreen!75!black,fonttitle=\bfseries,title=#1}
\newtcolorbox{purple_boxed_example}[1]{colback=mypurple!5!white,colframe=mypurple!75!black,fonttitle=\bfseries,title=#1}

\usepackage{tikz}
\usetikzlibrary{quantikz}
\usepackage{braket, enumitem}

\newcommand{\er}{\textnormal{e}}
\newcommand{\ir}{\textnormal{i}}

\usepackage{fancyhdr}
\fancypagestyle{headings}{%
\cfoot{\thepage}
}

\begin{document}
\title{Theoretical Guarantees for Permutation-Equivariant Quantum Neural Networks}

\author{Louis Schatzki}
\email{louisms2@illinois.edu}
\affiliation{Information Sciences, Los Alamos National Laboratory, Los Alamos, New Mexico 87545, USA}
\affiliation{Electrical and Computer Engineering, University of Illinois Urbana-Champaign, Urbana, Illinois, 61801, USA}

\author{Mart\'{i}n Larocca}
\thanks{The two first authors contributed equally.}
\affiliation{Theoretical Division, Los Alamos National Laboratory, Los Alamos, New Mexico 87545, USA}
\affiliation{Center for Nonlinear Studies, Los Alamos National Laboratory, Los Alamos, New Mexico 87545, USA}

\author{Quynh T. Nguyen}
\affiliation{Theoretical Division, Los Alamos National Laboratory, Los Alamos, New Mexico 87545, USA}
\affiliation{Harvard Quantum Initiative, Harvard University, Cambridge, Massachusetts 02138, USA}

\author{Fr\'{e}d\'{e}ric Sauvage}
\affiliation{Theoretical Division, Los Alamos National Laboratory, Los Alamos, New Mexico 87545, USA}

\author{M. Cerezo}
\email{cerezo@lanl.gov} 
\affiliation{Information Sciences, Los Alamos National Laboratory, Los Alamos, New Mexico 87545, USA}

\begin{abstract}
    Despite the great promise of quantum machine learning models,  there are several challenges one must overcome before unlocking their full potential. For instance, models based on quantum neural networks (QNNs) can suffer from excessive local minima and barren plateaus in their training landscapes. Recently, the nascent field of geometric quantum machine learning (GQML) has emerged as a potential solution to some of those issues. The key insight of GQML is that one should design architectures, such as equivariant QNNs, encoding the symmetries of the problem at hand.  Here, we focus on problems with permutation symmetry  (i.e., symmetry group $S_n$), and show how to build $S_n$-equivariant QNNs. We provide an analytical study of their performance, proving that  they do not suffer from barren plateaus, quickly reach overparametrization, and generalize well from small amounts of data. To verify our results, we perform numerical simulations for a graph state classification task. Our work provides theoretical guarantees for equivariant QNNs, thus indicating the power and potential of GQML.
\end{abstract}
\maketitle

\section*{Introduction}

Symmetry studies and formalizes the invariance of objects under some set of operations. A wealth of theory has gone into describing symmetries as mathematical entities through the concept of groups and representations. While the analysis of  symmetries in nature has greatly improved our understanding of the laws of physics, the study of symmetries in data has just recently gained momentum within the framework of learning theory.  In the past few years, classical machine learning practitioners realized that models tend to perform better when constrained to respect the underlying symmetries of the data. This has led to the blossoming field of geometric deep learning~\cite{cohen2016group,bronstein2021geometric,  kondor2018generalization, bogatskiy2022symmetry,bekkers2018roto}, where symmetries are incorporated as geometric priors into the learning architectures, improving trainability and generalization performance~\cite{schutt2017continuous,boyda2021sampling,rezende2019equivariant,thomas2018tensor,toth2019hamiltonian,kohler2020equivariant,anderson2019cormorant,bogatskiy2020lorentz}. 

The tremendous success of geometric deep learning has recently inspired researchers to import these ideas to the realm of quantum machine learning (QML)~\cite{schuld2015introduction,biamonte2017quantum,cerezo2022challenges}. QML is a new and exciting field at the intersection of classical machine learning, and quantum computing. By running routines in quantum hardware, and thus exploiting the exponentially large dimension of the Hilbert space, the hope is that QML algorithms can outperform their classical counterparts when learning from data~\cite{huang2021provably}.

The infusion of ideas from geometric deep learning to QML has been termed "geometric quantum machine learning"(GQML)~\cite{larocca2022group,meyer2022exploiting, sauvage2022building, zheng2022super,zheng2021speeding,nguyen2022atheory,wang2022symmetric}. GQML leverages the machinery of group and representation theory~\cite{ragone2022representation} to build quantum architectures that encode symmetry information about the problem at hand. For instance, when the model is parametrized through a quantum neural network (QNN)~\cite{abbas2020power,cerezo2022challenges,liu2022analytic,liu2021representation}, GQML indicates that the layers of the QNN should be equivariant under the action of the symmetry group associated to the dataset. That is, applying a symmetry transformation on the input to the QNN layers should be the same as applying it to its output. 

One of the main goals of GQML is to create architectures that solve, or at least significantly mitigate, some of the known issues of standard symmetry non-preserving QML models~\cite{cerezo2022challenges}. For instance, it has been shown that the optimization landscapes of generic QNNs can exhibit a large number of local minima~\cite{bittel2021training,anschuetz2022beyond,fontana2022nontrivial,larocca2021theory}, or be prone to the barren plateau phenomenon~\cite{mcclean2018barren,cerezo2020cost,sharma2020trainability,holmes2020barren,holmes2021connecting,cerezo2020impact,marrero2020entanglement,patti2020entanglement,uvarov2020barren,thanasilp2021subtleties,larocca2021diagnosing,wang2020noise,wiersema2020exploring} whereby the loss function gradients vanish exponentially with the problem size. Crucially, it is known that barren plateaus and excessive local minima are connected to the expressibility~\cite{sim2019expressibility,holmes2021connecting,larocca2021diagnosing,larocca2021theory,anschuetz2022beyond} of the QNN, so that problem-agnostic architectures are more likely to exhibit trainability issues. In this sense, it is expected that following the GQML program of baking symmetry directly into the algorithm, will lead to models with  sharp inductive biases that suitably limit their expressibility and search space.

In this work we leverage the GQML toolbox to create models that are permutation invariant, i.e., models whose outputs remain invariant under the action of the symmetric group~$S_n$ (see Fig.~\ref{fig:fig1}). We focus on this particular symmetry as learning problems with permutation symmetries abound. Examples include learning over sets of elements~\cite{zaheer2017deep,maron2020learning}, modeling relations between pairs (graphs)~\cite{maron2018invariant, keriven2019universal, maron2019provably,verdon2019quantumgraph,mernyei2022equivariant,skolik2022equivariant} or multiplets (hypergraphs) of entities~\cite{maron2019universality,thiede2020general,pan2022permutation}, problems defined on grids (such as condensed matter systems)~\cite{farhi2014quantum,hadfield2019quantum,cong2019quantum,caro2021generalization}, molecular systems~\cite{peruzzo2014variational,cerezo2020variationalreview,tang2019qubit}, evaluating genuine multipartite entanglement~\cite{horodecki2009quantum,walter2016multipartite,beckey2021computable,schatzki2022hierarchy}, or working with distributed quantum sensors~\cite{guo2020distributed,zhang2021distributed,huerta2022inference}.

Our first contribution is to provide guidelines to build unitary $S_n$-equivariant QNNs. We then derive rigorous theoretical guarantees for these architectures in terms of their trainability and generalization capabilities. Specifically, we prove  that $S_n$-equivariant QNNs do not lead to barren plateaus, can be overparametrized with polynomially deep circuits, and generalize well with a only a polynomial number of training points.  We also identify problems (i.e., datasets) for which the model is trainable, but also datasets leading to untrainability. All these appealing properties are also demonstrated in numerical simulations of a graph classification task. Our empirical results verify our theoretical ones, and even show that the performance of $S_n$-equivariant QNNs can, in practice, be better than that guaranteed by our theorems.

\section*{Results}

\subsection*{Preliminaries}

While the formalism of GQML can be readily applied to a wide range of tasks with  $S_n$ symmetry, here we will focus on supervised learning problems. We note however that our results can be readily extended to more general scenarios such as unsupervised learning~\cite{otterbach2017unsupervised,kerenidis2019q}, reinforced learning~\cite{saggio2021experimental,skolik2021quantum}, generative modeling~\cite{dallaire2018quantum,benedetti2019generative,kieferova2021quantum,romero2021variational}, or to the more task-oriented computational paradigm of variational quantum algorithms~\cite{cerezo2020variationalreview,bharti2021noisy}.

\begin{figure}[t]
    \centering
    \includegraphics[width=.8\columnwidth]{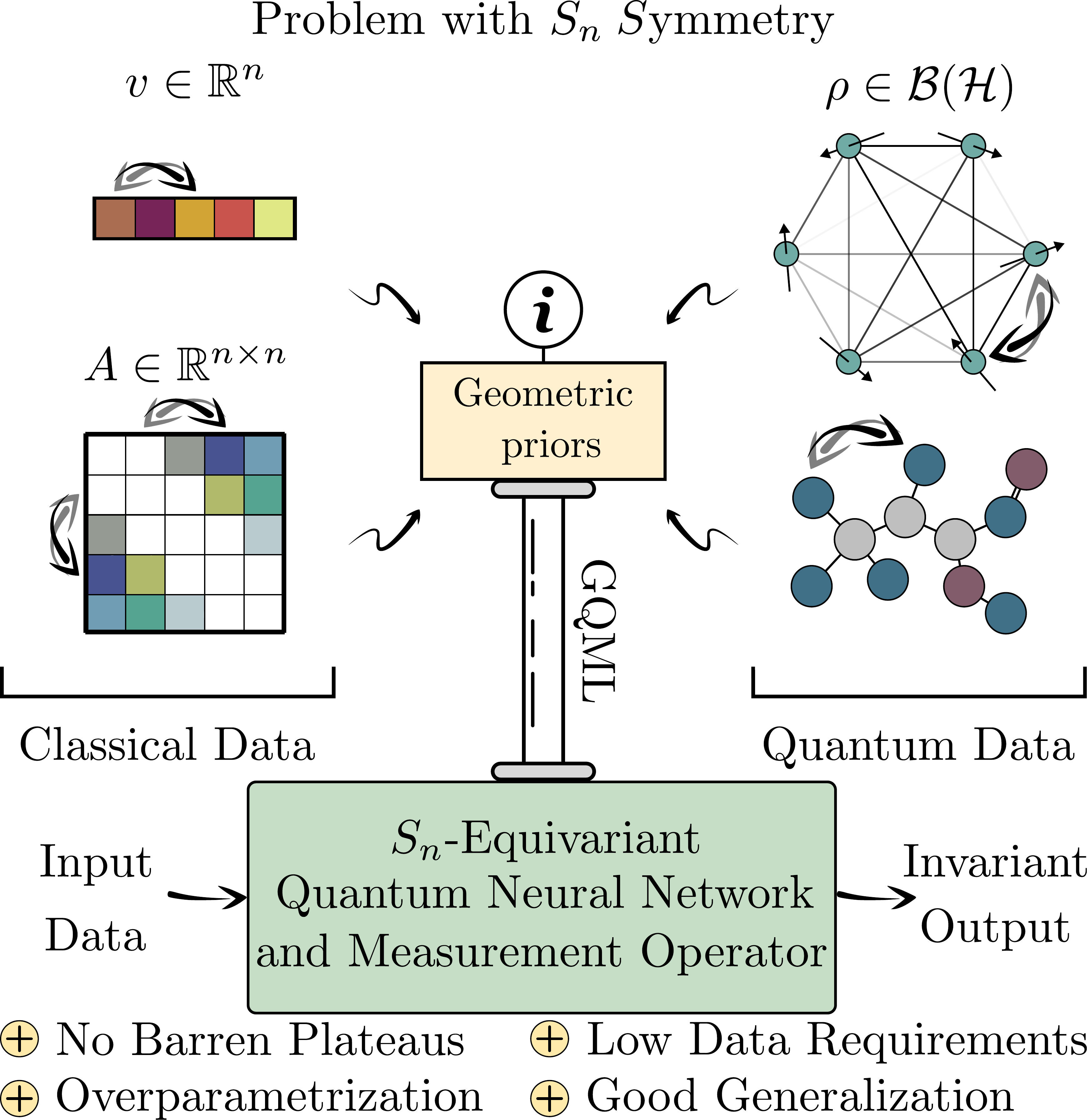}
    \caption{\textbf{GQML embeds geometric priors into a QML model.} Incorporating prior knowledge through $S_n$-equivariance heavily restricts the search space of the model. We show that such inductive biases lead to models that do not exhibit barren plateaus, can be efficiently overparametrized, and require small amounts of data to generalizing well.}
    \label{fig:fig1}
\end{figure}

Generally, a supervised quantum machine learning task can be phrased in terms of a data space $\mathcal{R}$ -a set of quantum states on some Hilbert space $\mathcal{H}$- and a real-valued label space $\mathcal{Y}$. We will assume $\mathcal{H}$ to be a tensor product of $n$ two-dimensional subsystems (qubits) and thus of dimension $d=2^n$. We are given repeated access to a training dataset $\SC = \{(\rho_i, y_i)\}_{i=1}^M$, where $\rho_i$ is sampled from $\mathcal{R}$ according to some probability $P$, and where $y_i\in\mathcal{Y}$. We further assume that the labels are assigned by some underlying (but unknown) function $f:\mathcal{R}\mapsto\mathcal{Y}$, that is, $y_i = f(\rho_i)$. We make no assumptions regarding the origins of $\rho_i$, meaning that these can correspond to classical data embedded in quantum states~\cite{havlivcek2019supervised,schuld2018supervised}, or to quantum data obtained from some quantum mechanical process~\cite{cong2019quantum,schatzki2021entangled,caro2021generalization}.

The goal is to produce a parametrized function $h_{\thv}:\mathcal{R}\mapsto\mathcal{Y}$ closely modeling the outputs of the unknown target $f$, where $\boldsymbol{\uptheta}$ are trainable parameters. That is, we want $h_{\thv}$ to accurately predict labels for the  data in the training set $\SC$ (low training error), as well as to predict the labels for new and previously unseen states  (small generalization error). We will focus on QML models that are parametrized through a QNN, a unitary channel $\UC_{\boldsymbol{\uptheta}}:\BC(\HC)\rightarrow\BC(\HC)$  such that $\UC_{\thv}(\rho)=U(\thv)\rho U(\thv)\ad$. Here, $\BC(\HC)$ denotes the space of bounded linear operators in $\HC$. Throughout this work we will restrict to $L$-layered QNNs 
\begin{equation}
\UC_{\thv}=\UC^L_{\theta_L}\circ\cdots\circ \UC^1_{\theta_1}\,,\,\,\, \text{where}\,\,\,\, \UC^l_{\theta_l}(\rho)=\er^{-\ir \theta_l H_l}\rho \er^{\ir \theta_l H_l}\,,
\end{equation}
for some Hermitian generators $\{H_l\}$, so that $U(\boldsymbol{\uptheta})=\prod_{l=1}^L \er^{-\ir \theta_l H_l}$. Moreover, we consider models that depend on a loss function of the form
\begin{equation}\label{eq_loss}
    \ell_{\thv}(\rho_i)=\Tr[\UC_{\thv}(\rho_i) O]\,,
\end{equation}
where $O$ is a Hermitian observable.  We quantify the training error via the so-called empirical loss, or training error, which  is defined as 
\begin{equation}\label{eq:empirical-loss}
    \widehat{\LC}(\thv)=\sum_{i=1}^M c_i\ell_{\thv}(\rho_i)\,.
\end{equation}
The model is trained by solving the optimization task $\argmin_{\thv} \widehat{\LC}(\thv)$~\cite{cerezo2020variationalreview}. Once a desired convergence in the optimization is achieved, the optimal parameters, along with the loss function $\ell_{\thv}$, are used to predict labels. For the case of binary classification, where $\mathcal{Y} = \{+1,-1\}$, one can choose $c_i := -\frac{y_i}{M}$. Then, if the measurement operator is normalized such that $\ell_{\thv}(\rho_i) \in [-1,1]$, this corresponds to the hinge loss, a standard loss function but not the only relevant one~\cite{janocha2017loss}) in machine learning.

We further remark that while Eq.~\eqref{eq:empirical-loss} approximates the error of the learned model, the true loss is defined as 
\begin{equation}\label{eq:gen-error}
    \LC(\thv) = \mathbb{E}_{\rho \sim P}[c(y)\ell_{\thv}(\rho)]\,.
\end{equation}
Here we have denoted the weights as $c(y)$ to make their dependency on the labels $y$ explicit. The difference between the true loss and the empirical one, known as the generalization error, is given by 
\begin{equation}\label{eq:generalization}
    {\rm gen}(\thv) = \lvert \LC(\thv) - \hat{\LC}(\thv) \rvert.
\end{equation}

We now turn to GQML, where the first step is identifying the underlying symmetries of the dataset, as this allows us to create suitable inductive biases for $h_{\thv}$. In particular, many problems of interest exhibit so-called label symmetry, i.e., the function $f$ produces labels that remain invariant under a set of operations on the inputs. Concretely, one can verify that such set of operations forms a group~\cite{larocca2022group}, which leads to the following definition. 
\begin{definition}[Label symmetries and $G$-invariance]
\label{def:label_symm}
Given a compact group $G$ and some unitary representation $R$ acting on quantum states $\rho$, we say $f$ has a label symmetry if it is $G$-invariant, i.e., if
\begin{equation}
f(R(g) \rho R(g)\ad) = f(\rho),\ \forall g\in G\,.
\end{equation}
\end{definition}
\noindent Here, we recall that a  representation is a mapping of a group into the space of invertible linear operators on some vector space (in this case the space of quantum states) that preserves the structure of the group~\cite{ragone2022representation}. Also, we note that some problems may have functions $f$ whose outputs change (rather than being invariant) in a way entirely determined by the action of $G$ on their inputs. While still captured by general GQML theory, these do not pertain to Definition~\ref{def:label_symm} and are not discussed further. Label invariance captures the scenario where the relevant information in $\rho$ is unchanged under the action of $G$.

Evidently, when searching for models $h_{\thv}$ that accurately predict outputs of $f$, it is natural to restrict our search to the space of models that respect the label symmetries of $f$. In this context, the theory of GQML provides a constructive approach to create $G$-invariant models, resting on the concept of equivariance~\cite{nguyen2022atheory}.
\begin{definition}[Equivariance]\label{def:equivariance}
We say that an observable $O$ is $G$-equivariant iff for all elements $g\in G$, $[O,R(g)]=0$. We say that a layer $\UC^l_{\theta_l}$ of a QNN is $G$-equivariant iff it is generated by a $G$-equivariant Hermitian operator. 
\end{definition}
By the previous definition, $G$-equivariant layers are maps that commute with the action of the group
\begin{equation}\label{eq:equivariance-U}
    \UC^l_{\th_l}(R(g) \rho R(g)^\dagger) = R(g) \UC^l_{\th_l}(\rho) R(g)^\dagger\,.
\end{equation}
Definition~\ref{def:equivariance} can be naturally extended to QNNs.

\begin{definition}[Equivariant QNN]\label{def:equivariant_QNN}
We say that a $L$-layered QNN is $G$-equivariant iff each of its layers is $G$-equivariant.
\end{definition}

Altogether, equivariant QNNs and measurement operators provide a recipe to design invariant models, i.e. models that respect the label symmetries.
Akin to their classical machine learning counterparts~\cite{cohen2016group,bronstein2021geometric,  kondor2018generalization, bogatskiy2022symmetry,bekkers2018roto}, such GQML models consist in a composition of many equivariant operations (realized by the $L$ layers of the equivariant QNN) and an invariant one (realized by the measurement of the equivariant observable)~\cite{nguyen2022atheory}.  Furthermore, model invariance extends to the loss function itself, as captured by the following Lemma.

\begin{lemma}[Invariance from equivariance]\label{lem:inv}
A loss function of the form in Eq.~\eqref{eq_loss} is $G$-invariant if its composed of a $G$-equivariant QNN and measurement.
\end{lemma}

A proof of this Lemma along with that of the following Lemmas and Theorems are presented in Supplementary Methods 2 and 3.


\begin{figure}[t]
    \centering
    \includegraphics[width=0.8\linewidth]{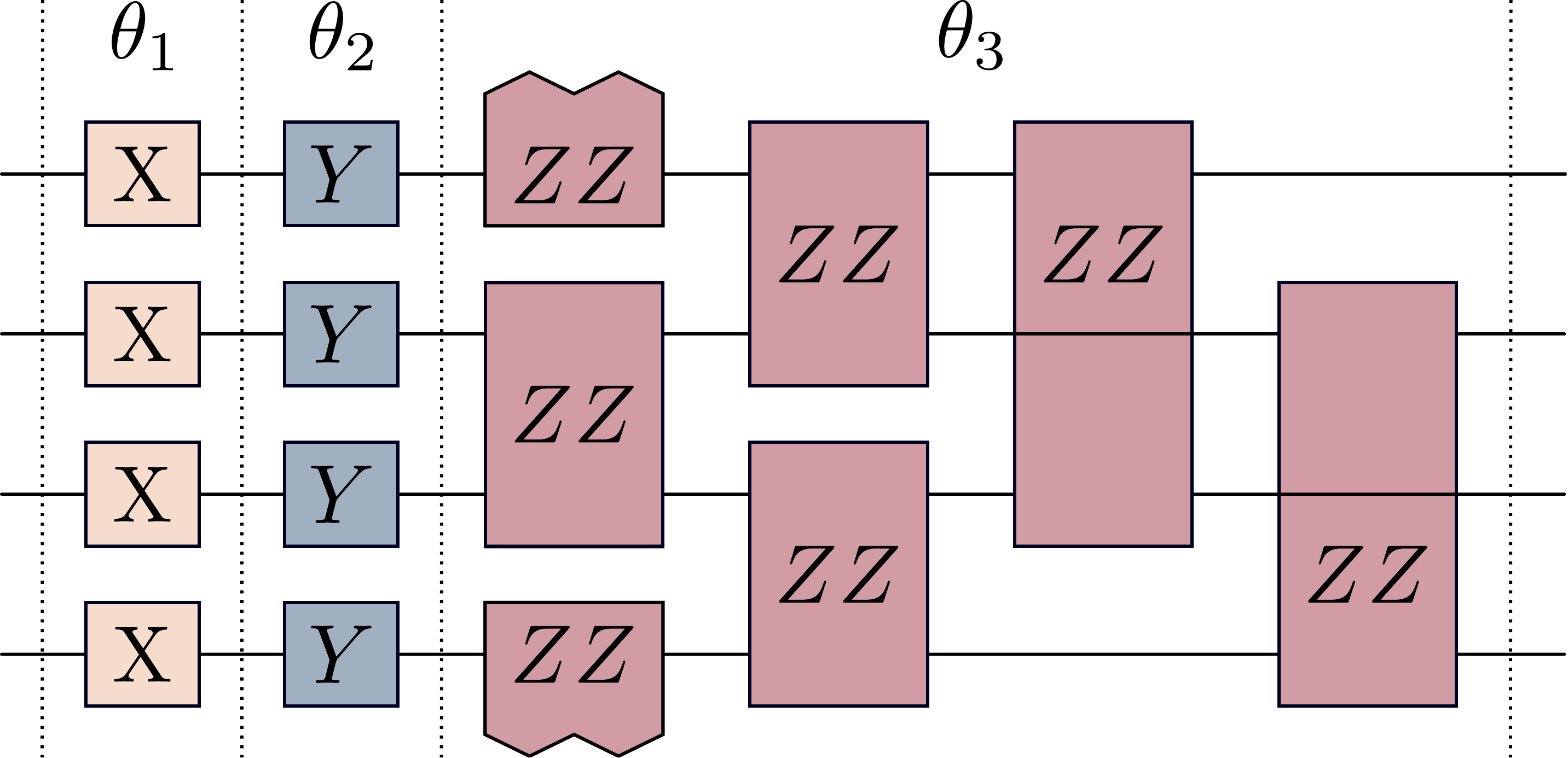}
    \caption{\textbf{Quantum circuit for an $S_n$-equivariant QNN.} Each layer of the QNN is obtained by exponentiation of a generator from the set $\GC$ in Eq.~\eqref{eq:generators-main}. Here we show a circuit with $L=3$ layers acting on $n=4$ qubits. Single-qubit blocks indicate a rotation about the $x$ or $y$ axis, while two-qubit blocks denote entangling gates generated by a $ZZ$ interaction. All colored gates between dashed horizontal lines share the same trainable parameter $\theta_l$. }
    \label{fig:fig-ansatz}
\end{figure}

\subsection*{$S_n$-Equivariant QNNs and measurements}

In the previous section we have described how to build generic G-invariant models. We now specialize to the case where $G$ is the symmetric group $S_n$, and where $R$ is the qubit-defining representation of $S_n$, i.e., the one permuting qubits which for any $\pi\in S_n$ acts as  
\begin{equation}\label{eq:qubit_permut}
R(\pi)\bigotimes_{i=1}^n \ket{\psi_i} = \bigotimes_{i=1}^n \ket{\psi_{\pi^{-1}(i)}}\,.
\end{equation}

\begin{figure*}[h!t]
    \centering
    \includegraphics[width=1\linewidth]{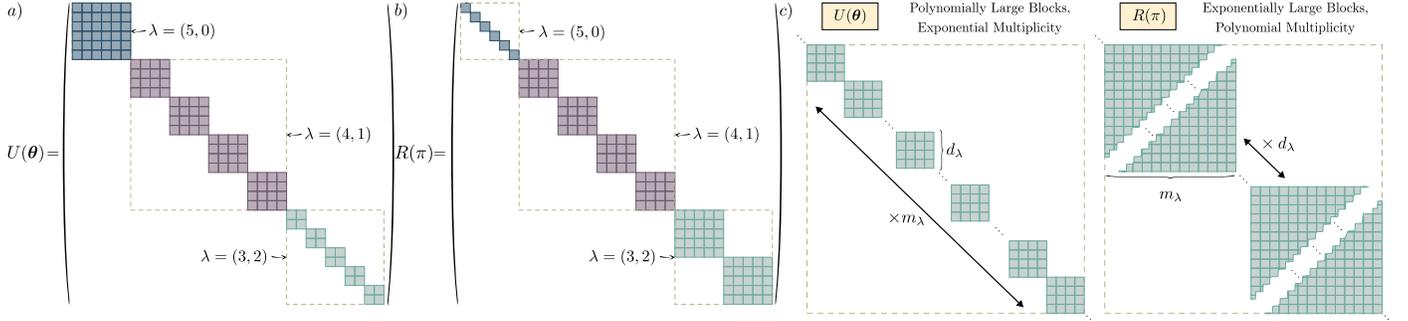}
    \caption{\textbf{Representation theory and $S_n$-equivariance.} Using tools from representation theory we find that the $S_n$-equivariant  QNN $U(\thv)$ and the representation of the group elements $R(\pi)$ -for any $\pi\in S_n$- admit an irrep block  decomposition as in Eq.~\eqref{eq:commutator} and  Eq.~\eqref{eq:Isotypic}, respectively. The irreps can be labeled with a single parameter $\lambda=(n-m,m)$ where $m=0,1,\ldots,\lfloor\frac{n}{2}\rfloor$. For a system of $n=5$ qubits, we show in  a) the block diagonal decomposition for $U(\thv)$ and in b) the decomposition of $R(\pi)$ as a representation of $S_5$. The dashed boxes denote the isotypic components labeled by $\lambda$. c)  As $n$ increases, $U(\thv)$ has a block diagonal decomposition which contains polynomially large blocks repeated a (potentially) exponential number of times. In contrast,  the block decomposition of $R(\pi)$ (for any $\pi\in S_n$)  contains blocks that can be exponentially large but that are only repeated a polynomial number of times.  }
    \label{fig:fig2}
\end{figure*}

Following Definitions~\ref{def:equivariance} and~\ref{def:equivariant_QNN}, the first step towards building $S_n$-equivariant QNNs is defining $S_n$-equivariant generators for each layer. In the Methods section we describe how such operators can be obtained, but here we will restrict our attention to the following set of generators
\begin{equation}\label{eq:generators-main}
    \GC=\left\{\frac{1}{n}\sum_{j=1}^n X_j,\frac{1}{n}\sum_{j=1}^n Y_j,\frac{2}{n(n-1)}\sum_{k<j} Z_jZ_k \right\}\,.
\end{equation}
Note that there is some freedom in the choice of generators. Any two sums over two distinct single qubit Pauli operators (the first two generators) plus a sum over pairs of the remaining Pauli operator (the third generator) suffices and we choose the above set without loss of generality. In Fig.~\ref{fig:fig-ansatz} we show an example of an $L=3$ layered $S_n$-equivariant QNN acting on $n=4$ qubits. 
While the single-qubit rotations generated by $\mathcal{G}$ are readily achievable in most quantum computing platforms, the collective $ZZ$ interactions are best suited to architectures allowing for reconfigurable connectivity~\cite{grzesiak2020efficient,pino2021demonstration,bluvstein2022quantum} or platforms that implement mediated all-to-all interactions~\cite{pedrozo2020entanglement,marciniak2022optimal}.
In fact, such interactions are referred to as one-axis twisting~\cite{kitagawa1993squeezed} in the context of spin squeezing~\cite{wineland1992spin} and form the basis of many quantum sensing protocols.

In addition, we will consider observables of the following form

\begin{equation}\label{eq:measurement-op}
    \MC=\left\{\frac{1}{n}\sum_{j=1}^n \chi_j, \frac{2}{n(n-1)}\sum_{k<j;j=1}^n \chi_j\chi_k,\prod_{j=1}^n\chi_j\right\}\,,
\end{equation}
where $\chi$ is a (fixed) Pauli matrix. It is straightforward to see that any $H_l\in\GC$ and $O\in \MC$ will commute with  $R(\pi)$ for any $\pi\in S_n$. We note that one could certainly consider other observables as well.

We now leverage tools from representation theory to understand and unravel the underlying structure of $S_n$-equivariant QNNs and measurement operators. The previous will allow us to derive, in the next section, theoretical guarantees for these GQML models.

One of the most notable results from representation theory is that  a given  finite dimensional representation  of a group decomposes into an orthogonal direct sum of fundamental building-blocks known as irreducible representations (irreps). As further explained in the Methods, the qubit-defining representation takes, under some appropriate global change of basis (which we denote with $\cong$), the block-diagonal form
\begin{equation}\label{eq:Isotypic}
R(\pi\in S_n)   \cong  \bigoplus_{\lambda}\bigoplus_{\mu=1}^{d_\lm} r_\lm(\pi)=\bigoplus_{\lm}r_\lm(\pi)\otimes\id_{d_\lm}\,.
\end{equation}
Here $\lm$ labels the irreps of $S_n$ and $r_\lm$ is the corresponding irrep itself, which appears $d_{\lm}$ times. The collection of these repeated irreps is called an isotypic component. Crucially, the only irreps appearing in $R$ correspond to two-row Young diagrams (see Methods) and can be parametrized by a single non-negative integer $m$, as $\lm\equiv\lm(m)=(n-m,m)$, where $m=0,1,\ldots,\lfloor\frac{n}{2}\rfloor$. It can be shown that
\begin{equation}\label{eq:dim-mult-lambda}
\begin{split}
d_\lm&=n-2m+1\,, \quad\text{ and} \\
m_\lm&=\frac{n!(n-2m+1)!}{(n-m+1)!m!(n-2m)!} \,
\end{split}
\end{equation}
where again $d_\lm$ is the number of times the irrep appears and $m_\lm$ is the dimension of the irrep itself. Note that every $d_\lm$ is in $\OC(n)$, whereas some $m_\lm$ can grow exponentially with the number of qubits. For instance, if $n$ is even and $m=n/2$, one finds that $m_\lambda = \Omega(4^n/n^2)$. We finally note that Eq.~\eqref{eq:Isotypic} implies $\sum_\lm m_\lm d_\lm = 2^n$.

Given the block-diagonal structure of $R$, $S_n$-equivariant unitaries and measurements must necessarily take the form
\begin{equation}\label{eq:commutator}
    U(\thv) \cong \bigoplus_{\lm} \id_{m_\lm} \otimes U_\lm(\thv), \quad \text{and} \quad
    O \cong \bigoplus_{\lm} \id_{m_\lm} \otimes O_\l.
\end{equation}
That is, both $U(\thv)$ and $O$ decompose into a direct sum of $d_\lm$-dimensional blocks repeated $m_\lm$ times (with $m_\lm$ called the multiplicity) on each isotypic component $\lm$. This decomposition is illustrated in Fig.~\ref{fig:fig2}. 

Let us highlight several crucial implications of the block diagonal structure arising from $S_n$-equivariance. First and foremost, we note that, under the action of an $S_n$-equivariant QNN, the Hilbert space decomposes as
\begin{equation}\label{eq:iso_space}
    \mathcal{H} \cong \bigoplus_{\lm}\bigoplus_{\nu=1}^{m_\lambda} \HC_\lm^\nu,
\end{equation}
where each $\HC_\lm^\nu$ denotes a $d_\lambda$-dimensional invariant subspace. Moreover, one can also see that when the QNN acts on an input quantum state as   $\UC_{\thv}(\rho)=U(\thv)\rho U(\thv\ad)$,  it can only access the information in $\rho$ which is contained in the invariant subspaces $\HC_\lm^\nu$ (see also~\cite{nguyen2022atheory}). This means that to solve the learning task, we require two ingredients: i) the data must encode the relevant information required for classification into these subspaces~\cite{ragone2022representation,nguyen2022atheory}, and ii) the QNN must be able to accurately process the information within each $\HC_\lm^\nu$. As discussed in the Methods, we can guarantee that the second condition will not be an issue, as the set of generators in Eq.~\eqref{eq:generators-main} is universal within each invariant subspace, i.e., the QNN can map any state in $\HC_\lm^\nu$ to any other state in $\HC_\lm^\nu$ (see also Ref.~\cite{albertini2018controllability}).

\begin{figure}[t]
    \centering
    \includegraphics[width=0.85\linewidth]{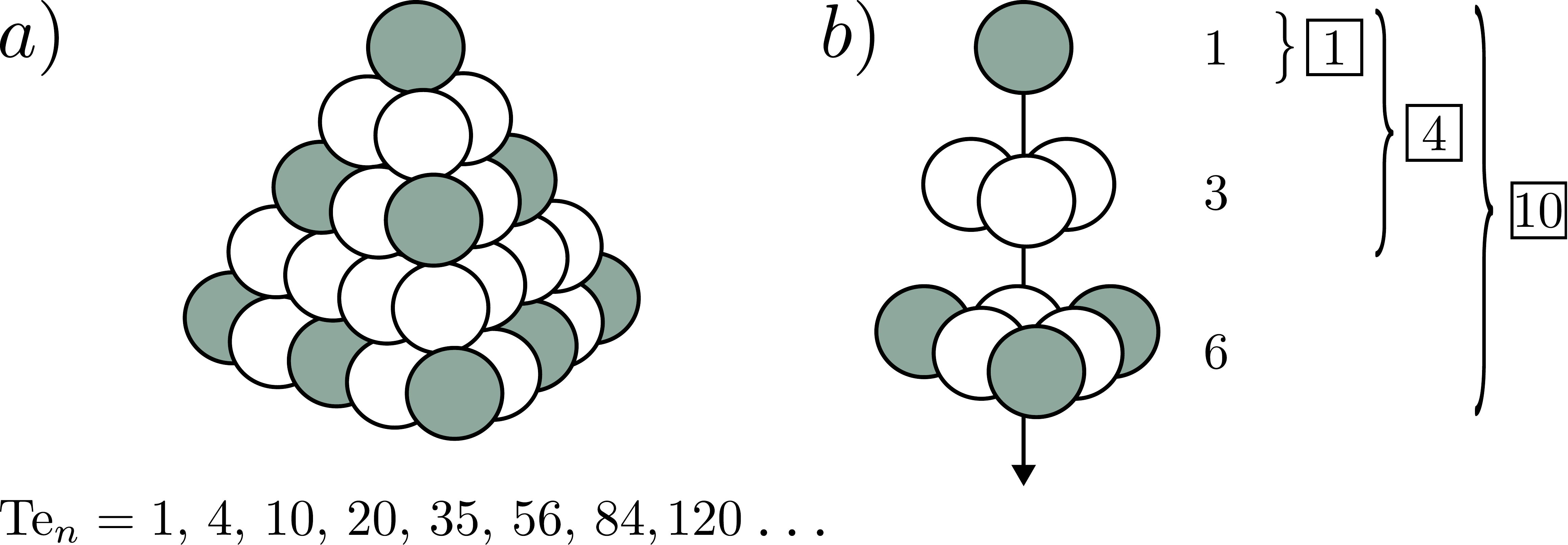}
    \caption{\textbf{Tetrahedral numbers.} a) The Tetrahedral numbers ${\rm Te}_n$ are obtained by  counting how many spheres can be stacked in the configuration of a tetrahedron (triangular base pyramid) of height $n$. b) One can also compute ${\rm Te}_n$ as the sum of consecutive triangular numbers, which count how many objects (e.g., spheres) can be arranged in an equilateral triangle. }
    \label{fig:figte}
\end{figure}

A second fundamental implication of Eq.\eqref{eq:commutator} is that the manifold of equivariant unitaries is of low-dimension. We make this explicit in the following lemma.

\begin{lemma}[Dimension of $S_n$-equivariant unitaries]\label{lem:poly}
The submanifold of $S_n-$equivariant unitaries is of dimension equal to the Tetrahedral numbers ${\rm Te}_{n+1}=\binom{n+3}{3}$ (see Fig.~\ref{fig:figte}), and therefore on the order of $\Theta(n^3)$.

\end{lemma}

Crucially, Lemma~\ref{lem:poly} shows that the equivariance constraint limits the degrees of freedom in the QNN (and concomitantly in any observable) from $4^n$ to only polynomially many.

\subsection*{Absence of barren plateaus in $S_n$-equivariant QNNs}

Barren plateaus have been recognized as one of the main challenges to overcome in order to guarantee the success of QML models using QNNs~\cite{cerezo2022challenges}. When a model exhibits a barren plateau, the loss landscape becomes, on average, exponentially flat and featureless as the problem size increases~\cite{mcclean2018barren,cerezo2020cost,sharma2020trainability,holmes2020barren,holmes2021connecting,cerezo2020impact,marrero2020entanglement,patti2020entanglement,uvarov2020barren,thanasilp2021subtleties,larocca2021diagnosing,wang2020noise,wiersema2020exploring}.  
This severely impedes its trainability, as  one needs to spend an exponentially large amount of resources  to correctly estimate a loss minimizing direction. 
Issues of barren plateaus arise primarily due to the structure of the models (including the choice of QNN, the input state and the observables)  employed~\cite{mcclean2018barren,cerezo2020cost,sharma2020trainability,holmes2020barren,holmes2021connecting,cerezo2020impact,marrero2020entanglement,patti2020entanglement,uvarov2020barren,thanasilp2021subtleties,larocca2021diagnosing,wiersema2020exploring} 
but can also be caused solely by effects of noise~\cite{wang2020noise}. 
In the rest of this section, we will only be concerned with the former type of barren plateaus, that is the most studied.

Recently, a great deal of effort has been put forward towards creating strategies capable of mitigating the effect of barren plateaus~\cite{grant2019initialization,skolik2020layerwise,sauvage2021flip,sack2022avoiding,rad2022surviving,broers2021optimization,liu2021parameter,friedrich2022avoiding,kulshrestha2022beinit,mele2022avoiding,zhang2022gaussian,grimsley2022adapt,cerezo2020variational,kieferova2021quantum}. While these are promising and have shown moderate success, the `holy grail' is identifying architectures which are immune to barren plateaus altogether, and thus enjoy trainability guarantees. Examples of such architectures are shallow hardware efficient ansatzes~\cite{cerezo2020cost}, quantum convolutional neural networks~\cite{pesah2020absence}, or the transverse field Ising model Hamiltonian variational ansatz~\cite{wiersema2020exploring,larocca2021diagnosing}. Here, we prove that another architecture can be added to this list: $S_n$-equivariant QNNs.

When studying barren plateaus, one typically analyzes the variance of the empirical loss function partial derivatives, $\partial_\mu \hat{\LC}(\thv)=\partial \hat{\LC}(\thv)/\partial\theta_\mu$, where  $\theta_\mu\in \thv$. We say that there is a barren plateau in the $\theta_\mu$ direction if $\mbb{E}_{\thv}[ \partial_\mu \hat{\LC}(\thv) ] = 0$ and $\Var_{\thv}[\partial_\mu \hat{\LC}(\thv)]$ is exponentially vanishing.

Before stating our main results, we introduce a bit of notation. 
Let us define $Q_\lambda^\nu$ to be the operator that maps vectors from $\HC$ to $\HC_\lambda^\nu$, such that $(Q_\lambda^\nu)^{\dagger} Q_\lambda^\nu$ realizes a projection onto $\HC_\lambda^\nu$ (see Supplementary Methods 4 for additional details).  
Given a matrix $B\in\mathbb{C}^{d \times d}$, we will denote its restriction to $\HC_\lambda^\nu$ as
\begin{equation}\label{eq:restriction}
  B_\lambda^{\nu} = Q_\lambda^\nu B (Q_\lambda^\nu)\ad\,,
\end{equation}
with  $B_\lm^\nu \in \mathbb{C}^{d_\lm \times d_\lm}$. We remark that the restriction of $S_n$-equivariant generators is independent of the $\nu$ multiplicity index (see Eq.~(13)). On the other hand, the restriction of non-equivariant operators (such as the input states $\rho_1$) are not independent of $\nu$, meaning that that the set composed of all the restrictions $\rho_\lambda^\nu$ contain an exponentially large amount of non-redundant information that the QNN can act on (see also~\cite{nguyen2022atheory}).

Denoting the weighted average of the input states as $\sigma =\sum_{i=1}^M c_i\rho_i$, we find:

\begin{theorem}[Variance of partial derivatives]\label{thm:main1}
     Let $\UC_{\thv}$ be an $S_n$-equivariant QNN, with  generators in $\GC$, and $O$ an $S_n$-equivariant measurement operator from $\MC$. Consider an empirical loss $\hat{\LC}(\thv)$ as in Eq.~\eqref{eq:empirical-loss}. Assuming a circuit depth $L$ such that the QNN forms independent 2-designs on each isotypic block, we  have $\langle \partial_\mu \hat{\LC}(\thv) \rangle_{\thv} = 0$, and 
     \small
    \begin{equation}\label{eq:var_main}
    \Var_{\thv}[ \partial_{\mu} \hat{\LC}(\thv)]\!=\!\sum_{\lambda}\frac{2d_\lambda}{(d_\lambda^2-1)^2}\Delta(H_{\mu,_\lambda})\Delta(O_{\lambda}) \Delta\left(\sum_{\nu=1}^{m_{\lambda}}\sigma_\lm^\nu\right)\,.
\end{equation}
\normalsize
Here, $\Delta(B) = \Tr[B^2]-\frac{\Tr[B]^2}{\dim(B)}$.
\end{theorem}
\noindent In the Methods we present a sketch of the proof for Theorem~\ref{thm:main1}, as well as its underlying assumptions.

We remark that while we have derived  Theorem~\ref{thm:main1} for $S_n$-equivariant QNNs and measurement operators, given some general finite-dimensional compact group $G$, the form of Eq.~\eqref{eq:var_main} is valid provided that one uses a $G$-equivariant QNN that is universal with each invariant subspace. In this case, the summation over $\lambda$ will run over the irreps of the representation of $G$.

Let us now analyze each term in Eq.~\eqref{eq:var_main} to identify potential sources of untrainability. First, let us consider the prefactors $\frac{2d_\lambda}{(d_\lambda^2-1)^2}$. From Eq.~\eqref{eq:dim-mult-lambda} we can readily see that $\frac{2d_\lambda}{(d_\lambda^2-1)^2}\in\Omega(\frac{1}{n^3})$ for any $\lm$. Next, it is convenient to separate the two remaining potential sources of barren plateaus into two categories: i) those that are QNN or measurement dependent, $\Delta(H_{\mu,_\lambda})$ and $\Delta(O_{\lambda})$, and ii) those that are dataset-dependent, $\Delta(\sum_{\nu}\sigma_\lm^\nu)$. This identification commonly appears when analyzing the absence of barren plateaus  (see Refs.~\cite{cerezo2020cost,pesah2020absence,larocca2021diagnosing,thanasilp2021subtleties,liu2021presence}) and allows one to study how the architecture and dataset individually affect the trainability. In what follows, we will say that some architecture does not induce  barren plateaus if the terms that are QNN or measurement dependent are not exponentially vanishing. 

Using tools from representation theory we can  obtain the following exact expressions for $S_n$-equivariant operators.
\begin{theorem}\label{thm:main2}
    Let $A$ be a $S_n$-equivariant operator.
\begin{equation}
        \begin{cases}
     \text{If } A=\sum_{j=1}^n \chi_j\,, \quad  \text{then } 
        \Delta(A_{\lm}) = 2\binom{d_\lm+1}{3}\,,\\
        \text{If }  A=\sum_{k<j} \chi_j\chi\,, \quad \text{then } 
        \Delta(A_{\lm}) = \frac{8}{3}\binom{d_\lm+2}{5}\,,\\
        \text{If } A=\prod_{j=1}^n \chi_j\,, \quad \text{then } \Delta(A_\lambda) = 
        \frac{d_\lm^2-1+n\,{\rm mod}2}{d_\lm}\,, 
    \end{cases}
\end{equation}
where $\chi\in\{\X,Y,Z\}$.
\end{theorem}
\noindent In Supplementary Methods 6 we also derive formulas for the case of $A$ being $k$-body operators.


Let us review the implications of Theorem~\ref{thm:main2}. First, note that all elements of our gate-set $\GC$ and measurement-set $\MC$ are of the form in Theorem~\ref{thm:main2}, and therefore belong in $\Omega(d_\lm)$. This follows from the fact that the binomial coefficient $\binom{n+a}{b}$ scales as a polynomial of degree $b$ in $n$. Since $d_\lm$ itself is in $\Theta(n)$ (see Eq.~\eqref{eq:dim-mult-lambda}), for all $\lm$ and $\mu$
\begin{equation}
\Delta(O_\lambda)\text{ and }\Delta(H_{\mu,\lambda}) \in \Omega(n)\,.
\end{equation}

Hence, combining this result with Theorem~\ref{thm:main1} allows us to argue that $S_n$-equivariant QNNs do not induce barren plateaus.
\begin{corollary}\label{cor:main1}
    Under the same assumptions as Theorem~\ref{thm:main1}, it follows that, if $\Delta(\sum_{\nu=1}^{m_{\lambda}}\sigma_\lm^\nu)\in\Omega\left(1/\poly(n)\right)$, then the empirical loss function satisfies
    \begin{equation}
        \Var_{\thv}[\partial_\mu \hat{\LC}]\in\Omega\left(\frac{1}{\poly(n)}\right)\,.
    \end{equation}
\end{corollary}
We note that a crucial requirement for Corollary~\ref{cor:main1} to hold is that  $\Delta(\sum_{\nu}\sigma_\lm^\nu)$ needs to be, at most, polynomially vanishing. 
In Sec.~\ref{sec:trainable}, we identify cases of datasets leading to trainability but also to untrainability. Finally, we note that as discussed in Supplementary Methods 9,  Corollary~\ref{cor:main1} is sufficient to guarantee that the loss function does not exhibit the narrow gorge phenomenon, whereby the minima of the loss occupy an exponentially small volume of parameter space~\cite{arrasmith2021equivalence}. In other words, we show that absence of barren plateau implies absence of narrow gorges and loss function anti-concentration.

\subsection*{Efficient overparametrization}

Absence of barren plateaus is a necessary, but not sufficient, condition for trainability, as there could be other issues compromising the parameter optimization. In particular, it has been shown that quantum landscapes can exhibit a large number of local minima~\cite{bittel2021training,anschuetz2022beyond,fontana2022nontrivial}. As such, here we consider a different aspect of the trainability of $S_n$-equivariant QNNs: their ability to converge to global minima. For this purpose, we find it convenient to recall the concept of overparametrization.

Overparametrization denotes a regime in machine learning where models have a capacity much larger than that necessary to represent the distribution of the training data. For example, when the number of parameters is greater than the number of training points. Models operating in the overparametrized regime have seen tremendous success in classical deep learning, as they closely fit the training data but still generalize well when presented with new data instances \cite{zhang2021understanding,allen2019convergence,allen2019learning,buhai2020empirical}. Recently, Ref.~\cite{larocca2021theory} studied  overparametrization in the context of QML models. A clear phase transition in the trainability of under- and overparametrized QNNs was evidenced: Below some critical number of parameters (underparametrized) the optimizer greatly struggles to minimize the loss function, whereas beyond that number of parameters (overparametrized) it converges exponentially fast to solutions (see Methods for further details).

Given the desirable features of overparametrization, it is  important to estimate how many parameters are needed to achieve this regime. Here,  we can derive the following theorem.
\begin{theorem}\label{theo:overparam}
    Let $\UC_{\thv}$ be a $S_n$-equivariant QNN with generators in $\GC$. Then, $\UC_{\thv}$ can be overparametrized with $\OC(n^3)$ parameters.
\end{theorem}
Theorem~\ref{theo:overparam} guarantees that $S_n$-equivariant QNNs only require a polynomial number of parameters to reach overparametrization.

\subsection*{Generalization from few data points}

Thus far we have seen that $S_n$-equivariant QNNs can be efficiently trained, as they exhibit no barren plateaus and can be overparametrized. However, in QML we are not only interested in achieving a small training error, we also aim at low generalization error~\cite{caro2021generalization,banchi2021generalization,caro2022outofdistribution,du2021efficient,huang2021power,abbas2020power}.

Computing the generalization error in Eq.~\eqref{eq:gen-error} is usually not possible, as the probability distribution $P$ over which the data is sampled is generally unknown. However, one can still derive bounds for ${\rm gen}(\thv)$ which guarantee a certain performance when the model sees new data. Here, we obtain an upper bound for the generalization error via the covering numbers (see Methods)~\cite{shalev2014understanding, caro2021generalization}, and prove that the following theorem holds. 
\begin{theorem}\label{thm:gen}
    Consider a QML problem with loss function as described in Eq. \eqref{eq:gen-error}. Suppose that an $n$-qubit $S_n$-equivariant QNN $\mathcal{U}(\thv)$ is trained on $M$ samples to obtain some trained parameters $\thv^*$. Then the following inequality holds with probability at least $1-\delta$
    \begin{equation}\label{eq:gen_bounds}
        {\rm gen}(\thv^*) \leq \mathcal{O}\left(\sqrt{\frac{\textnormal{Te}_{n+1}}{M}} + \sqrt{\frac{\log(1/\delta)}{M}} \right).
    \end{equation}
\end{theorem}
The crucial implication of Theorem~\ref{thm:gen} is that  we can guarantee ${\rm gen}(\thv^*) \leq \epsilon$ with high probability, if $M\in\mathcal{O}\left(\frac{\text{Te}_{n+1}+\log(1/\delta)}{\epsilon^2}\right)$. For fixed $\delta$ and $\epsilon$, this implies $M\in\mathcal{O}(n^3)$, i.e., we only need a  polynomial number of training points.
Also note that this results shows that minimizing the empirical loss closely minimizes the true loss with high probability. Say that $\hat{\LC}^* = \inf_{\thv} \hat{\LC}(\thv)$ is the minimal empirical loss and $\LC^* = \inf_{\thv}\LC(\thv)$ the minimal true loss. Then, with $M\in\mathcal{O}\left(\frac{\text{Te}_{n+1}+\log(1/\delta)}{\epsilon^2}\right)$ training data point the inequality $\lvert \hat{\LC}^* - \LC^* \vert \leq \epsilon$ holds with probability at least $1-\delta$.

Lastly, we remark that Theorem~\ref{thm:gen} can be readily adapted to other GQML models. As shown in Methods, this theorem stems from the fact that the equivariant unitary submanifold, in its block-diagonal form in Eq. \eqref{eq:commutator}, can be covered \cite{shalev2014understanding} by $\varepsilon$-balls in a block-wise manner. In Supplementary Methods 8, we also show that the VC dimension \cite{Hajek2021Statistical} of equivariant QNNs (and also more general parameterized channels) can be upper bounded by the dimension of the commutant of the symmetry group, a fact which could be of independent interest.

\begin{table}[t]
    \centering

    \begin{tabular}{|l||*{2}{c|}}\hline
    Input state&{Trainable?}&{Method}\\\hline\hline
    Symmetric   & Yes & Analytical\\\hline
    Fixed Hamming-weight encoding  & Yes & Analytical\\\hline
    Local Haar random   & Yes & Numerical\\\hline
    Fixed depth random circuit & Yes & Numerical \\\hline
    Disconnected graph state   & Yes & Numerical\\\hline
    $3$-regular graph state   & Yes & Numerical\\\hline
    $n/2$-regular graph state   & Yes & Numerical\\\hline
    Global Haar Random  & No & Analytical\\\hline
    Linear depth random circuit & No & Numerical \\\hline
    Erd\"{o}s–R\'{e}nyi random graph state& No & Numerical\\\hline
    \end{tabular}

    \caption{\textbf{Input pure states and their effect on the trainability of $S_n$-equivariant QNNs.} Trainable means that $\Delta(\sum_{\nu}\sigma_\lm^\nu)\in\Omega(1/\poly(n))$, whereas untrainable means $\Delta(\sum_{\nu}\sigma_\lm^\nu)\in\OC(1/2^n)$. Analytical method indicates that we can exactly compute the scaling of $\Delta(\sum_{\nu}\sigma_\lm^\nu)$, whereas numerical one means that we evaluate it numerically. The analytical proofs and details of the simulations can be found in Supplementary Methods 7. We note that, these results are obtained by computing the loss with a single data instance  (i.e., for $M=1$ in Eq.~\eqref{eq:empirical-loss}). }
    \label{tab:train_states}
\end{table}

\subsection*{Trainable States}\label{sec:trainable}

As discussed in the previous section, $S_n$-equivariant QNNs and measurement operators cannot induce barren plateaus. Thus, the trainability of the model hinges on the behavior of  $\Delta(\sum_{\nu}\sigma_\lm^\nu)$. We note that this dataset-dependent trainability it not unique to $S_n$-equivariant QNNs, but is rather present in all absence of barren plateaus results (see Refs.~\cite{cerezo2020cost,pesah2020absence,larocca2021diagnosing,thanasilp2021subtleties,liu2021presence,thanasilp2022exponential}) as there always exist datasets for which an otherwise trainable model can be rendered untrainable. 

To understand the conditions that lead to an exponentially vanishing of $\Delta(\sum_{\nu}\sigma_\lm^\nu)$ we note that for a Hermitian operator $B$, we have $\Delta(B) = D_{\textnormal{HS}}\left(B,\frac{\Tr[B]}{\dim(B)}\id\right)$, where $D_{\textnormal{HS}}(A,B)=\norm{A-B}_2^2$ is the Hilbert-Schmidt\ distance. Alternatively, we can interpret $\Delta(B)$ as the variance of the eigenvalues of $B$. From here, we can see that one will obtain trainability if at least one $\sigma_\lambda$ is not exponentially close to a multiple of  the identity in some subspace $\HC_\lambda^\nu$.

\begin{figure*}[t]
    \centering
    \includegraphics[width=1\linewidth]{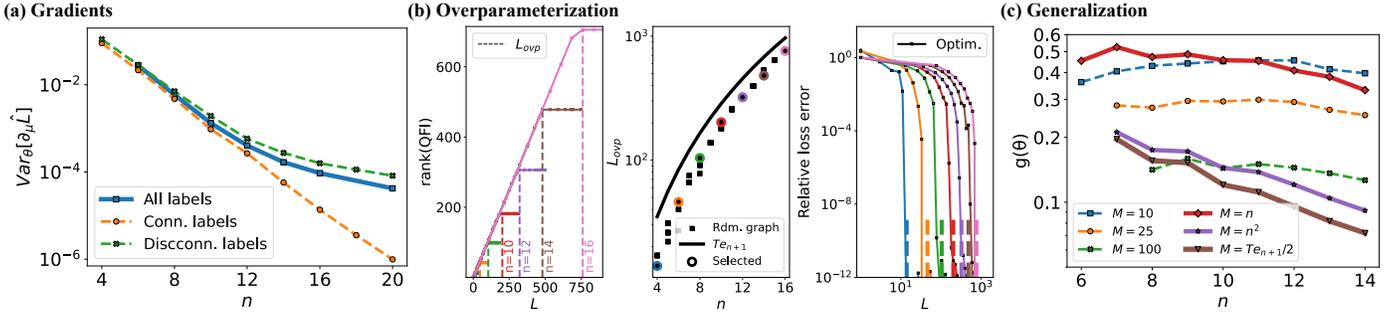}
    \caption{
    \textbf{Task of distinguishing connected from disconnect graphs with an $S_n$-equivariant QNN.}
    a) Variance of the loss function partial derivatives versus the number of qubits $n$ (in log-linear scale). 
    The square blue line depicts the variance for inputs of the QNN drawn from a dataset composed of connected and disconnected graph states. 
    To visualize how the data with different labels contributes to this variance, we also plot in green crosses (orange circles) the variances when the QNN is only fed connected (disconnected) graph states. 
    b) In the left panel, we show representative results for the rank of the QFIM (defined in the main text) versus the number of layers $L$ for different number of qubits $n$. 
    The critical value of layers at which this rank saturates, denoted $L_{ovp}$ (vertical dashed lines), corresponds to the onset of overparametrization.
    In the middle panel, we report the scaling of $L_{\text{ovp}}$ versus the number of qubits (log-linear scale). For each problem size, we present results for $10$ random input graph states and, as a comparison, also report the Tetrahedral numbers ${\rm Te}_{n+1}$ (solid line). 
    In the right panel, we report the relative loss error of optimized QNNs at given number of layers $L$ (in log-linear scale). These are obtained for different system sizes, with 
    the dashed vertical lines indicating the corresponding values of $L_{\text{ovp}}$. 
    c) Normalized generalization error versus number of qubits $n$ (in log-linear scale) for different training dataset sizes $M$. Here, we consider an overparametrized QNN with $L={\rm Te}_{n+1}$. 
    }
    \label{fig:numerics}
\end{figure*}

In Table~\ref{tab:train_states} we present examples of states for which $\Delta(\sum_{\nu}\sigma_\lm^\nu)$ vanishes polynomially, leading to a trainable model, but also cases where the input state leads to exponentially vanishing $\Delta(\sum_{\nu}\sigma_\lm^\nu)$ and thus to a barren plateau. While we leave the details of how each type of input state is generated for the Methods section, we note that the results in Table~\ref{tab:train_states} demonstrate  the critical role that the input states play in determining the trainability of a model (this will be further elucidated in numerical results below). Such insight is particularly important as one can create adversarial datasets yielding barren plateaus (see Supplementary Methods 10). Moreover, it indicates that care must be taken when  encoding classical data into quantum states as the embedding scheme can induce trainability issues~\cite{thanasilp2021subtleties,thanasilp2022exponential}.

\subsection*{Numerical results}

Here, we consider the task of classifying connected graph states from  disconnected graph states, which are prepared as follows. First, we generate $n$-node random graphs from the Erd\"{o}s–R\'{e}nyi distribution~\cite{erdds1959random}, with an edge probability of $40\%$. The ensuing graphs are binned into two categories: connected and disconnected. 
We then embed these graphs into quantum graph states via the canonical  scheme of~\cite{raussendorf2003measurement, hein2004multiparty} (see Methods section).  
 We highlight that such encoding preserves symmetries in the input data, in the sense that a permutation of the underlying graph yields a permutation of the qubits constituting its graph state (i.e., of the form Eq.~(\ref{eq:qubit_permut})).
The previous allows us to create a dataset where half of the states encodes connected graphs (label $y_i=+1$), and the other half encodes disconnected graphs (label $y_i=-1$). 
To analyze the data, we use an $S_n$-equivariant QNN with generators in Eq.~\eqref{eq:generators-main} (see also Fig.~\ref{fig:fig-ansatz}), and measure the operator 
$O=\frac{2}{n(n-1)}\sum_{k<j;j=1}^n X_jX_k$. 

In the following, we characterize the trainability and generalization properties of $S_n$-equivariant QNNs for this classification task, but we note that further aspects of the problem are discussed in the Supplementary Note. These include analyzing the effect of the graph encoding scheme in the trainability, the irrep contributions to the gradient variance, and comparing $S_n$-equivariant QNNs against problem-agnostic ones. In particular,  the latter shows that for the present graph classification task, problem-agnostic models are hard to train and tend to greatly overfit the data, i.e., they  have large generalization errors despite performing well on the training data.

\subsection*{Numerics on barren plateaus}

In Fig.~\ref{fig:numerics}(a) we show the variance of the cost function partial derivatives for a parameter $\theta_\mu$ in the middle of the QNN. 
Each point is evaluated for a total of $50$ random input states, and with $20$ random sets of parameters $\thv$ per input. We can see that when the variance is evaluated for states randomly drawn from the whole dataset -- with an equal number of connected and disconnected graphs -- then $\Var_{\thv}[\partial_\mu \hat{\LC}]$ only decreases  polynomially with the system size (as evidenced by the curved line in the log-linear scale), meaning that the model does not exhibit a barren plateau. We note that, as shown in  Fig.~\ref{fig:numerics}(a), when the input to the QNN is a disconnected graph state, then the variance vanishes polynomially, whereas if we input a connected graph state it vanishes exponentially. This illustrates a key fact of QML: when trained over a dataset, the data from different classes can contribute very differently to the model's trainability (see~\cite{larocca2022group} for a discussion on how this result enables new forms of classification).   

\subsection*{Numerics on overparametrization}

Following the results in Ref.~\cite{larocca2021theory}, let us analyze the overparametrization phenomenon by studying the rank of the quantum fisher information matrix (QFIM)~\cite{cheng2010quantum,meyer2021fisher}, denoted $F(\thv)$ and whose entries are given by 
\small
\begin{equation}
    [F(\thv)]_{jk} = 4\Re[ \bra{\partial_j \psi(\thv)} \ket{\partial_k \psi(\thv)} -\bra{\partial_j \psi(\thv)}  \ket{\psi(\thv)} \bra{\psi(\thv)} \ket{\partial_k \psi(\thv)}],
\nonumber
\end{equation}
\normalsize
with $\ket{\psi(\thv)} = U(\thv) \ket{\psi}$, and $\ket{\partial_i\psi(\thv)}=\partial \ket{\psi(\thv)}/\partial\theta_i=\partial_i\ket{\psi(\thv)}$ for $\theta_i\in\thv$. The rank of the QFIM quantifies the number of potentially accessible directions in state space. In this sense, the model is overparametrized if the QFIM rank is saturated, i.e., if adding more parameters (or layers) to the QNN does not further increase the QFIM rank. When this occurs, one can access all possible directions in state space and efficiently reach the solution manifold~\cite{larocca2021theory,larocca2020exploiting,larocca2020fourier}. On the other hand, the model is underparametrized if the QFIM rank is not maximal. In this case, there exists inaccessible directions in state space, leading to false local minima, that is, local minima that are not actual minima of the loss function.

In Fig.~\ref{fig:numerics}(b, left panel) we report representative results of the QFIM rank versus the number of layers $L$ for problems with even numbers $n\in[4,16]$ of qubits. These results correspond to random connected graphs and random values of $\thv$. Here we can see that, for a given $n$, as the number of layers increases, the rank of the QFIM also increases until it reaches a saturation point. 
Once this critical number of layers (denoted as $L_{\text{ovp}}$) is reached, the model is considered to be overparametrized~\cite{larocca2021theory}. In Fig.~\ref{fig:numerics}(b, middle panel) we plot the scaling of $L_{\text{ovp}}$ (for $10$ random connected or disconnected graphs per system size) versus $n$, as well as the Tetrahedral numbers ${\rm Te}_{n+1}$. 
As can be seen,  in all cases, the overparametrization onset occurs for a number of layers $L_{\text{ovp}}<{\rm Te}_{n+1}$, indicating efficient overparametrization.

To appreciate the practical effects of overparametrization, we report in Fig.~\ref{fig:numerics}(b, right panel) optimization performances of $S_n$-equivariant QNNs as a function of the number $L$ of layers employed.  
All the optimizations are performed using the hinge loss function, with the L-BFGS-B optimization algorithm~\cite{zhu1997algorithm}. The system sizes are in $n\in[4,16]$ qubits, and correspond to the graphs that were studied in the left panel and highlighted in the middle one. 
The relative loss error reported indicates how close an optimized QNN is from the best achievable model. Explicitly, it is defined as $|\hat{\LC}_L-\hat{\LC}_{\textnormal{min}}|/|\hat{\LC}_{\textnormal{min}}|$, where $\hat{\LC}_L$ is the loss achieved after optimization of a QNN with a given $L$, and where $\hat{\LC}_{\textnormal{min}}$ is the minimum loss achieved for any of the values $L$ considered, i.e., $\hat{\LC}_{\textnormal{min}}=\argmin_L \hat{\LC}_L$ (we systematically verify that for sufficient large $L$ all optimizations reliably converge to this same loss $\hat{\LC}_{\textnormal{min}}$). 
For every value of $n$ studied, we see that for a small number of layers the optimizer struggles to significantly minimize the loss. 
However, as $L$ increases, there exists a computational phase transition whereby the optimizer is able to easily identify optimal parameters and reach much smaller loss values. Notably, such computational phase transition occurs slightly before $L_{\text{ovp}}$ (indicated by a dashed vertical line), meaning that even before the QFIM rank saturates, the model has sufficient directions to efficiently reach the solution manifold. Overall, we see that for number of layers growing at most polynomially with $n$, one can ensure convergence to solution of the model.

\subsection*{Numerics on generalization error}

 In Fig.~\ref{fig:numerics}(c) we study the generalization error of an overparametrized $S_n$-equivariant QNN (with $L={\rm Te}_{n+1}$) for different training dataset sizes $M$  and with respect to test sets of size $M_{test}=2\times{\rm Te}_{n+1}$ that are independently drawn from the training ones. Generalization errors are evaluated for random QNNs parameters $\thv$ and we report the $90$-th percentile of the errors obtained, i.e., for $\delta=90\%$ in Eq.~\eqref{eq:gen_bounds}. 
 In the plot, we show the normalized generalization error $g(\thv)=\frac{\text{gen}(\thv)}{\Var^{1/2}_{\thv,\rho}[\ell(\thv,\rho)]}$. 
 We stress that such normalization can only increase the generalization errors obtained, and is only used in order to compare generalization errors across different values of $n$ without artifacts resulting from loss concentration effects as the system sizes grow.  As seen in  Fig.~\ref{fig:numerics}(c), when the size of the training set is constant, the generalization error is also approximately constant across problem sizes. However, when the training set size scales with $n$, the generalization error decreases with $n$, with this even occurring for $M=n$. Notably, if $M={\rm Te}_{n+1}\in\Theta(n^3)$, we can see that the generalization error significantly decreases with problem size. That is, for this problem, we found generalization errors to be better than the scaling of the bounds derived in Eq.~\eqref{eq:gen_bounds}.

\section*{Discussion}

GQML has recently been proposed as a framework  for systematically creating models with sharp geometric priors arising from the symmetries of the task at hand~\cite{larocca2022group,meyer2022exploiting, sauvage2022building, zheng2022super,zheng2021speeding}. Despite its great promise,  this nascent field has only seen  heuristic success as no true performance guarantees have been proved for its models.   In this work we provide the first theoretical guarantees for GQML models aimed at problems with permutation invariance. Our first contribution is the introduction of the $S_n$-equivariant QNN architecture. Using tools from representation theory, we rigorously find that these QNNs present salient features such as absence of barren plateaus (and narrow gorges), generalization from very few data points, and a capability of being efficiently overparametrized. All these favorable properties can be viewed as being direct consequences of the inductive biases embedded in the model, which greatly limits their expressibility~\cite{sim2019expressibility,nakaji2020expressibility,holmes2021connecting}. Namely, these $S_n$-equivariant QNNs act only on the --polynomially large-- multiplicity spaces of the qubit-defining representation of $S_n$. To complete our analysis, we performed numerical simulations for a graph classification task and heuristically found that the model's performance is even better than that predicted by our theoretical results.

Taken together, our results provide the first rigorous guarantees for equivariant QNNs, and   demonstrate that GQML may be a powerful tool in the QML repertoire. 
We highlight that while we focus on problems with $S_n$ symmetry, many of our proof techniques hold for general finite-dimensional compact groups. Hence, we hope that the representation-theory-based techniques used here can serve as blueprints to analyze the performance of other models. We envision that in the near future, GQML models with provable guarantees will be widely spread among the QML literature.

Finally, we note that while our results were derived in the absence of noise, it would be interesting to account for hardware imperfections. Clearly, the presence of noise would change our analysis, and most likely weaken our trainability guarantees. As such, while we can guarantee that $S_n$-equivariant QNNs will be useful on fault-tolerant quantum devices, we do not abandon hope that they can be used in the near-term era provided that noise levels are small enough.

\medskip

Note added: In light of the  recent preprint~\cite{anschuetz2022efficient}, we have added a detailed discussion in the Supplementary Note regarding the possibility of classically simulating $S_n$-equivariant QNNs. As we argue there, for most relevant cases in QML, the  algorithm in~\cite{anschuetz2022efficient} is not fully classical, as it require access to a quantum computer to obtain a ``classical description'' of the input data.   Moreover, even if one is given such ``classical description'', the ensuing algorithm that replaces the use of a QNN scales extremely poorly with the number of qubits. Taken together these results indicate that if one has access to a quantum computer, it is not entirely obvious whether one should use it to obtain a classical description of the data followed by expensive post-processing, or if one should run the QNN on the quantum device and exploit its favorable properties like efficient overparametrization and absence of barren plateaus. We will save such comparison for future work.


Now we will briefly compare $S_n$-equivariant QNNs to other barren-plateau-avoiding architectures.

First, let us consider the shallow hardware efficient ansatz (HEA)~\cite{kandala2017hardware,cerezo2020cost} and the quantum convolutional neural network (QCNN)~\cite{cong2019quantum,pesah2020absence}. While our goal is not to provide a comprehensive description of these models, we recall the three key properties leading to their trainability: locality of the gates, shallowness of the circuit, locality of the measurement operator.  Both the HEA and QCNN are composed of parametrized gates acting in a brick-like fashion on alternating pairs of neighboring qubits (local gates), and are composed of only a few --logarithmically many-- layers of such gates (shallowness of the circuit). The combination of these two factors leads to a low scrambling power and greatly limited expressibility of the QNN. Then, the final ingredient for their trainability requires measuring a local operators (i.e., an operator acting non-trivially on a small number of qubits). While this assumption is guaranteed for QCNNs --due to their feature-space reduction property--, the HEA can be shown to be untrainable for global measurement (i.e., operators acting non-trivially on all qubits). Here we can already see that $S_n$-equivariant QNNs do not share the properties leading to trainability in HEAs and QCNNs. To begin, we can see from the set of generators $\GC$ in Eq.~\eqref{eq:generators-main} that the $S_n$-equivariant architecture allows for all long-range interactions in each layer, breaking the locality of gates assumption. Moreover, and in stark contrast to HEAs, one can train  the $S_n$-equivariant QNN even when measuring global observables (for instance, we allow for the $O=\prod_{j=1}^n\X_j$ in Eq.~\eqref{eq:measurement-op}). Finally, we remark that HEAs and QCNNs cannot be efficiently overparametrized, as they require an exponentially large number of parameters to reach overparametrization~\cite{larocca2021diagnosing}. On the other hand, according to Theorem~\ref{theo:overparam} the $S_n$-equivariant QNN can be overparametrized with polynomially many  layers.

Next, let us consider the transverse field Ising model Hamiltonian variational ansatz (TFIM-HVA)~\cite{wiersema2020exploring,larocca2021diagnosing}. The mechanism leading to absence of barren plateaus in this architectures is more closely related to that of the   $S_n$-equivariant model, although there are still some crucial differences. On the one hand, it can be shown that the  TFIM-HVA has an extremely limited expressibility, having  only a maximum number of free parameters in $\OC(n^2)$, and being able to reach overparametrization with polynomially many layers.  While this is similar to the case of $S_n$-equivariant architectures (see Lemma~\ref{lem:poly} and Theorem~\ref{theo:overparam}),  the block diagonal structure of the TFIM-HVA is fundamentally different than that arising from $S_n$-equivariant: The TFIM-HVA unitary has four exponentially large blocks repeated a single time each, while $S_n$-equivariant unitaries have polynomially small blocks repeated exponentially many times. This subtle, albeit important, distinction makes it such that $S_n$-equivariant QNNs enjoy generalization guarantees (from Theorem~\ref{thm:gen}) which are not directly applicable to TFIM-HVA architectures.

The previous shows that $S_n$-equivariant QNNs stand-out amid the other trainable architectures, exhibit many favorable  properties that other models only partially enjoy.

Lastly, we now consider future directions and possible extensions of our work. We recall that  Definition~\ref{def:equivariant_QNN} requires \textit{every} layer of the QNN to be equivariant. This is evidently not general, as one could have several consecutive layers which are not individually equivariant, but compose to an equivariant unitary for certain $\thv$~\cite{larocca2022group,cincio2018learning}. While in this manuscript we do not consider this scenario, it is worth exploring how less strict equivariance conditions affect the performance and the trainability guarantees here derived. Second, we note that as indicated in this work, the block diagonal structure of the $S_n$-equivariant QNN restricts the information in the input data that the model can access. This could lead to conditions where the model cannot solve the learning task as it cannot `see' the relevant information in the input states. Such issue can be in principle solved by allowing the model to simultaneously act on multiple copies of the data, and even to change the representation of $S_n$ throughout the circuit~\cite{nguyen2022atheory}. We also leave this exploration for future work. 

Another potentially interesting research direction would be equivariant embeddings and re-uploading of classical data. For the purposes of this work, we make no assumptions to the source or form of the data, such as whether it is quantum or classical. However, when considering analyzing classical data on quantum computer, embeddings become important. We give one such example, which we call a "fixed Hamming-weight encoding". Another example is the standard encoding of a graph into a graph state, which we considered in our numerics. This is far from exhaustive and more sophisticated methods exist, including trainable encoding~\cite{skolik2022equivariant}. Similarly, we have not studied how our results change in the presence of data re-uploading~\cite{perez2020data}. We know that if the data is re-uploaded via  equivariant generators  (e.g., if the data re-uploading unitary takes the form $V(\vec{x})=\prod_{l'}\er^{-\ir x_l H_l}$, with $H_l$ being $S_n$-equivariant), then our theoretical guarantees results do not change. This follows from the fact that the DLA of the circuit will remain the same, and hence our results follow. We leave the study of more general encoding and re-uploading schemes for future work.

\section*{Methods}
This section provides an overview of the different tools used in the main text. Here we also present a sketch of the proof of our main results. Full details can be found in the Supplementary Methods.

\subsection*{Building $S_n$-equivariant operators}

Here we briefly describe how to build $S_n$-equivariant operators that can be used as generators of the QNN, or as measurement operators. In particular, we will focus on the so-called twirling method~\cite{meyer2022exploiting,nguyen2022atheory}. Take a unitary representation $R$ of a discrete group $G$ over a vector space $V$.
Then the twirl operator is the linear map $\TC_{G}: GL(V)\to GL(V)$, defined as
\begin{equation}
    \TC_G (A) = \frac{1}{\abs{G}} \sum_{g\in G} R(g) A R(g)\ad  \,.
\end{equation}
It can be readily verified that the twirling of any operator $A$ yields a $G$-equivariant operator, i.e., we have $[\TC_G (A),R(g)]=0$ for any $g\in G$.

The previous allows us to obtain a $G$-equivariant operator from any operator $A\in GL(V)$. For instance, let us consider the case in the case of $G=S_n$, $R$ the qubit-defining representation and $A=X_1$. Then, we have $\TC_G(X_1)=\frac{1}{n!} \sum_{\pi\in S_n} R(\pi) X_1 R(\pi)\ad =\frac{1}{n}\sum_{i=1}^n X_i=\TC_G(X_j)$ for any $1\leq j\leq n$. Note that twirling over $S_n$ cannot change the locality of an operator. That is, twirling a $k$-body operator leads to a sum of $k$-body operators.

\subsection*{Representation theory of $S_n$}

In this section we review a few basic notions from representation theory. For a more thorough treatment we refer the reader to Refs.~\cite{serre1977linear,fulton1991representation,sagan2001symmetric,knapp2001representation}, and more specifically to the tutorial in Ref.~\cite{ragone2022representation} which provides an introduction to representation theory from the perspective of QML.
We recall that we are interested in the  qubit-defining representation of $S_n$, i.e., the one permuting qubits
\begin{equation}
R(\pi\in S_n)\bigotimes_{i=1}^n \ket{\psi_i} = \bigotimes_{i=1}^n \ket{\psi_{\pi^{-1}(i)}}\,.\nonumber
\end{equation}
As mentioned in the main text, representations break down into fundamental building blocks called irreducible representations (irreps).
\begin{definition}[Irrep decomposition]
    Given some unitary representation $R$ of a compact group $G$, there exists a basis under which it takes a block diagonal form
    \begin{equation}\label{eq_meth_irrep_dec}
    R(g\in G)\cong\bigoplus_{\lambda}\bigoplus_{\mu=1}^{m_{r_\lambda}}r_\lambda(\pi)=\bigoplus_{\lambda}r_\lambda(\pi)\otimes\id_{m_{r_\lambda}}\,,
\end{equation}
with $r_\lambda(\pi)$ irreps of $G$ appearing $m_{r_\lambda}$ times.
\end{definition}
The irreps of the symmetric group are commonly labeled by the set of partitions of the integer $n$. A partition of a positive integer $n\in\mbb{N}$ is a non-decreasing sequence of positive integers $\lm=(\lm_1,\cdots,\lm_k)$ satisfying $\sum_i \lm_i = n$. Partitions are typically visualized using young diagrams, a set of empty, left-justified boxes arranged in rows such that there are $\lm_i$ boxes in the $i$-th row. For instance, the integer $n=3$ can split into

\begin{equation}\label{eq:method_partition}
\footnotesize
(3,0)=
\ytableausetup{centertableaux}
\begin{ytableau}
&  &  \\
\end{ytableau}\,, \quad
(2,1)=
\ytableausetup{centertableaux}
\begin{ytableau}
& \\
\\
\end{ytableau}\,, \quad
(1,1,1)=
\ytableausetup{centertableaux}
\begin{ytableau}
\\
\\
\\
\end{ytableau} \,.
\end{equation}
\normalsize
We note that in the case of the qubit-defining representation, the only $\lm$ appearing  in Eq.~\eqref{eq_meth_irrep_dec} have at most two rows (e.g., would not include the last partition in Eq.~\eqref{eq:method_partition}). 

The dimension of an $S_n$ irrep $r_\lm$ can be computed from the hook length formula
\begin{equation}\label{dim_rlm}
    \dim(r_\lm) = \frac{n!}{\prod_{b\in\lm} h_\lm(b)}\,,
\end{equation}
where each $h_\lm(b)$ is the hook length for box $b$ in $\lm$, which is the total number of boxes in a 'hook' (or 'l' shape) composed of box $b$ and every box beneath (in the same column) and to its right (in the same row). 

Given the block-diagonal structure of $R$ in Eq.~\eqref{eq_meth_irrep_dec}, one can see that a general $G$-equivariant operator has to be of the form
\begin{equation}\label{eq:commutator-met-SMh}
    A \cong \bigoplus_{\lm} \id_{\dim(r_\lm)} \otimes A_\lm,
\end{equation}
where $A_\lambda$ are $m_{r_\lm}$-dimensional matrices repeated $\dim(r_\lm)$ times. In general, the number of times an irrep appears in an arbitrary representation $R$ (i.e., $m_{r_\lambda}$ in Eq.~\eqref{eq_meth_irrep_dec}) can be determined through character theory. Instead, in our case, we will take a shortcut and exploit one of the most remarkable results in representation theory, called the Schur-Weyl duality~\cite{goodman2009symmetry}.

Consider the representation $Q$ of the unitary group $\mathbb{U}(2)$ acting on $\HC=(\mbb{C}^2)\tn$ through the $n$-fold tensor product $Q(W\in\mbb{U}(2))=W^{\otimes n}$. Evidently, according to Eq.~\eqref{eq_meth_irrep_dec}, $Q$ will also have an isotypic decomposition
\begin{equation}
    Q(W\in\mbb{U}(2))= \bigoplus_{s} \id_{m_{q_{s}}} \otimes q_{s}(W)\,,
\end{equation}
where $s$ labels the different (spin) irreps of $\mbb{U}(2)$. The Schur-Weyl duality, states that the matrix algebras $\mbb{C}[R]$ and $\mbb{C}[Q]$ mutually centralize each other, meaning that $\mbb{C}[R]$ is the space of $\mbb{U}(2)$-equivariant linear operators, and similarly $\mbb{C}[Q]$ is the space of $S_n$-equivariant ones. As a consequence of this duality, $\HC$ can be decomposed as $\mathcal{H} \cong \bigoplus_\lambda V_\lambda \otimes W_\lambda$, where $\lambda$ simultaneously labels irrep spaces $V_\lambda$ and $W_\lambda$ for $S_n$ and $\mbb{U}(2)$, respectively. That is, $\HC$ supports a simultaneous action of $S_n$ and $\mbb{U}(2)$, where the irreps of each appear exactly once and are correlated: Each of the two-row Young diagrams $\lm=(n-m,m)$ labeling the irreps in $R$ can be associated unequivocally with a spin label $s(\lm)$ for an $\mbb{U}(2)$ irrep appearing in $Q$
\begin{equation}
    s(\lm)=\frac{\lm_1-\lm_2}{2} =\frac{n-2m}{2}\,.
\end{equation}
Moreover, since under the joint action of $S_n\times \mbb{U}(2)$ the multiplicities are one, one can assert that the irrep $q_\lm$ of $\mbb{U}(2)$ appears $\dim(r_\lm)$-times in Q, and conversely, the irrep $r_\lm$ of $S_n$ appears $\dim(q_\lm)$-times in R. Using the well known dimension of spin irreps $\dim(q_{s})=2s+1$, we can derive an expression for the multiplicity of $S_n$ irreps
\begin{equation}
    m_{r_\lm} = \dim(q_{s(\lm)}) = 2 s(\lm) +1 = n-2m+1\,.
\end{equation}
Also, it is straightforwards to adapt the formula in Eq.~\eqref{dim_rlm} to two-row diagrams $\lm=(n-m,m)$
\begin{equation}\label{eq_mlm}
    \dim(r_\lm)=\frac{n!(n-2m+1)!}{(n-m+1)!m!(n-2m)!}.
\end{equation}
We finally note that, since we are ultimately interested in $S_n$-equivariant operators, in the main text we have defined $d_\lm\equiv m_{r_\lm}$ and $m_\lm\equiv\dim(r_\lm)$. That is, the dimension and multiplicity of an irrep in the main text are for the representations of $\mbb{U}$.

\subsection*{Universality, expressibility and dynamical Lie algebra}

In the main text we have argued that the set of generators in Eq.~\eqref{eq:generators-main} is universal within each invariant subspace. Here we will formalize this statement. 

First, let us recall that we say that a parametrized unitary is universal if it can generate any unitary (up to a global phase) in the space over which it acts. One can quantify the  capacity of being able to create different unitaries through the so-called measures of expressibility~\cite{sim2019expressibility,nakaji2020expressibility,holmes2021connecting,larocca2021diagnosing}. Here we will focus on the notion of potential expressibility of a given QNN, which is formalized via the dynamical Lie algebra of the architecture~\cite{zeier2011symmetry}.
\begin{definition}[Dynamical Lie algebra]\label{def_DLA}
    Given a set of generators $\mathcal{G}$ defining a QNN, its dynamical Lie algebra $\mathfrak{g}$ is the span of the Lie closure $\langle \cdot \rangle_{\rm{Lie}}$ of $\mathcal{G}$. 
    That is, $\mathfrak{g} = \spn_{\mbb{R}} \langle \GC \rangle_{\rm{Lie}}$, where $\langle \GC \rangle_{\rm{Lie}}$ is defined as the set of all the nested commutators generated by the elements of $\GC$.
\end{definition}
In particular, the dynamical Lie algebra (DLA) fully characterizes the group of unitaries that can be ultimately expressed by the circuit: for any unitary $U$ realized by a QNN with generators in $\GC$ there exists an anti-hermitian operator $\eta\in\mf{g}=\langle \GC \rangle_{\rm{Lie}}$ such that $U=\er^\eta$. Evidently, $\mf{g}\subseteq \mf{u}(d)$, that is, it is a subalgebra of the space of anti-hermitian operators. When  $\mf{g}$ is $\mf{su}(d)$ or $\mf{u}(d)$  we say that the QNN is controllable or universal since for any pair of states $\ket{\psi}$ and $\ket{\phi}$, there exists a unitary $U=\er^\eta$ with $\eta\in\mf{g}$ such that  $|\langle\phi|U|\psi\rangle|^2=1$.

In the framework of GQML one designs symmetry-respecting QNNs by using group-equivariant generators. This implies that the corresponding DLA is constrained and necessarily takes the form
\begin{equation}
    \mf{g}=\bigoplus_\lm\id_{m_\lm}\otimes \mf{g}_\lm\,,
\end{equation}
where $\mf{g}_\lm\subseteq \mf{u}(d_\lm)$. 
For this scenario, we provide a notion of controllability restricted to each of the invariant subspaces: 
We say that a QNN is subspace-controllable in the isotypic component $\lm$ if $\mf{g}_\lm$ is $\mf{su}(d_\lm)$ or $\mf{u}(d_\lm)$. This means that the QNN can map between any pair of states in every $\HC_\lm^\nu$. Notably, the following result follows from Refs.~\cite{albertini2018controllability,kazi2023universality}. 
\begin{lemma}[Subspace controllability]\label{lem:control}
The set of $S_n$-equivariant generators in Eq.~\eqref{eq:generators-main} is subspace-controllable in every $\lm$. 
\end{lemma}
As shown below, this result will be crucial for the proof of Theorem~\ref{thm:main1}.

\subsection*{Proof of absence of barren plateaus}

Here we sketch our proof of Theorem~\ref{thm:main1}. Our goal is to calculate $\Var_{\thv}[\partial_{\mu}\widehat{\LC}(\thv)]=\mathbb{E}_{\thv}[(\partial_{\mu}\widehat{\LC}(\thv))^2]-\mathbb{E}_{\thv}[\partial_{\mu}\widehat{\LC}(\thv)]^2$. In general, we will have to deal with integrals of the form $\int_{\DC_{\thv}} f(U(\thv))$ where $f$ is some parametrized function --for example the cost function or its partial derivatives-- and $\DC_{\thv}:[0,2\pi]^M\rightarrow [0,1]$ is some distribution over parameter space --typically the uniform distribution. The first step is to transform the integration over parameter space to an integration over the resulting QNN  unitary distribution $\DC$.
Since $\DC$ is known to converge (given enough depth) to $\epsilon$-approximate $2$-designs over the Lie group $\er^{\mf{g}}$~\cite{larocca2021diagnosing,ragone2023unified}, assuming $f$ is a polynomial of degree $\leq 2$ in the entries of $U$ (as is the case of interest), we can replace the integration over $\DC$ with an integration over the Haar measure over $e^{\mf{g}}$. In general, $\mf{g}$ is a reductive Lie algebra consisting of multiple orthogonal ideals $\mf{g}=\bigoplus_\lm \mf{g}_\lm$, where $\mf{g}_i$ is either simple or abelian, and the Lie group $\er^{\mf{g}}$ is the product group $\bigotimes_\lm \er^{\mf{g}_\lm}$. It can be shown (see Supplementary Methods 4) that the Haar measure over such a product group is the product of the Haar measures over the normal subgroups $\er^{\mf{g}_\lm}$. Finally, the 
ansatz with generators in Eq.~\eqref{eq:generators-main} has a DLA $\mf{g}$ that is subspace controllable, meaning that each simple $\mf{g}_\lm$ is either $\mf{su}(d_\lm)$ or $\mf{u}(d_\lm)$~\cite{albertini2018controllability,kazi2023universality}. Summarizing, we have
\small
\begin{align}\label{eq:unittohaar-main}
    \int_{\DC_{\thv}} d\thv f(U(\thv)) &= \int_{\DC} dU f(U) \nonumber\\
    &\rightarrow \int_{e^{\mf{g}}} d\mu(U) f(U) \nonumber\\
    &=  \prod_\lm \int_{\mbb{U}(d_\lm)} d\mu_\lm(U_\lm) f(\{U_\lm\})\,.
\end{align}
\normalsize

The main advantage of Eq.~\eqref{eq:unittohaar-main} is that we can use tools from Weingarten calculus to perform symbolic integration over the Haar measure of unitary groups~\cite{puchala2017symbolic}.
Explicitly, we care for the variance of $\partial_{\mu}\widehat{\LC}(\thv)=\sum_{i=1}^M c_i \partial_{\mu} \ell_{\thv}(\rho_i)$ where
\begin{align}
    \partial_{\mu} \ell_{\thv}(\rho_i) =\ir\Tr[ U_B \rho_i U_B\ad [H_\mu,U_A\ad O U_A]]\,,\nonumber
\end{align}
where $U_B$ and $U_A$ denote the unitary circuits before and after the parametrized gate we are differentiating. Assuming that the depth $L$ of the QNN is enough to guarantee that both $U_A$ and $U_B$ form independent $2$-designs on $\er^{\mf{g}}$, we can use Weingarten calculus to evaluate the terms in $\mathbb{E}_{\thv}[(\partial_{\mu}\widehat{\LC}(\thv))^2]$ $\mathbb{E}_{\thv}[\partial_{\mu}\widehat{\LC}(\thv)]^2$, and obtain Eq.~\eqref{eq:var_main} in Theorem~\ref{thm:main1}. The details of this calculation are presented in Supplementary Methods 4.

While the previous, along with the results in Theorem~\ref{thm:main2}, allow to prove by direct construction that $S_n$-equivariant QNNs do not lead to barren plateaus, we here provide further intuition for this result in terms of the expressibility reduction induced by the equivariance inductive biases. 
As shown in Ref.~\cite{holmes2021connecting}, QNNs that are too expressible exhibit exponentially vanishing gradients, whereas models whose expressibility is restricted can exhibit large gradients. Hence, we can expect the result in Corollary~\ref{cor:main1} to be a direct consequence of the reduced expressibility of the model. We can further formalize this statement using the results of Ref.~\cite{larocca2021diagnosing}. Therein, it was found that there exists a link between the presence or absence of barren plateaus and the dimension of the DLA. In particular, the authors conjecture, and prove for several examples (see also Ref.~\cite{zhang2021quantum} for an independent verification of the conjecture), that deep  QNNs have gradients that scale inversely with the size of the DLA, that is,
$\text{Var}_{\thv}[\partial \hat{\LC}(\thv)] \sim \frac{1}{\poly(\dim(\mathfrak{g}))}$. For the case of $S_n$-equivariant QNNs we know from Lemma~\ref{lem:control} that $\dim(\mathfrak{g})\in\Theta(n^3)$ thus indicating that the variance should  only vanish polynomially with $n$ (for an appropriate dataset). We note this conjecture was recently proven~\cite{ragone2023unified,fontana2023theadjoint}.

\subsection*{Intuition behind the overparametrization phenomenon}

Recently, Ref.~\cite{larocca2021theory} studied the overparametrization of QNNs from the perspective of a complexity phase transition in the loss landscape. In the underparametrized regime, we experience rough loss landscapes, which in turn can be traced back to a lack of control in parametrized state space. When the number of parameters is below the number of directions in state space, the parameter update can only access a subset of those potential directions. This constraint can be shown to introduce false local minima, that is, local minima that are not actual minima of the loss function (as a function of state space) but instead artifacts of a poor parametrization. Instead, upon introduction of more parameters the parametrized state starts accessing these previously unavailable directions, and false minima disappear as we transition into the overparametrized regime. Because in the overparametrized regime the number of parameters is greater than the number of ever accessible directions, solutions in the control landscape are degenerate and form multidimensional submanifolds, allowing the optimizer to reach them more easily~\cite{larocca2020exploiting,larocca2020fourier}.

The main contribution in Ref.~\cite{larocca2021theory} is the realization that, under standard assumptions, one needs one parameter per potentially accessible direction in state space, and that the latter can be formalized as the dimension of the orbit of the initial state under the Lie group $e^\mf{g}$ resulting from the exponential of the DLA $\mf{g}$. In particular, this means that exponential DLA architectures require an exponential number of parameters to be overparametrized, whereas polynomial DLA architectures only need a polynomial number of them.

With these definitions, the proof of Theorem~\ref{theo:overparam} is immediate. Since the ansatz is subspace controllable (Lemma~\ref{lem:control}), the dimension of the DLA is equal to the dimension of the commutant, which is $\Theta(n^3)$ (Lemma~\ref{lem:poly}).

To finish, we note that the definition of overparametrization employed here (in terms of saturating the number of available directions) might differ from some definitions of overparametrization in the classical neural network community. Namely, in classical machine learning researchers have studied overparametrization through the optics of generalization ~\cite{fan2022understanding,zhang2021understanding,du2018gradient,brutzkus2018sgd,bartlett2021deep}, while others have investigated the effect of overparametrization on the training processes. In particular, it has been proposed that the onset of overparametrization can be detected using metrics such as parameter redundancy which is captured by the rank of the classical Fisher information matrix~\cite{fukumizu1996regularity,liu1998overparameterization,roychowdhury2017reducing}. It is precisely this  notion of overparametrization that  Ref.~\cite{larocca2021theory} ported to quantum,  and the one used in the present work.

\subsection*{Generalization}

We consider the QML setting in this paper where the empirical loss function is of the form $\widehat{\LC}(\thv)=\sum_{i=1}^M c_i \Tr[U_{\thv}(\rho_i) O]$. We assume that the operator norm of $O$ is bounded by a constant and also $|c_i| \leq 1/M$. We follow closely the covering number-based generalization bound in \cite{caro2021generalization}. First recall that a set $V$ is $\varepsilon$-covered by a subset $K \subseteq V$ with respect to a distance metric $d$ if $\forall x \in V$, $\exists y \in K $ such that $d(x,y)\leq \varepsilon$. The $\varepsilon$-covering number (w.r.t. metric $d$) of $V$, denoted as $\mathcal{N}(V, d, \varepsilon)$, is the cardinality of the smallest such subset \cite{shalev2014understanding}. The following theorem bounds the $\varepsilon$-covering number of $S_n$-equivariant QNNs.
\begin{theorem} The $\varepsilon$-covering number of the set $\mathcal{V}_n$ of $n$-qubit unitary $S_n$-equivariant QNNs w.r.t. the operator norm $\|\cdot\|$ can be bounded as $\mathcal{N}(\mathcal{V}_n, \| \cdot \|, \varepsilon) \leq \left(\frac{6}{\varepsilon}\right)^{2\rm{Te}_{n+1}}$.
\end{theorem}
\begin{proof} Recall that an $S_n$-EQNN $U$ can be block-diagonalized as $U \cong \bigoplus_{\lambda} \id_{m_\lambda} \otimes U_\lambda $, where each $U_\lambda$ is a unitary for $U$ to be unitary. Let $\mathbb{U}(d_\lambda)$ denote the set of all unitaries of dimension $d_\lambda$.
Following Lemma 6 in \cite{caro2021generalization} and Section 4.2 in \cite{vershynin2018highdimensional} we can bound the $\varepsilon$-covering number of $\mathbb{U}_{d_\lambda}$ as follows
\begin{equation}
    \mathcal{N}(\mathbb{U}(d_\lambda), \| \cdot \|, \varepsilon) \leq \left(\frac{6}{\varepsilon} \right)^{2d_\lambda^2}.
\end{equation}
Next, we construct an $\varepsilon$-covering subset of the $S_n$-equivariant unitary set, $\mathcal{V}_n$, from the $\varepsilon$-covering subsets, $K_\lambda$, of the blocks $\lambda$. Indeed, given any $U \cong \bigoplus_\lambda \id_{m_\lambda} \otimes U_\lambda$, we can identify unitaries $\Tilde{U}_\lambda$ from $K_\lambda$ such that $\|U_\lambda - \Tilde{U}_\lambda\| \leq \varepsilon, \forall \lambda$. The unitary $\Tilde{U} \cong \bigoplus_\lambda \id_{m_\lambda} \otimes \Tilde{U}_\lambda$ then satisfies
\begin{equation}
    \|  U -\Tilde{U} \| \leq \max_\lambda \|U_\lambda - \Tilde{U}_\lambda\| \leq \varepsilon.
\end{equation}
Therefore, there exists an $\varepsilon$-covering net of $\mathcal{V}_n$ of size $\prod_\lambda \left(\frac{6}{\varepsilon} \right)^{2d_\lambda^2} = \left(\frac{6}{\varepsilon} \right)^{2 \rm{Te}_{n+1}}$, concluding the proof.
\end{proof}
Having established this bound on the $\varepsilon$-covering numbers of $S_n$-EQNN, we apply a known result from \cite{caro2021generalization} (with some extra care) to obtain Theorem \ref{thm:gen}.
\begin{proof}[Proof of Theorem \ref{thm:gen}]
We assume knowledge of Theorem 6 in \cite{caro2021generalization}. In step two of the proof where the authors use the chaining argument \cite{dudley1999uniform} to bound the generalization error, notice that the covering number $\mathcal{N}_j$ in their Eq. (64) is replaced by $\left(\frac{6}{\varepsilon} \right)^{2\rm{Te}_{n+1}}$ in our case. In other words, there is no architecture-dependence (the number of gates $T$ in their case) inside the logarithm in the resulting Eq. (65). Applying this change to the rest of their proof leads to our claimed generalization bound.
\end{proof}

We note that in the previous derivation, we have used  knowledge of the isotypic decomposition of the $S_n$-equivariant QNN, which allows us to obtain a specialized generalization error bound that does not follow from a direct application of the results in~\cite{caro2021generalization}.

\subsection*{Trainable and untrainable states}

Here we describe how the states in Table~\ref{tab:train_states} are obtained. The "symmetric states" are obtained from the symmetric subspace~\cite{harrow2013church}, i.e., the set of states $\{\ket{\psi}\in \mathcal{H}\,|\, R(\pi) \ket{\psi} = \ket{\psi},\ \forall \pi \in S_n\}$. The so-called "fixed Hamming-weight encoded" states correspond to states representing classical data: Given an array of real values $\{x_i\}$, such that $\sum_i x_i^2 = 1$, each $x_i$ is encoded as the weight of a unique bitstring $\textbf{z}$ of Hamming weight $k$, where $k$ is some fixed constant. That is, prepare the state $\ket{\textbf{x}} = \sum_{\substack{\textbf{z} \text{ s.t. } w(\textbf{z}) = k}} x_\textbf{z}\ket{\textbf{z}}$, where we are now indexing $x_i$ with a bitstring $\textbf{z}$. "Local  Haar random" states are obtained by preparing the state $\ket{0}^{\otimes n}$ and applying a Haar random single-qubit unitary to each qubit.  "Global  Haar random" states are obtained by preparing the state $\ket{0}^{\otimes n}$ and applying a random $n$-qubit unitary sampled from the Haar measure over $\mathbb{U}(d)$. The "fixed and linear depth random circuit" states correspond to the states obtained by preparing the state $\ket{0}^{\otimes n}$  and respectively applying a constant-depth, or linear-depth layered hardware-efficient quantum circuit~\cite{kandala2017hardware,cerezo2020cost} with random parameters. For the "graph states", we use a canonical encoding to embed a  graph into a quantum state~\cite{raussendorf2003measurement, hein2004multiparty}. Specifically, to create a graph state, one starts with the state $\ket{+}^{\otimes n}$, and applies a controlled-Z rotation for every edge in the graph. We consider $3$-regular and $n/2$-regular graphs, as well as random graphs generated according to the Erd\"{o}s–R\'{e}nyi model~\cite{erdds1959random}. 

\section*{DATA AVAILABILITY}
Data generated and analyzed during current study are available from the corresponding author upon reasonable request.

\section*{CODE AVAILABILILTY}
Code used to generate data in this study are available from the corresponding author upon reasonable request.

\section*{ACKNOWLEDGMENTS}

We thank Michael Ragone and Paolo Braccia for insightful discussion on geometric quantum machine learning. We also thank Felix Leditzky for discussion regarding Hermitian Young operators and Dylan Herman for useful questions and comments. L.S. was partially supported by the NSF Quantum Leap Challenge Institute for Hybrid Quantum Architectures and Networks (NSF Award 2016136). L.S. also acknowledges  supported by LANL ASC Beyond Moore’s Law project.  M.L. acknowledges initial support by the Center for Nonlinear Studies at Los Alamos National Laboratory (LANL) and by the U.S. Department of Energy (DOE), Office of Science, Office of Advanced Scientific Computing Research, under the Accelerated Research in Quantum Computing (ARQC) program. F.S. acknowledges support by the  Directed Research and Development (LDRD) program of LANL under project number 20220745ER. M.C. and M.L. were partially supported by Directed Research and Development (LDRD) program of LANL under project number 20210116DR and 20230049DR. This work was also supported by NSEC Quantum Sensing at LANL.

\section*{AUTHOR CONTRIBUTIONS}
The project was conceived by ML and MC. The manuscript was written by LS, ML, QTN, FS and MC. Theoretical results were derived by LS, ML, QTN and MC. Numerical simulations were performed by LS, ML and FS. LS and ML contributed equally and are to be considered co-first authors.

\section*{COMPETING INTERESTS}
The authors declare no competing interests.


\begin{thebibliography}{100}
\expandafter\ifx\csname url\endcsname\relax
  \def\url#1{\texttt{#1}}\fi
\expandafter\ifx\csname urlprefix\endcsname\relax\def\urlprefix{URL }\fi
\providecommand{\bibinfo}[2]{#2}
\providecommand{\eprint}[2][]{\url{#2}}

\bibitem{cohen2016group}
\bibinfo{author}{Cohen, T.} \& \bibinfo{author}{Welling, M.}
\newblock \bibinfo{title}{Group equivariant convolutional networks}.
\newblock \emph{\bibinfo{journal}International Conference on Machine Learning} \textbf{\bibinfo{volume}{33}}, (\bibinfo{year}{2016}).

\bibitem{bronstein2021geometric}
\bibinfo{author}{Bronstein, M.~M.}, \bibinfo{author}{Bruna, J.},
  \bibinfo{author}{Cohen, T.} \& \bibinfo{author}{Veli{\v{c}}kovi{\'c}, P.}
\newblock \bibinfo{title}{Geometric deep learning: Grids, groups, graphs,
  geodesics, and gauges}.
\newblock Preprint at \url{https://arxiv.org/abs/2104.13478} (\bibinfo{year}{2021}).

\bibitem{kondor2018generalization}
\bibinfo{author}{Kondor, R.} \& \bibinfo{author}{Trivedi, S.}
\newblock \bibinfo{title}{On the generalization of equivariance and convolution
  in neural networks to the action of compact groups}.
\newblock \emph{\bibinfo{journal}International Conference on Machine Learning} \textbf{\bibinfo{volume}{35}}, (\bibinfo{year}{2018}).

\bibitem{bogatskiy2022symmetry}
\bibinfo{author}{Bogatskiy, A.} \emph{et~al.}
\newblock \bibinfo{title}{Symmetry group equivariant architectures for
  physics}.
\newblock Preprint at \url{https://arxiv.org/abs/2203.06153} (\bibinfo{year}{2022}).

\bibitem{bekkers2018roto}
\bibinfo{author}{Bekkers, E.~J.} \emph{et~al.}
\newblock \bibinfo{title}{Roto-translation covariant convolutional networks for
  medical image analysis}.
\newblock \emph{\bibinfo{booktitle}{International conference on medical
  image computing and computer-assisted intervention}},
  \bibinfo{pages}{440--448} (\bibinfo{year}{2018}).

\bibitem{schutt2017continuous}
\bibinfo{author}{Sch\"{u}tt, K.~T.} \emph{et~al.}
\newblock \bibinfo{title}{Schnet: A continuous-filter convolutional neural
  network for modeling quantum interactions}.
\newblock \emph{Adv. Neural Inf. Process. Sys.} \textbf{\bibinfo{volume}{31}}, (\bibinfo{year}{2017}).

\bibitem{boyda2021sampling}
\bibinfo{author}{Boyda, D.} \emph{et~al.}
\newblock \bibinfo{title}{Sampling using su (n) gauge equivariant flows}.
\newblock \emph{\bibinfo{journal}{Phys. Rev. D}}
  \textbf{\bibinfo{volume}{103}}, \bibinfo{pages}{074504}
  (\bibinfo{year}{2021}).

\bibitem{rezende2019equivariant}
\bibinfo{author}{Rezende, D.~J.}, \bibinfo{author}{Racani{\`e}re, S.},
  \bibinfo{author}{Higgins, I.} \& \bibinfo{author}{Toth, P.}
\newblock \bibinfo{title}{Equivariant {H}amiltonian flows}.
\newblock Preprint at \url{https://arxiv.org/abs/1909.13739} (\bibinfo{year}{2019}).

\bibitem{thomas2018tensor}
\bibinfo{author}{Thomas, N.} \emph{et~al.}
\newblock \bibinfo{title}{Tensor field networks: Rotation-and
  translation-equivariant neural networks for 3d point clouds}.
\newblock Preprint at \url{https://arxiv.org/abs/1802.08219} (\bibinfo{year}{2018}).

\bibitem{toth2019hamiltonian}
\bibinfo{author}{Toth, P.} \emph{et~al.}
\newblock \bibinfo{title}{Hamiltonian generative networks}.
\newblock Preprint at \url{https://arxiv.org/abs/1909.13789} (\bibinfo{year}{2019}).

\bibitem{kohler2020equivariant}
\bibinfo{author}{K\"{o}hler, J.}, \bibinfo{author}{Klein, L.} \&
  \bibinfo{author}{No\'{e}, F.}
\newblock \bibinfo{title}{Equivariant flows: Exact likelihood generative
  learning for symmetric densities}.
\newblock \emph{\bibinfo{journal}International Conference on Machine Learning} \textbf{\bibinfo{volume}{37}}, (\bibinfo{year}{2020}).

\bibitem{anderson2019cormorant}
\bibinfo{author}{Anderson, B.}, \bibinfo{author}{Hy, T.~S.} \&
  \bibinfo{author}{Kondor, R.}
\newblock \bibinfo{title}{Cormorant: Covariant molecular neural networks}.
\newblock \emph{\bibinfo{journal}{Adv. Neural Inf. Process. Sys.}} \textbf{\bibinfo{volume}{32}} (\bibinfo{year}{2019}).

\bibitem{bogatskiy2020lorentz}
\bibinfo{author}{Bogatskiy, A.} \emph{et~al.}
\newblock \bibinfo{title}{Lorentz group equivariant neural network for particle
  physics}.
\newblock \emph{\bibinfo{journal}International Conference on Machine Learning} \textbf{\bibinfo{volume}{37}}, (\bibinfo{year}{2020}).

\bibitem{schuld2015introduction}
\bibinfo{author}{Schuld, M.}, \bibinfo{author}{Sinayskiy, I.} \&
  \bibinfo{author}{Petruccione, F.}
\newblock \bibinfo{title}{An introduction to quantum machine learning}.
\newblock \emph{\bibinfo{journal}{Contemp. Phys.}}
  \textbf{\bibinfo{volume}{56}}, \bibinfo{pages}{172--185}
  (\bibinfo{year}{2015}).

\bibitem{biamonte2017quantum}
\bibinfo{author}{Biamonte, J.} \emph{et~al.}
\newblock \bibinfo{title}{Quantum machine learning}.
\newblock \emph{\bibinfo{journal}{Nature}} \textbf{\bibinfo{volume}{549}},
  \bibinfo{pages}{195--202} (\bibinfo{year}{2017}).

\bibitem{cerezo2022challenges}
\bibinfo{author}{Cerezo, M.}, \bibinfo{author}{Verdon, G.},
  \bibinfo{author}{Huang, H.-Y.}, \bibinfo{author}{Cincio, L.} \&
  \bibinfo{author}{Coles, P.~J.}
\newblock \bibinfo{title}{Challenges and opportunities in quantum machine
  learning}.
\newblock \emph{\bibinfo{journal}{Nat. Comput. Sci.}}
  (\bibinfo{year}{2022}).

\bibitem{huang2021provably}
\bibinfo{author}{Huang, H.-Y.}, \bibinfo{author}{Kueng, R.},
  \bibinfo{author}{Torlai, G.}, \bibinfo{author}{Albert, V.~V.} \&
  \bibinfo{author}{Preskill, J.}
\newblock \bibinfo{title}{Provably efficient machine learning for quantum
  many-body problems}.
\newblock \emph{\bibinfo{journal}{Science}} \textbf{\bibinfo{volume}{377}},
  \bibinfo{pages}{eabk3333} (\bibinfo{year}{2022}).

\bibitem{larocca2022group}
\bibinfo{author}{Larocca, M.} \emph{et~al.}
\newblock \bibinfo{title}{Group-invariant quantum machine learning}.
\newblock \emph{\bibinfo{journal}{PRX Quantum}} \textbf{\bibinfo{volume}{3}},
  \bibinfo{pages}{030341} (\bibinfo{year}{2022}).

\bibitem{meyer2022exploiting}
\bibinfo{author}{Meyer, J.~J.} \emph{et~al.}
\newblock \bibinfo{title}{Exploiting symmetry in variational quantum machine
  learning}.
\newblock \emph{\bibinfo{journal}{PRX Quantum}} \textbf{\bibinfo{volume}{4}},
  \bibinfo{pages}{010328} (\bibinfo{year}{2023}).

\bibitem{sauvage2022building}
\bibinfo{author}{Sauvage, F.}, \bibinfo{author}{Larocca, M.},
  \bibinfo{author}{Coles, P.~J.} \& \bibinfo{author}{Cerezo, M.}
\newblock \bibinfo{title}{Building spatial symmetries into parameterized
  quantum circuits for faster training}.
\newblock Preprint at \url{https://arxiv.org/abs/2207.14413} (\bibinfo{year}{2022}).

\bibitem{zheng2022super}
\bibinfo{author}{Zheng, H.}, \bibinfo{author}{Li, Z.}, \bibinfo{author}{Liu,
  J.}, \bibinfo{author}{Strelchuk, S.} \& \bibinfo{author}{Kondor, R.}
\newblock \bibinfo{title}{On the super-exponential quantum speedup of
  equivariant quantum machine learning algorithms with su($ d $) symmetry}.
\newblock Preprint at \url{https://arxiv.org/abs/2207.07250} (\bibinfo{year}{2022}).

\bibitem{zheng2021speeding}
\bibinfo{author}{Zheng, H.}, \bibinfo{author}{Li, Z.}, \bibinfo{author}{Liu,
  J.}, \bibinfo{author}{Strelchuk, S.} \& \bibinfo{author}{Kondor, R.}
\newblock \bibinfo{title}{Speeding up learning quantum states through group
  equivariant convolutional quantum ans{\"a}tze}.
\newblock \emph{\bibinfo{journal}{PRX Quantum}} \textbf{\bibinfo{volume}{4}},
  \bibinfo{pages}{020327} (\bibinfo{year}{2023}).

\bibitem{nguyen2022atheory}
\bibinfo{author}{Nguyen, Q.~T.} \emph{et~al.}
\newblock \bibinfo{title}{A theory for equivariant quantum neural networks}.
\newblock Preprint at \url{https://arxiv.org/abs/2210.08566} (\bibinfo{year}{2022}).

\bibitem{wang2022symmetric}
\bibinfo{author}{Wang, X.} \emph{et~al.}
\newblock \bibinfo{title}{Symmetric pruning in quantum neural networks}.
\newblock Preprint at \url{https://arxiv.org/abs/2208.14057} (\bibinfo{year}{2022}).

\bibitem{ragone2022representation}
\bibinfo{author}{Ragone, M.} \emph{et~al.}
\newblock \bibinfo{title}{Representation theory for geometric quantum machine
  learning}.
\newblock Preprint at \url{https://arxiv.org/abs/2210.07980} (\bibinfo{year}{2022}).

\bibitem{abbas2020power}
\bibinfo{author}{Abbas, A.} \emph{et~al.}
\newblock \bibinfo{title}{The power of quantum neural networks}.
\newblock \emph{\bibinfo{journal}{Nat. Comput. Sci.}}
  \textbf{\bibinfo{volume}{1}}, \bibinfo{pages}{403--409}
  (\bibinfo{year}{2021}).

\bibitem{liu2022analytic}
\bibinfo{author}{Liu, J.} \emph{et~al.}
\newblock \bibinfo{title}{An analytic theory for the dynamics of wide quantum
  neural networks}.
\newblock Preprint at \url{https://arxiv.org/abs/2203.16711} (\bibinfo{year}{2022}).

\bibitem{liu2021representation}
\bibinfo{author}{Liu, J.}, \bibinfo{author}{Tacchino, F.},
  \bibinfo{author}{Glick, J.~R.}, \bibinfo{author}{Jiang, L.} \&
  \bibinfo{author}{Mezzacapo, A.}
\newblock \bibinfo{title}{Representation learning via quantum neural tangent
  kernels}.
\newblock Preprint at \url{https://arxiv.org/abs/2111.04225} (\bibinfo{year}{2021}).

\bibitem{bittel2021training}
\bibinfo{author}{Bittel, L.} \& \bibinfo{author}{Kliesch, M.}
\newblock \bibinfo{title}{Training variational quantum algorithms is
  {NP}-hard}.
\newblock \emph{\bibinfo{journal}{Phys. Rev. Lett.}}
  \textbf{\bibinfo{volume}{127}}, \bibinfo{pages}{120502}
  (\bibinfo{year}{2021}).

\bibitem{anschuetz2022beyond}
\bibinfo{author}{Anschuetz, E.~R.} \& \bibinfo{author}{Kiani, B.~T.}
\newblock \bibinfo{title}{Beyond barren plateaus: Quantum variational
  algorithms are swamped with traps}.
\newblock \emph{\bibinfo{journal}{Nat. Commun.}}
  \textbf{\bibinfo{volume}{13}}, \bibinfo{pages}{7760} (\bibinfo{year}{2022}).

\bibitem{fontana2022nontrivial}
\bibinfo{author}{Fontana, E.}, \bibinfo{author}{Cerezo, M.},
  \bibinfo{author}{Arrasmith, A.}, \bibinfo{author}{Rungger, I.} \&
  \bibinfo{author}{Coles, P.~J.}
\newblock \bibinfo{title}{Non-trivial symmetries in quantum landscapes and
  their resilience to quantum noise}.
\newblock \emph{\bibinfo{journal}{Quantum}} \textbf{\bibinfo{volume}{6}},
  \bibinfo{pages}{804}; 10.22331/q-2022-09-15-804 (\bibinfo{year}{2022}).

\bibitem{larocca2021theory}
\bibinfo{author}{Larocca, M.}, \bibinfo{author}{Ju, N.},
  \bibinfo{author}{García-Martín, D.}, \bibinfo{author}{Coles, P.~J.} \&
  \bibinfo{author}{Cerezo, M.}
\newblock \bibinfo{title}{Theory of overparametrization in quantum neural
  networks}.
\newblock \emph{\bibinfo{journal}{Nat. Comput. Sci.}}
  \textbf{\bibinfo{volume}{3}}, \bibinfo{pages}{542--–551}
  (\bibinfo{year}{2023}).

\bibitem{mcclean2018barren}
\bibinfo{author}{McClean, J.~R.}, \bibinfo{author}{Boixo, S.},
  \bibinfo{author}{Smelyanskiy, V.~N.}, \bibinfo{author}{Babbush, R.} \&
  \bibinfo{author}{Neven, H.}
\newblock \bibinfo{title}{Barren plateaus in quantum neural network training
  landscapes}.
\newblock \emph{\bibinfo{journal}{Nat. Commun.}}
  \textbf{\bibinfo{volume}{9}}, \bibinfo{pages}{1--6} (\bibinfo{year}{2018}).

\bibitem{cerezo2020cost}
\bibinfo{author}{Cerezo, M.}, \bibinfo{author}{Sone, A.},
  \bibinfo{author}{Volkoff, T.}, \bibinfo{author}{Cincio, L.} \&
  \bibinfo{author}{Coles, P.~J.}
\newblock \bibinfo{title}{Cost function dependent barren plateaus in shallow
  parametrized quantum circuits}.
\newblock \emph{\bibinfo{journal}{Nat. Commun.}}
  \textbf{\bibinfo{volume}{12}}, \bibinfo{pages}{1--12} (\bibinfo{year}{2021}).

\bibitem{sharma2020trainability}
\bibinfo{author}{Sharma, K.}, \bibinfo{author}{Cerezo, M.},
  \bibinfo{author}{Cincio, L.} \& \bibinfo{author}{Coles, P.~J.}
\newblock \bibinfo{title}{Trainability of dissipative perceptron-based quantum
  neural networks}.
\newblock \emph{\bibinfo{journal}{Phys. Rev. Lett.}}
  \textbf{\bibinfo{volume}{128}}, \bibinfo{pages}{180505}
  (\bibinfo{year}{2022}).

\bibitem{holmes2020barren}
\bibinfo{author}{Holmes, Z.} \emph{et~al.}
\newblock \bibinfo{title}{Barren plateaus preclude learning scramblers}.
\newblock \emph{\bibinfo{journal}{Phys. Rev. Lett.}}
  \textbf{\bibinfo{volume}{126}}, \bibinfo{pages}{190501}
  (\bibinfo{year}{2021}).

\bibitem{holmes2021connecting}
\bibinfo{author}{Holmes, Z.}, \bibinfo{author}{Sharma, K.},
  \bibinfo{author}{Cerezo, M.} \& \bibinfo{author}{Coles, P.~J.}
\newblock \bibinfo{title}{Connecting ansatz expressibility to gradient
  magnitudes and barren plateaus}.
\newblock \emph{\bibinfo{journal}{PRX Quantum}} \textbf{\bibinfo{volume}{3}},
  \bibinfo{pages}{010313} (\bibinfo{year}{2022}).

\bibitem{cerezo2020impact}
\bibinfo{author}{Cerezo, M.} \& \bibinfo{author}{Coles, P.~J.}
\newblock \bibinfo{title}{Higher order derivatives of quantum neural networks
  with barren plateaus}.
\newblock \emph{\bibinfo{journal}{Quantum Sci. Technol.}}
  \textbf{\bibinfo{volume}{6}}, \bibinfo{pages}{035006} (\bibinfo{year}{2021}).

\bibitem{marrero2020entanglement}
\bibinfo{author}{Marrero, C.~O.}, \bibinfo{author}{Kieferov{\'a}, M.} \&
  \bibinfo{author}{Wiebe, N.}
\newblock \bibinfo{title}{Entanglement-induced barren plateaus}.
\newblock \emph{\bibinfo{journal}{PRX Quantum}} \textbf{\bibinfo{volume}{2}},
  \bibinfo{pages}{040316} (\bibinfo{year}{2021}).

\bibitem{patti2020entanglement}
\bibinfo{author}{Patti, T.~L.}, \bibinfo{author}{Najafi, K.},
  \bibinfo{author}{Gao, X.} \& \bibinfo{author}{Yelin, S.~F.}
\newblock \bibinfo{title}{Entanglement devised barren plateau mitigation}.
\newblock \emph{\bibinfo{journal}{Phys. Rev. Res.}}
  \textbf{\bibinfo{volume}{3}}, \bibinfo{pages}{033090} (\bibinfo{year}{2021}).

\bibitem{uvarov2020barren}
\bibinfo{author}{Uvarov, A.} \& \bibinfo{author}{Biamonte, J.~D.}
\newblock \bibinfo{title}{On barren plateaus and cost function locality in
  variational quantum algorithms}.
\newblock \emph{\bibinfo{journal}{Journal of Physics A: Mathematical and
  Theoretical}} \textbf{\bibinfo{volume}{54}}, \bibinfo{pages}{245301}
  (\bibinfo{year}{2021}).

\bibitem{thanasilp2021subtleties}
\bibinfo{author}{Thanasilp, S.}, \bibinfo{author}{Wang, S.},
  \bibinfo{author}{Nghiem, N.~A.}, \bibinfo{author}{Coles, P.~J.} \&
  \bibinfo{author}{Cerezo, M.}
\newblock \bibinfo{title}{Subtleties in the trainability of quantum machine
  learning models}.
\newblock Preprint at \url{https://arxiv.org/abs/2110.14753} (\bibinfo{year}{2021}).

\bibitem{larocca2021diagnosing}
\bibinfo{author}{Larocca, M.} \emph{et~al.}
\newblock \bibinfo{title}{Diagnosing {B}arren {P}lateaus with {T}ools from
  {Q}uantum {O}ptimal {C}ontrol}.
\newblock \emph{\bibinfo{journal}{{Quantum}}} \textbf{\bibinfo{volume}{6}},
  \bibinfo{pages}{824}; 10.22331/q-2022-09-29-824 (\bibinfo{year}{2022}).

\bibitem{wang2020noise}
\bibinfo{author}{Wang, S.} \emph{et~al.}
\newblock \bibinfo{title}{Noise-induced barren plateaus in variational quantum
  algorithms}.
\newblock \emph{\bibinfo{journal}{Nat. Commun.}}
  \textbf{\bibinfo{volume}{12}}, \bibinfo{pages}{1--11} (\bibinfo{year}{2021}).

\bibitem{wiersema2020exploring}
\bibinfo{author}{Wiersema, R.} \emph{et~al.}
\newblock \bibinfo{title}{Exploring entanglement and optimization within the
  {H}amiltonian variational ansatz}.
\newblock \emph{\bibinfo{journal}{PRX Quantum}} \textbf{\bibinfo{volume}{1}},
  \bibinfo{pages}{020319} (\bibinfo{year}{2020}).

\bibitem{sim2019expressibility}
\bibinfo{author}{Sim, S.}, \bibinfo{author}{Johnson, P.~D.} \&
  \bibinfo{author}{Aspuru-Guzik, A.}
\newblock \bibinfo{title}{Expressibility and entangling capability of
  parameterized quantum circuits for hybrid quantum-classical algorithms}.
\newblock \emph{\bibinfo{journal}{Adv. Quantum Technol.}}
  \textbf{\bibinfo{volume}{2}}, \bibinfo{pages}{1900070}
  (\bibinfo{year}{2019}).

\bibitem{zaheer2017deep}
\bibinfo{author}{Zaheer, M.} \emph{et~al.}
\newblock \bibinfo{title}{Deep sets}.
\newblock \bibinfo{editor}{Guyon, I.} \emph{et~al.} (eds.)
  \emph{\bibinfo{booktitle}{Adv. Neural Inf. Process. Sys.}} \textbf{\bibinfo{volume}{30}}, (\bibinfo{year}{2017}).

\bibitem{maron2020learning}
\bibinfo{author}{Maron, H.}, \bibinfo{author}{Litany, O.},
  \bibinfo{author}{Chechik, G.} \& \bibinfo{author}{Fetaya, E.}
\newblock \bibinfo{title}{On learning sets of symmetric elements}.
\newblock \emph{\bibinfo{journal}International Conference on Machine Learning} \textbf{\bibinfo{volume}{37}}, (\bibinfo{year}{2020}).

\bibitem{maron2018invariant}
\bibinfo{author}{Maron, H.}, \bibinfo{author}{Ben-Hamu, H.},
  \bibinfo{author}{Shamir, N.} \& \bibinfo{author}{Lipman, Y.}
\newblock \bibinfo{title}{Invariant and equivariant graph networks}.
\newblock In \emph{\bibinfo{booktitle}{International Conference on Learning
  Representations}} (\bibinfo{year}{2019}).

\bibitem{keriven2019universal}
\bibinfo{author}{Keriven, N.} \& \bibinfo{author}{Peyr\'{e}, G.}
\newblock \bibinfo{title}{Universal invariant and equivariant graph neural
  networks}.
\newblock \emph{Adv. Neural Inf. Process. Sys.} \textbf{\bibinfo{volume}{32}}, (\bibinfo{year}{2019}).

\bibitem{maron2019provably}
\bibinfo{author}{Maron, H.}, \bibinfo{author}{Ben-Hamu, H.},
  \bibinfo{author}{Serviansky, H.} \& \bibinfo{author}{Lipman, Y.}
\newblock \bibinfo{title}{Provably powerful graph networks}.
\newblock \emph{\bibinfo{journal}{Adv. Neural Inf. Process. Sys.}} \textbf{\bibinfo{volume}{32}}, (\bibinfo{year}{2019}).

\bibitem{verdon2019quantumgraph}
\bibinfo{author}{Verdon, G.} \emph{et~al.}
\newblock \bibinfo{title}{Quantum graph neural networks}.
\newblock Preprint at \url{https://arxiv.org/abs/1909.12264} (\bibinfo{year}{2019}).

\bibitem{mernyei2022equivariant}
\bibinfo{author}{Mernyei, P.}, \bibinfo{author}{Meichanetzidis, K.} \&
  \bibinfo{author}{Ceylan, I.~I.}
\newblock \bibinfo{title}{Equivariant quantum graph circuits}.
\newblock \emph{\bibinfo{journal}International Conference on Machine Learning} \textbf{\bibinfo{volume}{39}}, (\bibinfo{year}{2022}).

\bibitem{skolik2022equivariant}
\bibinfo{author}{Skolik, A.}, \bibinfo{author}{Cattelan, M.},
  \bibinfo{author}{Yarkoni, S.}, \bibinfo{author}{B{\"a}ck, T.} \&
  \bibinfo{author}{Dunjko, V.}
\newblock \bibinfo{title}{Equivariant quantum circuits for learning on weighted
  graphs}.
\newblock \emph{\bibinfo{journal}{Npj Quantum Inf.}}
  \textbf{\bibinfo{volume}{9}}, \bibinfo{pages}{47} (\bibinfo{year}{2023}).

\bibitem{maron2019universality}
\bibinfo{author}{Maron, H.}, \bibinfo{author}{Fetaya, E.},
  \bibinfo{author}{Segol, N.} \& \bibinfo{author}{Lipman, Y.}
\newblock \bibinfo{title}{On the universality of invariant networks}.
\newblock \emph{\bibinfo{journal}International Conference on Machine Learning} \textbf{\bibinfo{volume}{36}}, (\bibinfo{year}{2019}).

\bibitem{thiede2020general}
\bibinfo{author}{Thiede, E.~H.}, \bibinfo{author}{Hy, T.~S.} \&
  \bibinfo{author}{Kondor, R.}
\newblock \bibinfo{title}{The general theory of permutation equivarant neural
  networks and higher order graph variational encoders}.
\newblock Preprint at \url{https://arxiv.org/abs/2004.03990} (\bibinfo{year}{2020}).

\bibitem{pan2022permutation}
\bibinfo{author}{Pan, H.} \& \bibinfo{author}{Kondor, R.}
\newblock \bibinfo{title}{Permutation equivariant layers for higher order
  interactions}.
\newblock In \bibinfo{editor}{Camps-Valls, G.}, \bibinfo{editor}{Ruiz, F.
  J.~R.} \& \bibinfo{editor}{Valera, I.} (eds.)
  \emph{\bibinfo{booktitle}{Proceedings of The 25th International Conference on
  Artificial Intelligence and Statistics}}, vol. \bibinfo{volume}{151} of
  \emph{\bibinfo{series}{Proceedings of Machine Learning Research}},
  \bibinfo{pages}{5987--6001} (\bibinfo{publisher}{PMLR},
  \bibinfo{year}{2022}).

\bibitem{farhi2014quantum}
\bibinfo{author}{Farhi, E.}, \bibinfo{author}{Goldstone, J.} \&
  \bibinfo{author}{Gutmann, S.}
\newblock \bibinfo{title}{A quantum approximate optimization algorithm}.
\newblock Preprint at \url{https://arxiv.org/abs/1411.4028} (\bibinfo{year}{2014}).

\bibitem{hadfield2019quantum}
\bibinfo{author}{Hadfield, S.} \emph{et~al.}
\newblock \bibinfo{title}{From the quantum approximate optimization algorithm
  to a quantum alternating operator ansatz}.
\newblock \emph{\bibinfo{journal}{Algorithms}} \textbf{\bibinfo{volume}{12}},
  \bibinfo{pages}{34} (\bibinfo{year}{2019}).

\bibitem{cong2019quantum}
\bibinfo{author}{Cong, I.}, \bibinfo{author}{Choi, S.} \&
  \bibinfo{author}{Lukin, M.~D.}
\newblock \bibinfo{title}{Quantum convolutional neural networks}.
\newblock \emph{\bibinfo{journal}{Nat. Phys.}}
  \textbf{\bibinfo{volume}{15}}, \bibinfo{pages}{1273--1278}
  (\bibinfo{year}{2019}).

\bibitem{caro2021generalization}
\bibinfo{author}{Caro, M.~C.} \emph{et~al.}
\newblock \bibinfo{title}{Generalization in quantum machine learning from few
  training data}.
\newblock \emph{\bibinfo{journal}{Nat. Commun.}}
  \textbf{\bibinfo{volume}{13}} (\bibinfo{year}{2022}).

\bibitem{peruzzo2014variational}
\bibinfo{author}{Peruzzo, A.} \emph{et~al.}
\newblock \bibinfo{title}{A variational eigenvalue solver on a photonic quantum
  processor}.
\newblock \emph{\bibinfo{journal}{Nat. Commun.}}
  \textbf{\bibinfo{volume}{5}}, \bibinfo{pages}{1--7} (\bibinfo{year}{2014}).

\bibitem{cerezo2020variationalreview}
\bibinfo{author}{Cerezo, M.} \emph{et~al.}
\newblock \bibinfo{title}{Variational quantum algorithms}.
\newblock \emph{\bibinfo{journal}{Nat. Rev. Phys.}}
  \textbf{\bibinfo{volume}{3}}, \bibinfo{pages}{625–644}
  (\bibinfo{year}{2021}).

\bibitem{tang2019qubit}
\bibinfo{author}{Tang, H.~L.} \emph{et~al.}
\newblock \bibinfo{title}{qubit-adapt-vqe: An adaptive algorithm for
  constructing hardware-efficient ans{\"a}tze on a quantum processor}.
\newblock \emph{\bibinfo{journal}{PRX Quantum}} \textbf{\bibinfo{volume}{2}},
  \bibinfo{pages}{020310} (\bibinfo{year}{2021}).

\bibitem{horodecki2009quantum}
\bibinfo{author}{Horodecki, R.}, \bibinfo{author}{Horodecki, P.},
  \bibinfo{author}{Horodecki, M.} \& \bibinfo{author}{Horodecki, K.}
\newblock \bibinfo{title}{Quantum entanglement}.
\newblock \emph{\bibinfo{journal}{Rev. Mod. Phys.}}
  \textbf{\bibinfo{volume}{81}}, \bibinfo{pages}{865} (\bibinfo{year}{2009}).

\bibitem{walter2016multipartite}
\bibinfo{author}{Walter, M.}, \bibinfo{author}{Gross, D.} \&
  \bibinfo{author}{Eisert, J.}
\newblock \bibinfo{title}{Multipartite entanglement}.
\newblock \emph{\bibinfo{journal}{Quantum Information: From Foundations to
  Quantum Technology Applications}} \bibinfo{pages}{293--330}
  (\bibinfo{year}{2016}).

\bibitem{beckey2021computable}
\bibinfo{author}{Beckey, J.~L.}, \bibinfo{author}{Gigena, N.},
  \bibinfo{author}{Coles, P.~J.} \& \bibinfo{author}{Cerezo, M.}
\newblock \bibinfo{title}{Computable and operationally meaningful multipartite
  entanglement measures}.
\newblock \emph{\bibinfo{journal}{Phys. Rev. Lett.}}
  \textbf{\bibinfo{volume}{127}}, \bibinfo{pages}{140501}
  (\bibinfo{year}{2021}).

\bibitem{schatzki2022hierarchy}
\bibinfo{author}{Schatzki, L.}, \bibinfo{author}{Liu, G.},
  \bibinfo{author}{Cerezo, M.} \& \bibinfo{author}{Chitambar, E.}
\newblock \bibinfo{title}{A hierarchy of multipartite correlations based on
  concentratable entanglement}.
\newblock Preprint at \url{https://arxiv.org/pdf/2209.07607.pdf} (\bibinfo{year}{2022}).

\bibitem{guo2020distributed}
\bibinfo{author}{Guo, X.} \emph{et~al.}
\newblock \bibinfo{title}{Distributed quantum sensing in a continuous-variable
  entangled network}.
\newblock \emph{\bibinfo{journal}{Nat. Phys.}}
  \textbf{\bibinfo{volume}{16}}, \bibinfo{pages}{281--284}
  (\bibinfo{year}{2020}).

\bibitem{zhang2021distributed}
\bibinfo{author}{Zhang, Z.} \& \bibinfo{author}{Zhuang, Q.}
\newblock \bibinfo{title}{Distributed quantum sensing}.
\newblock \emph{\bibinfo{journal}{Quantum Sci. Technol.}}
  \textbf{\bibinfo{volume}{6}}, \bibinfo{pages}{043001} (\bibinfo{year}{2021}).

\bibitem{huerta2022inference}
\bibinfo{author}{{Huerta Alderete}, C.} \emph{et~al.}
\newblock \bibinfo{title}{Inference-based quantum sensing}.
\newblock \emph{\bibinfo{journal}{Phys. Rev. Lett.}}
  \textbf{\bibinfo{volume}{129}}, \bibinfo{pages}{190501}
  (\bibinfo{year}{2022}).

\bibitem{otterbach2017unsupervised}
\bibinfo{author}{{Otterbach}, J.~S.} \emph{et~al.}
\newblock \bibinfo{title}{Unsupervised machine learning on a hybrid quantum
  computer}.
\newblock Preprint at \url{https://arxiv.org/abs/1712.05771} (\bibinfo{year}{2017}).

\bibitem{kerenidis2019q}
\bibinfo{author}{Kerenidis, I.}, \bibinfo{author}{Landman, J.},
  \bibinfo{author}{Luongo, A.} \& \bibinfo{author}{Prakash, A.}
\newblock \bibinfo{title}{q-means: A quantum algorithm for unsupervised machine
  learning}.
\newblock \emph{\bibinfo{journal}{Adv. Neural Inf. Process. Sys.}} \textbf{\bibinfo{volume}{32}}, (\bibinfo{year}{2019}).

\bibitem{saggio2021experimental}
\bibinfo{author}{{Saggio}, V.} \emph{et~al.}
\newblock \bibinfo{title}{Experimental quantum speed-up in reinforcement
  learning agents}.
\newblock \emph{\bibinfo{journal}{Nature}} \textbf{\bibinfo{volume}{591}},
  \bibinfo{pages}{229--233} (\bibinfo{year}{2021}).

\bibitem{skolik2021quantum}
\bibinfo{author}{Skolik, A.}, \bibinfo{author}{Jerbi, S.} \&
  \bibinfo{author}{Dunjko, V.}
\newblock \bibinfo{title}{Quantum agents in the gym: a variational quantum
  algorithm for deep q-learning}.
\newblock Preprint at \url{https://arxiv.org/abs/2103.15084} (\bibinfo{year}{2021}).

\bibitem{dallaire2018quantum}
\bibinfo{author}{Dallaire-Demers, P.-L.} \& \bibinfo{author}{Killoran, N.}
\newblock \bibinfo{title}{Quantum generative adversarial networks}.
\newblock \emph{\bibinfo{journal}{Phys. Rev. A}}
  \textbf{\bibinfo{volume}{98}}, \bibinfo{pages}{012324}
  (\bibinfo{year}{2018}).

\bibitem{benedetti2019generative}
\bibinfo{author}{Benedetti, M.} \emph{et~al.}
\newblock \bibinfo{title}{A generative modeling approach for benchmarking and
  training shallow quantum circuits}.
\newblock \emph{\bibinfo{journal}{Npj Quantum Inf.}}
  \textbf{\bibinfo{volume}{5}}, \bibinfo{pages}{45} (\bibinfo{year}{2019}).

\bibitem{kieferova2021quantum}
\bibinfo{author}{Kieferova, M.}, \bibinfo{author}{Carlos, O.~M.} \&
  \bibinfo{author}{Wiebe, N.}
\newblock \bibinfo{title}{Quantum generative training using r\'{e}nyi
  divergences}.
\newblock Preprint at \url{https://arxiv.org/abs/2106.09567} (\bibinfo{year}{2021}).

\bibitem{romero2021variational}
\bibinfo{author}{Romero, J.} \& \bibinfo{author}{Aspuru-Guzik, A.}
\newblock \bibinfo{title}{Variational quantum generators: Generative
  adversarial quantum machine learning for continuous distributions}.
\newblock \emph{\bibinfo{journal}{Adv. Quantum Technol.}}
  \textbf{\bibinfo{volume}{4}}, \bibinfo{pages}{2000003}
  (\bibinfo{year}{2021}).

\bibitem{bharti2021noisy}
\bibinfo{author}{Bharti, K.} \emph{et~al.}
\newblock \bibinfo{title}{Noisy intermediate-scale quantum algorithms}.
\newblock \emph{\bibinfo{journal}{Rev. Mod. Phys.}}
  \textbf{\bibinfo{volume}{94}}, \bibinfo{pages}{015004}
  (\bibinfo{year}{2022}).

\bibitem{havlivcek2019supervised}
\bibinfo{author}{Havl{\'\i}{\v{c}}ek, V.} \emph{et~al.}
\newblock \bibinfo{title}{Supervised learning with quantum-enhanced feature
  spaces}.
\newblock \emph{\bibinfo{journal}{Nature}} \textbf{\bibinfo{volume}{567}},
  \bibinfo{pages}{209--212} (\bibinfo{year}{2019}).

\bibitem{schuld2018supervised}
\bibinfo{author}{Schuld, M.} \& \bibinfo{author}{Petruccione, F.}
\newblock \emph{\bibinfo{title}{Supervised learning with quantum computers}},
  vol.~\bibinfo{volume}{17} (\bibinfo{publisher}{Springer},
  \bibinfo{year}{2018}).

\bibitem{schatzki2021entangled}
\bibinfo{author}{Schatzki, L.}, \bibinfo{author}{Arrasmith, A.},
  \bibinfo{author}{Coles, P.~J.} \& \bibinfo{author}{Cerezo, M.}
\newblock \bibinfo{title}{Entangled datasets for quantum machine learning}.
\newblock Preprint at \url{https://arxiv.org/abs/2109.03400} (\bibinfo{year}{2021}).

\bibitem{janocha2017loss}
\bibinfo{author}{Janocha, K.} \& \bibinfo{author}{Czarnecki, W.~M.}
\newblock \bibinfo{title}{On loss functions for deep neural networks in
  classification}.
\newblock Preprint at \url{https://arxiv.org/abs/1702.05659} (\bibinfo{year}{2017}).

\bibitem{grzesiak2020efficient}
\bibinfo{author}{Grzesiak, N.} \emph{et~al.}
\newblock \bibinfo{title}{Efficient arbitrary simultaneously entangling gates
  on a trapped-ion quantum computer}.
\newblock \emph{\bibinfo{journal}{Nat. Commun.}}
  \textbf{\bibinfo{volume}{11}}, \bibinfo{pages}{2963} (\bibinfo{year}{2020}).

\bibitem{pino2021demonstration}
\bibinfo{author}{Pino, J.~M.} \emph{et~al.}
\newblock \bibinfo{title}{Demonstration of the trapped-ion quantum ccd computer
  architecture}.
\newblock \emph{\bibinfo{journal}{Nature}} \textbf{\bibinfo{volume}{592}},
  \bibinfo{pages}{209--213} (\bibinfo{year}{2021}).

\bibitem{bluvstein2022quantum}
\bibinfo{author}{Bluvstein, D.} \emph{et~al.}
\newblock \bibinfo{title}{A quantum processor based on coherent transport of
  entangled atom arrays}.
\newblock \emph{\bibinfo{journal}{Nature}} \textbf{\bibinfo{volume}{604}},
  \bibinfo{pages}{451--456} (\bibinfo{year}{2022}).

\bibitem{pedrozo2020entanglement}
\bibinfo{author}{Pedrozo-Pe{\~n}afiel, E.} \emph{et~al.}
\newblock \bibinfo{title}{Entanglement on an optical atomic-clock transition}.
\newblock \emph{\bibinfo{journal}{Nature}} \textbf{\bibinfo{volume}{588}},
  \bibinfo{pages}{414--418} (\bibinfo{year}{2020}).

\bibitem{marciniak2022optimal}
\bibinfo{author}{Marciniak, C.~D.} \emph{et~al.}
\newblock \bibinfo{title}{Optimal metrology with programmable quantum sensors}.
\newblock \emph{\bibinfo{journal}{Nature}} \textbf{\bibinfo{volume}{603}},
  \bibinfo{pages}{604--609} (\bibinfo{year}{2022}).

\bibitem{kitagawa1993squeezed}
\bibinfo{author}{Kitagawa, M.} \& \bibinfo{author}{Ueda, M.}
\newblock \bibinfo{title}{Squeezed spin states}.
\newblock \emph{\bibinfo{journal}{Phys. Rev. A}}
  \textbf{\bibinfo{volume}{47}}, \bibinfo{pages}{5138--5143}
  (\bibinfo{year}{1993}).

\bibitem{wineland1992spin}
\bibinfo{author}{Wineland, D.~J.}, \bibinfo{author}{Bollinger, J.~J.},
  \bibinfo{author}{Itano, W.~M.}, \bibinfo{author}{Moore, F.} \&
  \bibinfo{author}{Heinzen, D.~J.}
\newblock \bibinfo{title}{Spin squeezing and reduced quantum noise in
  spectroscopy}.
\newblock \emph{\bibinfo{journal}{Phys. Rev. A}}
  \textbf{\bibinfo{volume}{46}}, \bibinfo{pages}{R6797} (\bibinfo{year}{1992}).

\bibitem{albertini2018controllability}
\bibinfo{author}{Albertini, F.} \& \bibinfo{author}{D’Alessandro, D.}
\newblock \bibinfo{title}{Controllability of symmetric spin networks}.
\newblock \emph{\bibinfo{journal}{J. Math. Phys.}}
  \textbf{\bibinfo{volume}{59}}, \bibinfo{pages}{052102}
  (\bibinfo{year}{2018}).

\bibitem{grant2019initialization}
\bibinfo{author}{Grant, E.}, \bibinfo{author}{Wossnig, L.},
  \bibinfo{author}{Ostaszewski, M.} \& \bibinfo{author}{Benedetti, M.}
\newblock \bibinfo{title}{An initialization strategy for addressing barren
  plateaus in parametrized quantum circuits}.
\newblock \emph{\bibinfo{journal}{Quantum}} \textbf{\bibinfo{volume}{3}},
  \bibinfo{pages}{214}; 10.22331/q-2019-12-09-214 (\bibinfo{year}{2019}).

\bibitem{skolik2020layerwise}
\bibinfo{author}{Skolik, A.}, \bibinfo{author}{McClean, J.~R.},
  \bibinfo{author}{Mohseni, M.}, \bibinfo{author}{van~der Smagt, P.} \&
  \bibinfo{author}{Leib, M.}
\newblock \bibinfo{title}{Layerwise learning for quantum neural networks}.
\newblock \emph{\bibinfo{journal}{Quantum Mach. Intell.}}
  \textbf{\bibinfo{volume}{3}}, \bibinfo{pages}{1--11} (\bibinfo{year}{2021}).

\bibitem{sauvage2021flip}
\bibinfo{author}{Sauvage, F.} \emph{et~al.}
\newblock \bibinfo{title}{Flip: A flexible initializer for arbitrarily-sized
  parametrized quantum circuits}.
\newblock Preprint at \url{https://arxiv.org/abs/2103.08572} (\bibinfo{year}{2021}).

\bibitem{sack2022avoiding}
\bibinfo{author}{Sack, S.~H.}, \bibinfo{author}{Medina, R.~A.},
  \bibinfo{author}{Michailidis, A.~A.}, \bibinfo{author}{Kueng, R.} \&
  \bibinfo{author}{Serbyn, M.}
\newblock \bibinfo{title}{Avoiding barren plateaus using classical shadows}.
\newblock \emph{\bibinfo{journal}{PRX Quantum}} \textbf{\bibinfo{volume}{3}},
  \bibinfo{pages}{020365} (\bibinfo{year}{2022}).

\bibitem{rad2022surviving}
\bibinfo{author}{Rad, A.}, \bibinfo{author}{Seif, A.} \&
  \bibinfo{author}{Linke, N.~M.}
\newblock \bibinfo{title}{Surviving the barren plateau in variational quantum
  circuits with bayesian learning initialization}.
\newblock Preprint at \url{https://arxiv.org/abs/2203.02464} (\bibinfo{year}{2022}).

\bibitem{broers2021optimization}
\bibinfo{author}{Broers, L.} \& \bibinfo{author}{Mathey, L.}
\newblock \bibinfo{title}{Optimization of quantum algorithm protocols without
  barren plateaus}.
\newblock Preprint at \url{https://arxiv.org/abs/2111.08085} (\bibinfo{year}{2021}).

\bibitem{liu2021parameter}
\bibinfo{author}{Liu, H.-Y.}, \bibinfo{author}{Sun, T.-P.},
  \bibinfo{author}{Wu, Y.-C.}, \bibinfo{author}{Han, Y.-J.} \&
  \bibinfo{author}{Guo, G.-P.}
\newblock \bibinfo{title}{A parameter initialization method for variational
  quantum algorithms to mitigate barren plateaus based on transfer learning}.
\newblock Preprint at url{https://arxiv.org/abs/2112.10952} (\bibinfo{year}{2021}).

\bibitem{friedrich2022avoiding}
\bibinfo{author}{Friedrich, L.} \& \bibinfo{author}{Maziero, J.}
\newblock \bibinfo{title}{Avoiding barren plateaus with classical deep neural
  networks}.
\newblock \emph{\bibinfo{journal}{Phys. Rev. A}}
  \textbf{\bibinfo{volume}{106}}, \bibinfo{pages}{042433}
  (\bibinfo{year}{2022}).

\bibitem{kulshrestha2022beinit}
\bibinfo{author}{Kulshrestha, A.} \& \bibinfo{author}{Safro, I.}
\newblock \bibinfo{title}{Beinit: Avoiding barren plateaus in variational
  quantum algorithms}.
\newblock In \emph{\bibinfo{booktitle}{2022 IEEE International Conference on
  Quantum Computing and Engineering (QCE)}}, \bibinfo{pages}{197--203}
  (\bibinfo{organization}{IEEE}, \bibinfo{year}{2022}).

\bibitem{mele2022avoiding}
\bibinfo{author}{Mele, A.~A.}, \bibinfo{author}{Mbeng, G.~B.},
  \bibinfo{author}{Santoro, G.~E.}, \bibinfo{author}{Collura, M.} \&
  \bibinfo{author}{Torta, P.}
\newblock \bibinfo{title}{Avoiding barren plateaus via transferability of
  smooth solutions in {H}amiltonian variational ansatz}.
\newblock Preprint at \url{https://arxiv.org/abs/2206.01982} (\bibinfo{year}{2022}).

\bibitem{zhang2022gaussian}
\bibinfo{author}{Zhang, K.}, \bibinfo{author}{Hsieh, M.-H.},
  \bibinfo{author}{Liu, L.} \& \bibinfo{author}{Tao, D.}
\newblock \bibinfo{title}{Gaussian initializations help deep variational
  quantum circuits escape from the barren plateau}.
\newblock Preprint at \url{https://arxiv.org/abs/2203.09376} (\bibinfo{year}{2022}).

\bibitem{grimsley2022adapt}
\bibinfo{author}{Grimsley, H.~R.}, \bibinfo{author}{Mayhall, N.~J.},
  \bibinfo{author}{Barron, G.~S.}, \bibinfo{author}{Barnes, E.} \&
  \bibinfo{author}{Economou, S.~E.}
\newblock \bibinfo{title}{Adaptive, problem-tailored variational quantum
  eigensolver mitigates rough parameter landscapes and barren plateaus}.
\newblock \emph{\bibinfo{journal}{Npj Quantum Inf.}}
  \textbf{\bibinfo{volume}{9}}, \bibinfo{pages}{19} (\bibinfo{year}{2023}).

\bibitem{cerezo2020variational}
\bibinfo{author}{Cerezo, M.}, \bibinfo{author}{Sharma, K.},
  \bibinfo{author}{Arrasmith, A.} \& \bibinfo{author}{Coles, P.~J.}
\newblock \bibinfo{title}{Variational quantum state eigensolver}.
\newblock \emph{\bibinfo{journal}{Npj Quantum Inf.}}
  \textbf{\bibinfo{volume}{8}}, \bibinfo{pages}{1--11} (\bibinfo{year}{2022}).

\bibitem{pesah2020absence}
\bibinfo{author}{Pesah, A.} \emph{et~al.}
\newblock \bibinfo{title}{Absence of barren plateaus in quantum convolutional
  neural networks}.
\newblock \emph{\bibinfo{journal}{Phys. Rev. X}}
  \textbf{\bibinfo{volume}{11}}, \bibinfo{pages}{041011}
  (\bibinfo{year}{2021}).

\bibitem{liu2021presence}
\bibinfo{author}{Liu, Z.}, \bibinfo{author}{Yu, L.-W.}, \bibinfo{author}{Duan,
  L.-M.} \& \bibinfo{author}{Deng, D.-L.}
\newblock \bibinfo{title}{The presence and absence of barren plateaus in
  tensor-network based machine learning}.
\newblock \emph{\bibinfo{journal}{Phys. Rev. Lett.}}
  \textbf{\bibinfo{volume}{129}}, \bibinfo{pages}{270501}
  (\bibinfo{year}{2022}).

\bibitem{arrasmith2021equivalence}
\bibinfo{author}{Arrasmith, A.}, \bibinfo{author}{Holmes, Z.},
  \bibinfo{author}{Cerezo, M.} \& \bibinfo{author}{Coles, P.~J.}
\newblock \bibinfo{title}{Equivalence of quantum barren plateaus to cost
  concentration and narrow gorges}.
\newblock \emph{\bibinfo{journal}{Quantum Sci. Technol.}}
  \textbf{\bibinfo{volume}{7}}, \bibinfo{pages}{045015} (\bibinfo{year}{2022}).

\bibitem{zhang2021understanding}
\bibinfo{author}{Zhang, C.}, \bibinfo{author}{Bengio, S.},
  \bibinfo{author}{Hardt, M.}, \bibinfo{author}{Recht, B.} \&
  \bibinfo{author}{Vinyals, O.}
\newblock \bibinfo{title}{Understanding deep learning (still) requires
  rethinking generalization}.
\newblock \emph{\bibinfo{journal}{Communications of the ACM}}
  \textbf{\bibinfo{volume}{64}}, \bibinfo{pages}{107--115}
  (\bibinfo{year}{2021}).

\bibitem{allen2019convergence}
\bibinfo{author}{Allen-Zhu, Z.}, \bibinfo{author}{Li, Y.} \&
  \bibinfo{author}{Song, Z.}
\newblock \bibinfo{title}{A convergence theory for deep learning via
  over-parameterization}.
\newblock \emph{\bibinfo{journal}International Conference on Machine Learning} \textbf{\bibinfo{volume}{36}}, (\bibinfo{year}{2019}).

\bibitem{allen2019learning}
\bibinfo{author}{Allen-Zhu, Z.}, \bibinfo{author}{Li, Y.} \&
  \bibinfo{author}{Liang, Y.}
\newblock \bibinfo{title}{Learning and generalization in overparameterized
  neural networks, going beyond two layers}.
\newblock \emph{\bibinfo{journal}{Adv. Neural Inf. Process. Sys.}} \textbf{\bibinfo{volume}{33}} ,(\bibinfo{year}{2019}).

\bibitem{buhai2020empirical}
\bibinfo{author}{Buhai, R.-D.}, \bibinfo{author}{Halpern, Y.},
  \bibinfo{author}{Kim, Y.}, \bibinfo{author}{Risteski, A.} \&
  \bibinfo{author}{Sontag, D.}
\newblock \bibinfo{title}{Empirical study of the benefits of
  overparameterization in learning latent variable models}.
\newblock \emph{\bibinfo{journal}International Conference on Machine Learning} \textbf{\bibinfo{volume}{37}}, (\bibinfo{year}{2020}).

\bibitem{banchi2021generalization}
\bibinfo{author}{Banchi, L.}, \bibinfo{author}{Pereira, J.} \&
  \bibinfo{author}{Pirandola, S.}
\newblock \bibinfo{title}{Generalization in quantum machine learning: A quantum
  information standpoint}.
\newblock \emph{\bibinfo{journal}{PRX Quantum}} \textbf{\bibinfo{volume}{2}},
  \bibinfo{pages}{040321} (\bibinfo{year}{2021}).

\bibitem{caro2022outofdistribution}
\bibinfo{author}{Caro, M.~C.} \emph{et~al.}
\newblock \bibinfo{title}{Out-of-distribution generalization for learning
  quantum dynamics}.
\newblock Preprint at \url{https://arxiv.org/abs/2204.10268} (\bibinfo{year}{2022}).

\bibitem{du2021efficient}
\bibinfo{author}{Du, Y.}, \bibinfo{author}{Tu, Z.}, \bibinfo{author}{Yuan, X.}
  \& \bibinfo{author}{Tao, D.}
\newblock \bibinfo{title}{Efficient measure for the expressivity of variational
  quantum algorithms}.
\newblock \emph{\bibinfo{journal}{Phys. Rev. Lett.}}
  \textbf{\bibinfo{volume}{128}}, \bibinfo{pages}{080506}
  (\bibinfo{year}{2022}).

\bibitem{huang2021power}
\bibinfo{author}{Huang, H.-Y.} \emph{et~al.}
\newblock \bibinfo{title}{Power of data in quantum machine learning}.
\newblock \emph{\bibinfo{journal}{Nat. Commun.}}
  \textbf{\bibinfo{volume}{12}}, \bibinfo{pages}{1--9} (\bibinfo{year}{2021}).

\bibitem{shalev2014understanding}
\bibinfo{author}{Shalev-Shwartz, S.} \& \bibinfo{author}{Ben-David, S.}
\newblock \emph{\bibinfo{title}{Understanding machine learning: From theory to
  algorithms}} (\bibinfo{publisher}{Cambridge university press},
  \bibinfo{year}{2014}).

\bibitem{Hajek2021Statistical}
\bibinfo{author}{Hajek, B.} \& \bibinfo{author}{Raginsky, M.}
\newblock \bibinfo{title}{Ece 543: Statistical learning theory}
  (\bibinfo{year}{2021}).
\newblock \urlprefix\url{http://maxim.ece.illinois.edu/teaching/SLT/}.

\bibitem{thanasilp2022exponential}
\bibinfo{author}{Thanasilp, S.}, \bibinfo{author}{Wang, S.},
  \bibinfo{author}{Cerezo, M.} \& \bibinfo{author}{Holmes, Z.}
\newblock \bibinfo{title}{Exponential concentration and untrainability in
  quantum kernel methods}.
\newblock Preprint at \url{https://arxiv.org/abs/2208.11060} (\bibinfo{year}{2022}).

\bibitem{erdds1959random}
\bibinfo{author}{Erdos, P.} \& \bibinfo{author}{Renyi, A.}
\newblock \bibinfo{title}{On random graphs i}.
\newblock \emph{\bibinfo{journal}{Publ. M.ath. Debrecen}}
  \textbf{\bibinfo{volume}{6}}, \bibinfo{pages}{18} (\bibinfo{year}{1959}).

\bibitem{raussendorf2003measurement}
\bibinfo{author}{Raussendorf, R.}, \bibinfo{author}{Browne, D.~E.} \&
  \bibinfo{author}{Briegel, H.~J.}
\newblock \bibinfo{title}{Measurement-based quantum computation on cluster
  states}.
\newblock \emph{\bibinfo{journal}{Phys. Rev. A}}
  \textbf{\bibinfo{volume}{68}}, \bibinfo{pages}{022312}
  (\bibinfo{year}{2003}).

\bibitem{hein2004multiparty}
\bibinfo{author}{Hein, M.}, \bibinfo{author}{Eisert, J.} \&
  \bibinfo{author}{Briegel, H.~J.}
\newblock \bibinfo{title}{Multiparty entanglement in graph states}.
\newblock \emph{\bibinfo{journal}{Phys. Rev. A}}
  \textbf{\bibinfo{volume}{69}}, \bibinfo{pages}{062311}
  (\bibinfo{year}{2004}).

\bibitem{cheng2010quantum}
\bibinfo{author}{Cheng, R.}
\newblock \bibinfo{title}{Quantum geometric tensor (fubini-study metric) in
  simple quantum system: A pedagogical introduction}.
\newblock Preprint at \url{https://arxiv.org/abs/1012.1337} (\bibinfo{year}{2010}).

\bibitem{meyer2021fisher}
\bibinfo{author}{Meyer, J.~J.}
\newblock \bibinfo{title}{Fisher {I}nformation in {N}oisy
  {I}ntermediate-{S}cale {Q}uantum {A}pplications}.
\newblock \emph{\bibinfo{journal}{{Quantum}}} \textbf{\bibinfo{volume}{5}},
  \bibinfo{pages}{539}; 10.22331/q-2021-09-09-539 (\bibinfo{year}{2021}).

\bibitem{larocca2020exploiting}
\bibinfo{author}{Larocca, M.}, \bibinfo{author}{Calzetta, E.} \&
  \bibinfo{author}{Wisniacki, D.~A.}
\newblock \bibinfo{title}{Exploiting landscape geometry to enhance quantum
  optimal control}.
\newblock \emph{\bibinfo{journal}{Phys. Rev. A}}
  \textbf{\bibinfo{volume}{101}}, \bibinfo{pages}{023410}
  (\bibinfo{year}{2020}).

\bibitem{larocca2020fourier}
\bibinfo{author}{Larocca, M.}, \bibinfo{author}{Calzetta, E.} \&
  \bibinfo{author}{Wisniacki, D.}
\newblock \bibinfo{title}{Fourier compression: A customization method for
  quantum control protocols}.
\newblock \emph{\bibinfo{journal}{Phys. Rev. A}}
  \textbf{\bibinfo{volume}{102}}, \bibinfo{pages}{033108}
  (\bibinfo{year}{2020}).

\bibitem{zhu1997algorithm}
\bibinfo{author}{Zhu, C.}, \bibinfo{author}{Byrd, R.~H.}, \bibinfo{author}{Lu,
  P.} \& \bibinfo{author}{Nocedal, J.}
\newblock \bibinfo{title}{Algorithm 778: L-bfgs-b: Fortran subroutines for
  large-scale bound-constrained optimization}.
\newblock \emph{\bibinfo{journal}{ACM Transactions on mathematical software
  (TOMS)}} \textbf{\bibinfo{volume}{23}}, \bibinfo{pages}{550--560}
  (\bibinfo{year}{1997}).

\bibitem{nakaji2020expressibility}
\bibinfo{author}{Nakaji, K.} \& \bibinfo{author}{Yamamoto, N.}
\newblock \bibinfo{title}{Expressibility of the alternating layered ansatz for
  quantum computation}.
\newblock \emph{\bibinfo{journal}{Quantum}} \textbf{\bibinfo{volume}{5}},
  \bibinfo{pages}{434}; 10.22331/q-2021-04-19-434 (\bibinfo{year}{2021}).

\bibitem{anschuetz2022efficient}
\bibinfo{author}{Anschuetz, E.~R.}, \bibinfo{author}{Bauer, A.},
  \bibinfo{author}{Kiani, B.~T.} \& \bibinfo{author}{Lloyd, S.}
\newblock \bibinfo{title}{Efficient classical algorithms for simulating
  symmetric quantum systems}.
\newblock Preprint at \url{https://arxiv.org/abs/2211.16998} (\bibinfo{year}{2022}).

\bibitem{kandala2017hardware}
\bibinfo{author}{Kandala, A.} \emph{et~al.}
\newblock \bibinfo{title}{Hardware-efficient variational quantum eigensolver
  for small molecules and quantum magnets}.
\newblock \emph{\bibinfo{journal}{Nature}} \textbf{\bibinfo{volume}{549}},
  \bibinfo{pages}{242--246} (\bibinfo{year}{2017}).

\bibitem{cincio2018learning}
\bibinfo{author}{Cincio, L.}, \bibinfo{author}{Suba{\c{s}}{\i}, Y.},
  \bibinfo{author}{Sornborger, A.~T.} \& \bibinfo{author}{Coles, P.~J.}
\newblock \bibinfo{title}{Learning the quantum algorithm for state overlap}.
\newblock \emph{\bibinfo{journal}{New J. Physical.}}
  \textbf{\bibinfo{volume}{20}}, \bibinfo{pages}{113022}
  (\bibinfo{year}{2018}).

\bibitem{perez2020data}
\bibinfo{author}{P{\'e}rez-Salinas, A.}, \bibinfo{author}{Cervera-Lierta, A.},
  \bibinfo{author}{Gil-Fuster, E.} \& \bibinfo{author}{Latorre, J.~I.}
\newblock \bibinfo{title}{Data re-uploading for a universal quantum
  classifier}.
\newblock \emph{\bibinfo{journal}{Quantum}} \textbf{\bibinfo{volume}{4}},
  \bibinfo{pages}{226}; 10.22331/q-2020-02-06-226 (\bibinfo{year}{2020}).

\bibitem{serre1977linear}
\bibinfo{author}{Serre, J.-P.} \emph{et~al.}
\newblock \emph{\bibinfo{title}{Linear representations of finite groups}},
  vol.~\bibinfo{volume}{42} (\bibinfo{publisher}{Springer},
  \bibinfo{year}{1977}).

\bibitem{fulton1991representation}
\bibinfo{author}{Fulton, W.} \& \bibinfo{author}{Harris, J.}
\newblock \emph{\bibinfo{title}{Representation Theory: A First Course}}
  (\bibinfo{publisher}{Springer}, \bibinfo{year}{1991}).

\bibitem{sagan2001symmetric}
\bibinfo{author}{Sagan, B.}
\newblock \emph{\bibinfo{title}{The symmetric group: representations,
  combinatorial algorithms, and symmetric functions}}, vol.
  \bibinfo{volume}{203} (\bibinfo{publisher}{Springer Science \& Business
  Media}, \bibinfo{year}{2001}).

\bibitem{knapp2001representation}
\bibinfo{author}{Knapp, A.~W.}
\newblock \emph{\bibinfo{title}{Representation theory of semisimple groups: an
  overview based on examples}} (\bibinfo{publisher}{Princeton university
  press}, \bibinfo{address}{Princeton}, \bibinfo{year}{2001}).

\bibitem{goodman2009symmetry}
\bibinfo{author}{Goodman, R.} \& \bibinfo{author}{Wallach, N.~R.}
\newblock \emph{\bibinfo{title}{Symmetry, representations, and invariants}},
  vol. \bibinfo{volume}{255} (\bibinfo{publisher}{Springer},
  \bibinfo{year}{2009}).

\bibitem{zeier2011symmetry}
\bibinfo{author}{Zeier, R.} \& \bibinfo{author}{Schulte-Herbr{\"u}ggen, T.}
\newblock \bibinfo{title}{Symmetry principles in quantum systems theory}.
\newblock \emph{\bibinfo{journal}{J. Math. Phys.}}
  \textbf{\bibinfo{volume}{52}}, \bibinfo{pages}{113510}
  (\bibinfo{year}{2011}).

\bibitem{kazi2023universality}
\bibinfo{author}{Kazi, S.}, \bibinfo{author}{Larocca, M.} \&
  \bibinfo{author}{Cerezo, M.}
\newblock \bibinfo{title}{On the universality of $s_n$-equivariant $ k $-body
  gates}.
\newblock Preprint at \url{https://arxiv.org/abs/2303.00728} (\bibinfo{year}{2023}).

\bibitem{ragone2023unified}
\bibinfo{author}{Ragone, M.} \emph{et~al.}
\newblock \bibinfo{title}{A unified theory of barren plateaus for deep
  parametrized quantum circuits}.
\newblock Preprint at url{https://arxiv.org/abs/2309.09342} (\bibinfo{year}{2023}).

\bibitem{puchala2017symbolic}
\bibinfo{author}{Puchala, Z.} \& \bibinfo{author}{Miszczak, J.~A.}
\newblock \bibinfo{title}{Symbolic integration with respect to the haar measure
  on the unitary groups}.
\newblock \emph{\bibinfo{journal}{Bulletin of the Polish Academy of Sciences
  Technical Sciences}} \textbf{\bibinfo{volume}{65}}, \bibinfo{pages}{21--27}
  (\bibinfo{year}{2017}).

\bibitem{zhang2021quantum}
\bibinfo{author}{Zhang, B.}, \bibinfo{author}{Sone, A.} \&
  \bibinfo{author}{Zhuang, Q.}
\newblock \bibinfo{title}{Quantum computational phase transition in
  combinatorial problems}.
\newblock \emph{\bibinfo{journal}{Npj Quantum Inf.}}
  \textbf{\bibinfo{volume}{8}}, \bibinfo{pages}{1--11} (\bibinfo{year}{2022}).

\bibitem{fontana2023theadjoint}
\bibinfo{author}{Fontana, E.} \emph{et~al.}
\newblock \bibinfo{title}{The adjoint is all you need: Characterizing barren
  plateaus in quantum ans\"atze}.
\newblock Preprint at url{https://arxiv.org/abs/2309.07902} (\bibinfo{year}{2023}).

\bibitem{fan2022understanding}
\bibinfo{author}{Fan, J.}, \bibinfo{author}{Yang, Z.} \& \bibinfo{author}{Yu,
  M.}
\newblock \bibinfo{title}{Understanding implicit regularization in
  over-parameterized single index model}.
\newblock \emph{\bibinfo{journal}{J. Am. Stat. Assoc.}} \bibinfo{pages}{1--14} (\bibinfo{year}{2022}).

\bibitem{du2018gradient}
\bibinfo{author}{Du, S.~S.}, \bibinfo{author}{Zhai, X.},
  \bibinfo{author}{Poczos, B.} \& \bibinfo{author}{Singh, A.}
\newblock \bibinfo{title}{Gradient descent provably optimizes
  over-parameterized neural networks}.
\newblock In \emph{\bibinfo{booktitle}{International Conference on Learning
  Representations}} (\bibinfo{year}{2019}).

\bibitem{brutzkus2018sgd}
\bibinfo{author}{Brutzkus, A.}, \bibinfo{author}{Globerson, A.},
  \bibinfo{author}{Malach, E.} \& \bibinfo{author}{Shalev-Shwartz, S.}
\newblock \bibinfo{title}{{SGD} learns over-parameterized networks that
  provably generalize on linearly separable data}.
\newblock In \emph{\bibinfo{booktitle}{International Conference on Learning
  Representations}} (\bibinfo{year}{2018}).

\bibitem{bartlett2021deep}
\bibinfo{author}{Bartlett, P.~L.}, \bibinfo{author}{Montanari, A.} \&
  \bibinfo{author}{Rakhlin, A.}
\newblock \bibinfo{title}{Deep learning: a statistical viewpoint}.
\newblock \emph{\bibinfo{journal}{Acta numerica}}
  \textbf{\bibinfo{volume}{30}}, \bibinfo{pages}{87--201}
  (\bibinfo{year}{2021}).

\bibitem{fukumizu1996regularity}
\bibinfo{author}{Fukumizu, K.}
\newblock \bibinfo{title}{A regularity condition of the information matrix of a
  multilayer perceptron network}.
\newblock \emph{\bibinfo{journal}{Neural networks}}
  \textbf{\bibinfo{volume}{9}}, \bibinfo{pages}{871--879}
  (\bibinfo{year}{1996}).


\bibitem{liu1998overparameterization}
\bibinfo{author}{Liu, M.} \& \bibinfo{author}{Zhang, H.~H.}
\newblock \bibinfo{title}{Overparameterization in the seminonparametric density
  estimation}.
\newblock \emph{\bibinfo{journal}{Econ. Lett.}}
  \textbf{\bibinfo{volume}{60}}, \bibinfo{pages}{11--18}
  (\bibinfo{year}{1998}).

\bibitem{roychowdhury2017reducing}
\bibinfo{author}{RoyChowdhury, A.}, \bibinfo{author}{Sharma, P.},
  \bibinfo{author}{Learned-Miller, E.} \& \bibinfo{author}{Roy, A.}
\newblock \bibinfo{title}{Reducing duplicate filters in deep neural networks}.
\newblock \emph{\bibinfo{journal}{Adv. Neural. Inf. Process. Sys. Workshop on Deep Learning}} \textbf{\bibinfo{volume}{1}}, (\bibinfo{year}{2017}).

\bibitem{vershynin2018highdimensional}
\bibinfo{author}{Vershynin, R.}
\newblock \emph{\bibinfo{title}{High-Dimensional Probability: An Introduction
  with Applications in Data Science}} (\bibinfo{publisher}{Cambridge University
  Press}, \bibinfo{year}{2018}).

\bibitem{dudley1999uniform}
\bibinfo{author}{Dudley, R.~M.}
\newblock \emph{\bibinfo{title}{Uniform Central Limit Theorems}}
  (\bibinfo{publisher}{Cambridge University Press}, \bibinfo{year}{1999}).

\bibitem{harrow2013church}
\bibinfo{author}{Harrow, A.~W.}
\newblock \bibinfo{title}{The church of the symmetric subspace}.
\newblock Preprint at \url{https://arxiv.org/abs/1308.6595} (\bibinfo{year}{2013}).

\end{thebibliography}
\end{document}


\title{Supplementary Information for 'Theoretical Guarantees for Permutation-Equivariant Quantum Neural Networks'}

\author{Louis Schatzki}
\email{louisms2@illinois.edu}
\affiliation{Information Sciences, Los Alamos National Laboratory, Los Alamos, New Mexico 87545, USA}
\affiliation{Electrical and Computer Engineering, University of Illinois Urbana-Champaign, Urbana, Illinois, 61801, USA}

\author{Mart\'{i}n Larocca}
\thanks{The two first authors contributed equally.}
\affiliation{Theoretical Division, Los Alamos National Laboratory, Los Alamos, New Mexico 87545, USA}
\affiliation{Center for Nonlinear Studies, Los Alamos National Laboratory, Los Alamos, New Mexico 87545, USA}

\author{Quynh T. Nguyen}
\affiliation{Theoretical Division, Los Alamos National Laboratory, Los Alamos, New Mexico 87545, USA}
\affiliation{Harvard Quantum Initiative, Harvard University, Cambridge, Massachusetts 02138, USA}

\author{Fr\'{e}d\'{e}ric Sauvage}
\affiliation{Theoretical Division, Los Alamos National Laboratory, Los Alamos, New Mexico 87545, USA}

\author{M. Cerezo}
\email{cerezo@lanl.gov} 
\affiliation{Information Sciences, Los Alamos National Laboratory, Los Alamos, New Mexico 87545, USA}

\maketitle
\onecolumngrid
\setcounter{theorem}{0}
\setcounter{corollary}{0}

In this supplementary information we present additional details, as well as proofs of our main results. We start by presenting preliminary results regarding Haar integration in Sec.~\ref{ap:sec:prel}. In Secs.~\ref{ap:sec:lem1} and~\ref{ap:sec:lem2} we present a proof for Lemmas~\ref{lem:inv} and~\ref{lem:poly}, respectively. Section~\ref{ap:sec:the1} contains a proof of Theorem~\ref{thm:main1}, while we prove Theorem~\ref{thm:main2} in Sec.~\ref{si_sec_thm2}, with a generalization to $k$-Local Pauli Strings presented in Sec.~\ref{ap:sec:gen}. In Sec.~\ref{app:states} we discuss the trainability of different input states. Then, we elaborate on the VC dimension of QML models in Sec.~\ref{app:vc}. Additional numerics are reported in Secs.~\ref{app:numerics} and~\ref{ap:sec:comparison}.

\section{Supplementary Methods 1: Haar Integration}\label{ap:sec:prel}

Since the derivation of our main results requires computing expectation values of Haar distributed unitaries, we here recall some fundamental results that will be used below. First, let us denote as $d\mu(V)$ the volume element of the Haar measure, where $V\in \mathbb{U}(d)$, with  $\mathbb{U}(d)$ being the unitary group of degree $d$. Using the Weingarten calculus~\cite{puchala2017symbolic}, we can prove the following identities:
\begin{align}
    \int_{\mathbb{U}(d)} d\mu(V) V A V\ad =&\frac{\Tr[A]}{d}\id\,,\label{eq:int_1}\\
    \int_{\mathbb{U}(d)} d\mu(V) V^{\otimes 2} B (V\ad)^{\otimes 2}=&\frac{1}{d^2-1}\left(\Tr[B]-\frac{\Tr[B\,\, \SWAP]}{d}\right)\id\otimes \id +\frac{1}{d^2-1}\left(\Tr[B\,\, \SWAP]-\frac{\Tr[B]}{d}\right) \SWAP\label{eq:int_2}\,,
\end{align}
where $A\in\BC(\HC)$ and $B\in\BC(\HC^{\otimes 2})$.
Here $\id$ denotes the $d\times d$ identity matrix, and $\SWAP$ is the operator acting $\HC^{\otimes 2}$  whose action is $SWAP\ket{i}\otimes\ket{j}=\ket{j}\otimes\ket{i}$ for any two vectors $\ket{i},\ket{j}\in\HC$.

From Eqs.~\eqref{eq:int_1} and~\eqref{eq:int_2} can can derive the following identities (see~\cite{cerezo2021cost} for a proof):
\footnotesize
\begin{align}
    \int_{\mathbb{U}(d)} d \mu(U) \operatorname{Tr}\left[U A U^{\dagger} B\right]&=\frac{\operatorname{Tr}[A] \operatorname{Tr}[B]}{d},\label{Eq_lemma1}\\
    \int_{\mathbb{U}(d)} d \mu(U) \operatorname{Tr}\left[U A U^{\dagger} B U C U^{\dagger} D\right] &=\frac{\operatorname{Tr}[A] \operatorname{Tr}[C] \operatorname{Tr}[B D]+\operatorname{Tr}[A C] \operatorname{Tr}[B] \operatorname{Tr}[D]}{d^{2}-1}  -\frac{\operatorname{Tr}[A C] \operatorname{Tr}[B D]+\operatorname{Tr}[A] \operatorname{Tr}[B] \operatorname{Tr}[C] \operatorname{Tr}[D]}{d\left(d^{2}-1\right)},\label{Eq_lemma2}\\
        \int_{\mathbb{U}(d)} d \mu(U) \operatorname{Tr}\left[U A U^{\dagger} B\right] \operatorname{Tr}\left[U C U^{\dagger} D\right]& = \frac{\operatorname{Tr}[A] \operatorname{Tr}[B] \operatorname{Tr}[C] \operatorname{Tr}[D]+\operatorname{Tr}[A C] +\operatorname{Tr}[B D]}{d^{2}-1}-\frac{\operatorname{\Tr}[A C] \operatorname{Tr}[B] \operatorname{Tr}[D]+\operatorname{Tr}[A] \operatorname{Tr}[C] \operatorname{Tr}[B D]}{d\left(d^{2}-1\right)},\label{Eq_lemma3}
\end{align}
\normalsize
where $A, B, C$, and $D$ are linear operators on $\HC$. 

\medskip

\section{Supplementary Methods 2: Proof of Lemma~\ref{lem:inv}}\label{ap:sec:lem1}

Here, we provide a proof for Lemma~\ref{lem:inv}, which we recall for convenience.
\begin{lemma}[Invariance from equivariance]\label{lem:inv}
A loss function of the form $l_{\thv}(\rho_i) = \Tr[\mathcal{U}_{\thv}(\rho_i)O]$ is $G$-invariant if its composed of $G$-equivariant QNN and measurement.
\end{lemma}

\begin{proof}
For any $g\in G$ we have
\begin{equation}
\begin{split}
    \ell_{\thv}( R(g)\rho R(g)\ad)&=\Tr[\UC_{\thv}(R(g)\rho R(g)\ad) O  ]\nonumber\\ &=\Tr[R(g)\UC_{\thv}(\rho) R(g)\ad O ]\nonumber\\
    &=\Tr[\UC_{\thv}(\rho) R(g)\ad O R(g)]\nonumber\\
    &= \Tr[ \UC_{\thv}(\rho) O]=\ell_{\thv}(\rho)\nonumber\,,
\end{split}
\end{equation}
where in the second line we have used the equivariance of the QNN, in the third line cyclicity of the trace,  and in the fourth the equivariance of the measurement.
\end{proof}

\section{Supplementary Methods 3: Proof of Lemma~\ref{lem:poly}}\label{ap:sec:lem2}

Here, we provide a proof for Lemma~\ref{lem:poly}, which characterizes the number of free parameters in $S_n$-equivariant unitaries.

\begin{lemma}[Dimension of $S_n$-equivariant unitaries]\label{lem:poly}
The submanifold of $S_n-$equivariant unitaries is of dimension equal to the Tetrahedral numbers ${\rm Te}_{n+1}=\binom{n+3}{3}$ (see Fig.4 in the main text), and therefore on the order of $\Theta(n^3)$.
\end{lemma}

\begin{proof}
As shown in Eq.~13 of the main text, any $S_n$-equivariant unitary $U(\thv)$ is fully characterized by its $d_\lm$-dimensional  isotypic components $U_\lm(\thv)$, each having $d_\lm^2$ free real-valued parameters. Thus, the unitary $U(\thv)$ has $\sum_\lm d_\lm^2$ free parameters, which can be shown (by induction)  to be equal to the Tetrahedral numbers ${\rm Te}_{n+1}$.
\end{proof}

\section{Supplementary Methods 4: Proof of Theorem~\ref{thm:main1}}\label{ap:sec:the1}
Here, we provide a proof for Theorem~\ref{thm:main1}, yielding an exact expression for the variance of the partial derivatives of the loss function 
\begin{equation}\label{eq:empirical-loss}
    \widehat{\LC}(\thv)=\sum_{i=1}^M c_i\ell_{\thv}(\rho_i)\,.
\end{equation}

\begin{theorem}\label{thm:main1}
     Let $\UC_{\thv}$ be an $S_n$-equivariant QNN, with  generators in $\GC$. Let $O$ be an $S_n$-equivariant measurement operator from $\MC$, and consider an empirical loss function $\hat{\LC}(\thv)$ with the form in Eq.~\ref{eq:empirical-loss}, for some given training set $S$. Assuming a circuit depth $L$ such that the QNN forms independent 2-designs on each isotypic block, we  have $\langle \partial_\mu \hat{\LC}(\boldsymbol{\theta}) \rangle_{\boldsymbol{\theta}} = 0$, and 
     \small
    \begin{equation}\label{eq:var_main-SI}
    \Var_{\thv}[ \partial_{\mu} \hat{\LC}(\thv)]\!=\!\sum_{\lambda}\frac{2d_\lambda}{(d_\lambda^2-1)^2}\Delta(H_{\mu,_\lambda})\Delta(O_{\lambda}) \Delta(\sum_{\nu=1}^{m_{\lambda}}\sigma_\lm^\nu)\,.
\end{equation}
\normalsize
Here, $\Delta(B) = \Tr[B^2]-\frac{\Tr[B]^2}{\dim(B)}$.
\end{theorem}

\begin{proof}
    
First, let us recall from the main text that under the action of an $S_n$-equivariant QNN, the Hilbert space decomposes as
\begin{equation}\label{eq:iso_space-SI}
    \mathcal{H} \cong \bigoplus_{\lm}\bigoplus_{\nu=1}^{m_\lambda} \HC_\lm^\nu,
\end{equation}
where we have
\begin{equation}
     d_\lambda=\dim(\HC_\lambda^\nu)\,.
\end{equation}
Note that, for a given $\lambda$ all the subspaces $\HC_\lambda^\nu$ have the same dimension irrespective of $\nu$.

Then, let us define as $Q_\lambda^\nu$ the $d_\lambda \times d$ matrix that results from horizontally stacking the basis elements of  $\HC_\lambda^\nu$ 
\begin{equation}
(Q_\lambda^\nu)\ad = \begin{bmatrix} \vdots &  \vdots  && \vdots \\ \ket{v_1^\nu}, & \ket{v_2^\nu},&  &, \ket{v_{d_\lambda}^\nu} \\ \vdots &  \vdots  && \vdots \end{bmatrix}, 
\end{equation}
such that $Q_\lambda^\nu$ maps vectors from $\HC$ to $\HC_\lambda^\nu$. These matrices satisfy the following property
\begin{equation}\label{eq:proj2}
Q_\lambda^\nu (Q_\eta^{\nu'})\ad = \id_{d_\lambda} \delta_{\lambda\eta}\delta_{\nu\nu'}\,,
\end{equation}
allowing us to define the projector $\mathbb{P}_\lambda^\nu$ onto each subspace as
\begin{equation}
    \mathbb{P}_\lambda^\nu=  (Q_\lambda^\nu)\ad Q_\lambda^\nu \,, \quad \text{such that}\quad \sum_{\lambda}\sum_{\nu} \mathbb{P}_\lambda^\nu=\id\,.
\end{equation}
In what follows we will denote as
\begin{equation}\label{Eq_red}
  A_\lambda^{\nu} = Q_\lambda^\nu A (Q_\lambda^\nu)\ad\,,
\end{equation}
the operator obtained by reducing $A$ onto the subspace labeled by $\lambda$ and $\nu$. In particular, we note  that if $A=A\ad$ then $(A_\lambda^\nu)\ad=A_\lambda^\nu$ (the converse is not necessarily true). Moreover, if $A$ is positive semi-definite, then so is $A_\lambda^\nu$. This can be seen from the fact that given a $d_\lambda$-dimensional vector $\ket{\vec{x}}$, then $\langle \vec{x}|A_\lambda^\nu|\vec{x}\rangle=\langle\widetilde{\vec{x}}|A|\widetilde{\vec{x}}\rangle$, where $|\widetilde{\vec{x}}\rangle$ is a $d$-dimensional vector where $\vec{x}$ is padded with zeros.

Here we also recall that both the QNN and measurement operator are block diagonal, such that
\begin{equation}\label{eq:commutator-sm}
    U(\thv) \cong \bigoplus_{\lm} \id_{m_\lm} \otimes U_\lm(\thv), \quad \text{and} \quad
    O \cong \bigoplus_{\lm} \id_{m_\lm} \otimes O_\l\,,
\end{equation}
or alternatively, such that
\begin{align}\label{eq:projections}
   U(\thv)=\sum_{\lambda}\sum_{\nu=1}^{m_{\lambda}} \mathbb{P}_\lambda^\nu U(\thv) \mathbb{P}_\lambda^\nu\,,  \quad \text{and} \quad  O=\sum_{\lambda}\sum_{\nu=1}^{m_{\lambda}} \mathbb{P}_\lambda^\nu O\mathbb{P}_\lambda^\nu \,.
\end{align}

We can now explicitly evaluate the loss function partial derivative with respect to the parameter $\th_{\mu}$ in the $\mu$-th layer. First, we recall that the loss function is 
\begin{equation}
    \ell_{\thv}(\rho_i,y_i)=\Tr[\UC_{\thv}(\rho_i) O_i]\,,
\end{equation}
with the $S_n$-equivariant quantum neural network unitary being
\begin{equation}
    U(\thv)=\prod_{l=1}^L\er^{-\ir \theta_l H_l}\,.
\end{equation}
For convenience, let us also introduce the following notation
\begin{equation}
U_{B} = \prod_{l=0}^\mu \er^{-\ir\theta_l H_l}\,,\quad \quad 
U_A = \prod_{l=\mu+1}^L \er^{-\ir\theta_l H_l}\,,
\end{equation}
such that $U(\thv)=U_AU_B$. Here we have omitted the explicit dependence on $\thv$ for simplicity of notation. 
An explicit calculation yields
\begin{align}
    \partial_{\mu} \ell_{\thv}(\rho_i) &=\Tr[(\partial_\mu U(\thv))\rho_i U\ad(\thv) O] + \Tr[ U(\thv)\rho_i (\partial_\mu U\ad(\thv)) O]
\nonumber\\
&=\ir\Tr[ U_B \rho_i U_B\ad [H_\mu,U_A\ad O U_A]]\,.
\end{align}
Using Eqs.~\eqref{eq:projections} we can see that 
\begin{align}
\Tr[ U_B \rho_i U_B\ad [H_\mu,U_A\ad O U_A]]&=\sum_{\vec{\lambda}}\sum_{\vec{\nu}}  \Tr[ \mathbb{P}_{\lambda_1}^{\nu_1} U_B \mathbb{P}_{\lambda_1}^{\nu_1} \rho_i \mathbb{P}_{\lambda_2}^{\nu_2} U_B\ad \mathbb{P}_{\lambda_2}^{\nu_2}[\mathbb{P}_{\lambda_3}^{\nu_3} H_\mu \mathbb{P}_{\lambda_3}^{\nu_3}, \mathbb{P}_{\lambda_4}^{\nu_4}  U_A\ad \mathbb{P}_{\lambda_4}^{\nu_4} \mathbb{P}_{\lambda_5}^{\nu_5} O \mathbb{P}_{\lambda_5}^{\nu_5}  \mathbb{P}_{\lambda_6}^{\nu_6} U_A \mathbb{P}_{\lambda_6}^{\nu_6}]]\nonumber\\
&=\sum_{\lambda}\sum_{\nu=1}^{m_{\lambda}}  \Tr[ \mathbb{P}_{\lambda}^{\nu} U_B \mathbb{P}_{\lambda}^{\nu} \rho_i \mathbb{P}_{\lambda}^{\nu} U_B\ad \mathbb{P}_{\lambda}^{\nu}[\mathbb{P}_{\lambda}^{\nu} H_\mu \mathbb{P}_{\lambda}^{\nu}, \mathbb{P}_{\lambda}^{\nu}  U_A\ad \mathbb{P}_{\lambda}^{\nu} \mathbb{P}_{\lambda}^{\nu} O \mathbb{P}_{\lambda}^{\nu}  \mathbb{P}_{\lambda}^{\nu} U_A \mathbb{P}_{\lambda}^{\nu}]]\nonumber\\
&=\sum_{\lambda}\sum_{\nu=1}^{m_{\lambda}}\Tr[ U_{B,\lambda} \rho_{i,\lambda}^\nu (U_{B,\lambda})\ad [H_{k,\lambda},(U_{A,\lambda})\ad O_{\lambda} U_{A,\lambda}]]\,.
\end{align}
Here, $\vec{\lambda}=(\lambda_1,\lambda_2,\lambda_3,\lambda_4,\lambda_5,\lambda_6)$, $\vec{\nu}=(\nu_1,\nu_2,\nu_3,\nu_4,\nu_5,\nu_6)$, and we have used in the second line the fact that the projectors $\mathbb{P}_{\lambda}^{\nu}$ are orthogonal. Thus, we can write
\begin{align}
    \partial_{\mu} \ell_{\thv}(\rho_i) =\ir\sum_{\lambda}\sum_{\nu=1}^{m_{\lambda}}\Tr[ U_{B,\lambda} \rho_{i,\lambda}^\nu (U_{B,\lambda})\ad [H_{k,\lambda},\widetilde{O}_{A,\lambda}]]\,,\label{eq:cost_irreps}
\end{align}
where we have defined $\widetilde{O}_{A,\lambda}^\nu=Q_\lambda^\nu U_A\ad O U_A(Q_\lambda^\nu)\ad$.
Equation~\eqref{eq:cost_irreps} shows that the loss function can be expressed as a summation of contributions from the blocks associated to the different invariant subspaces. While both $U(\thv)$ and $O$ are block diagonal in this bases, the same is not necessarily true for $\rho_i$. That is,  while the states $\rho_i$ may  not be  block diagonal in the irreps of $\mathbb{U}(2)$, the loss function one only takes into account  their projection into the irrep subspaces.

For a dataset, we evaluate our model with the training error $\hat{\LC}(\boldsymbol{\theta})$. Via linearity of expectation,
\begin{equation}
    \partial_\mu \hat{\LC}(\boldsymbol{\theta}) = \sum_{i=1}^M c_i \partial_{\mu} \ell_{\thv}(\rho_i)\,.
\end{equation}

Ultimately, we are interested in computing the  variance of the loss function partial derivative, which is defined as
\begin{equation}\label{eq:var}
    \Var_{\thv}[ \partial_{\mu} \hat{\LC}(\thv)]=\langle(\partial_{\mu} \hat{\LC}(\boldsymbol{\theta}))^2\rangle_{\thv}-\langle\partial_{\mu} \hat{\LC}(\boldsymbol{\theta})\rangle_{\thv}^2\,,
\end{equation}
where $\langle\cdot\rangle_{\thv}$ indicates the average over the set of parameters $\thv$ in the $S_n$-equivariant quantum neural network.

As discussed in the main text, we can always replace the integral over parameter space by an integral over a distribution of unitaries, which is known to converge to $2$-designs in the Lie group $\er^{\mf{g}}$ allowing us to in turn replace --under the assumption that $f$ is a polynomial of degree $\leq 2$ in the matrix elements of $U$--the previous with an integration over the Haar measure of this Lie group. We then use the fact that this Lie group is in general a product group (in our case all the simple groups are unitary or special unitary) and that its Haar measure is the product over the Haar measure of each of the simple normal subgroups~\footnote{We can see that this is the case by proposing the product over the Haar measure of each of the simple normal subgroups as the Haar measure over the Lie group, and realizing that it satisfies all the desired properties (e.g., normalized, left- and right-invariant). Then, since the Haar measure is unique, they must match. }. We have
\small
\begin{align}\label{ap:eq:unittohaar-main}
    \int_{\DC_{\thv}} d\thv f(U(\thv)) &= \int_{\DC} dU f(U) \nonumber\\
    &\rightarrow \int_{\er^{\mf{g}}} d\mu(U) f(U) \nonumber\\
    &=  \prod_\lm \int_{\er^{\mf{g}_{\lm}}} d\mu_\lm(U_\lm) f(\{U_\lm\})\,\\
    &=  \prod_\lm \int_{\mbb{U}(d_\lm)} d\mu_\lm(U_\lm) f(\{U_\lm\})\,.
\end{align}
\normalsize
As mentioned in the main text, integration over Haar of the unitary group can be done via Weingarten calculus.

Let us begin with the second term in Eq.~\eqref{eq:var}. That is, we want to compute $\sum_{i=1}^M c_i \langle\partial_{\mu} \ell_{\thv}(\rho_i)\rangle_{\thv}$. Explicitly, we have for the mean of the single-state expectation value $\ell_{\thv}(\rho_i)$
\begin{equation}\label{eq:exp-1}
    \langle\partial_{\mu} \ell_{\thv}(\rho_i)\rangle_{\thv}=i\sum_{\lambda}\sum_{\nu=1}^{m_{\lambda}}\int_{\thv} d\thv f_\lm^\nu(U_{B,\lm}(\thv),U_{A,\lm}(\thv))\,,
\end{equation}
where $f_\lm^\nu(U_{B,\lm}(\thv),U_{A,\lm}(\thv))=\Tr[ U_{B,\lambda}(\thv) \rho_{i,\lambda}^\nu (U_{B,\lambda}(\thv))\ad [H_{k,\lambda},\widetilde{O}_{A,\lambda}(\thv)]]$. Using Eq.~\eqref{ap:eq:unittohaar-main} on both $U_A$ and $U_B$
\small
\begin{align}\label{eq:exp-2}
    \langle\partial_{\mu} \ell_{\thv}(\rho_i)\rangle_{\thv}&= i\sum_{\lambda}\sum_{\nu=1}^{m_{\lambda}}  \,\prod_{\lm'} \int_{\mathbb{U}(d_{\lambda'})} d\mu(U_{A,\lambda'})\int_{\mathbb{U}(d_{\lambda'})}d\mu(U_{B,\lambda'}   )\Tr[ U_{B,\lambda} \rho_{i,\lambda}^\nu (U_{B,\lambda})\ad [H_{k,\lambda},\widetilde{O}_{A,\lambda}]]\,\\
    &=i\sum_{\lambda}\sum_{\nu=1}^{m_{\lambda}}\int_{\mathbb{U}(d_\lambda)} d\mu(U_{A,\lambda})\int_{\mathbb{U}(d_\lambda)}d\mu(U_{B,\lambda}   )\Tr[ U_{B,\lambda} \rho_{i,\lambda}^\nu (U_{B,\lambda})\ad [H_{k,\lambda},\widetilde{O}_{A,\lambda}]]\,.,
\end{align}
\normalsize
where in the second equation we have used the trivial integration over all $\lm'\neq \lm$.  Using Eq.~\eqref{Eq_lemma1}, we can evaluate the terms
\begin{equation}\label{eq:in-zero}
    \int_{\mathbb{U}(d_\lambda)}d\mu(U_{B,\lambda}   )\Tr[ U_{B,\lambda} \rho_{i,\lambda}^\nu (U_{B,\lambda})\ad [H_{k,\lambda},\widetilde{O}_{A,\lambda}]]=\frac{\Tr[\rho_{i,\lambda}^\nu]}{d_\lambda}\Tr[[H_{k,\lambda},\widetilde{O}_{A,\lambda}]]=0\,,
\end{equation}
where we have used that the trace of a commutator is always zero. From the previous we have shown that the first moment of the cost function partial derivative is zero, i.e., we have proved that 
\begin{equation}
    \langle\partial_{\mu} \ell_{\thv}(\rho_i)\rangle_{\thv}=0\,,
\end{equation}
and clearly, in turn $\langle \partial_\mu \hat{\LC}(\boldsymbol{\theta}) \rangle_{\boldsymbol{\theta}} = 0$. 

Next, let us evaluate the variance via the second moment
\small
\begin{align}
    \Var_{\thv}[ \partial_{\mu} \LC(\thv)]&=\mbb{E}_{\thv}[ (\partial_{\mu} \LC(\thv))^2]\nonumber\\
    &=-\sum_{i,j=1}^{M} c_ic_j\sum_{\lambda,\lambda'\in\tilde{\mathbb{U}}(2)}\sum_{\nu,}^{m_{\lambda}}\sum_{\nu'=1}^{m_{\lambda'}}
    \prod_{\a,\a'}
    \int_{\mathbb{U}(d_\a)} d\mu(U_{A,\a})\int_{\mathbb{U}(d_{\a'})}d\mu(U_{B,\a'}) \, I_i^\nu(U_{A,\lambda},U_{B,\lambda}) I_j^{\nu'}(U_{A,\lambda'},U_{B,\lambda'})\,,\label{eq:var2}
\end{align}
\normalsize
where we have defined $I_i^\nu(U_{A,\lambda},U_{B,\lambda})=\Tr[ U_{B,\lambda} \rho_{i,\lambda}^\nu (U_{B,\lambda})\ad [H_{k,\lambda},\widetilde{O}_{A,\lambda}]]$ and $\alpha$ and $\alpha'$ label irreps. Let us focus on a single summand, which we denote by $S_{\lm,\lm'}^{\nu,\nu',i,j}$. Two cases arise: either $\lm'\neq\lm$ or $\lm=\lm'$. Let us first consider $\lm'\neq\lm$
\small
\begin{align}
S_{\lm,\lm'}^{\nu,\nu',i,j}
&=\prod_{\a,\a'}
\int_{\mathbb{U}(d_\a)} d\mu(U_{A,\a})\int_{\mathbb{U}(d_{\a'})}d\mu(U_{B,\a'}) \, I_i^\nu(U_{A,\lambda},U_{B,\lambda}) I_j^{\nu'}(U_{A,\lambda'},U_{B,\lambda'}) \\ &= \int_{\mathbb{U}(d_\lambda)} d\mu(U_{A,\lambda})\int_{\mathbb{U}(d_\lambda)}d\mu(U_{B,\lambda}) \int_{\mathbb{U}(d_{\lambda'})} d\mu(U_{A,\lambda'})\int_{\mathbb{U}(d_{\lambda'})}d\mu(U_{B,\lambda'}) I_i^\nu(U_{A,\lambda},U_{B,\lambda}) I_j^{\nu'}(U_{A,\lambda'},U_{B,\lambda'})\,,
\end{align}
\normalsize
where after taking the trivial integrations we are left with four integrals. An explicit integration of $U_B$ (either $\lm$ or $\lm'$) via Eq.~\eqref{Eq_lemma3} leads to 
\begin{align}
   \int_{\mathbb{U}(d_\lambda)}d\mu(U_{B,\lambda}) I_i^\nu(U_{A,\lambda},U_{B,\lambda}) I_j^{\nu'}(U_{A,\lambda'},U_{B,\lambda'})&= I_j^{\nu'}(U_{A,\lambda'},U_{B,\lambda'}) \int_{\mathbb{U}(d_\lambda)}d\mu(U_{B,\lambda}) \Tr[ U_{B,\lambda} \rho_{i,\lambda}^\nu (U_{B,\lambda})\ad [H_{k,\lambda},\widetilde{O}_{A,\lambda}]]\nonumber\\
   &= 0\,,
\end{align}
where in the last equality we have used Eq.~\eqref{eq:in-zero}. On the other hand, when $\lambda = \lambda'$, we have 
\small
\begin{align}
S_{\lm,\lm'}^{\nu,\nu',i,j}
&=\prod_{\a,\a'}
\int_{\mathbb{U}(d_\a)} d\mu(U_{A,\a})\int_{\mathbb{U}(d_{\a'})}d\mu(U_{B,\a'}) \, I_i^\nu(U_{A,\lambda},U_{B,\lambda}) I_j^{\nu'}(U_{A,\lambda'},U_{B,\lambda'}) \\ &= \int_{\mathbb{U}(d_\lambda)} d\mu(U_{A,\lambda})\int_{\mathbb{U}(d_\lambda)}d\mu(U_{B,\lambda}) I_i^\nu(U_{A,\lambda},U_{B,\lambda}) I_j^{\nu'}(U_{A,\lambda'},U_{B,\lambda'})\,,
\end{align}
\normalsize
and in the end we are left with only two integrals. Integration over $U_B$ leads to
\begin{align}
    \int_{\mathbb{U}(d_\lambda)}d\mu(U_{B,\lambda}) I_i^\nu(U_{A,\lambda},U_{B,\lambda}) I_j^{\nu'}(U_{A,\lambda},U_{B,\lambda})=  \frac{\Delta(\rho_{i,\lambda}^\nu,\rho_{j,\lambda}^{\nu'})}{d_\lambda^2-1} \Tr[[H_{k,\lambda},\widetilde{O}_{A,\lambda}][H_{k,\lambda},\widetilde{O}_{A,\lambda}]]\,,
\end{align}
where we define the operator
\begin{equation}\label{eq:Delta}
    \Delta(A,B)=\Tr[AB]-\frac{\Tr[A]\Tr[B]}{d}\,,
\end{equation}
for any two $d\times d$ matrices  $A$ and $B$. Note that $A=B$ leads to $\Delta(A,A)=D_{\textnormal{HS}}(A,\Tr[A]\id/d)$ with $D_{\textnormal{HS}}=\Tr[(A-B)(A-B)\ad]$ the Hilbert-Schmidt distance. We will also use the notation $\Delta(A) = \Delta(A,A)$.
Thus, we can write
\begin{align}
    \int_{\mathbb{U}(d_\lambda)}d\mu(U_{B,\lambda}) I_i^\nu(U_{A,\lambda},U_{B,\lambda}) I_j^{\nu'}(U_{A,\lambda'},U_{B,\lambda'})=  \frac{\Delta(\rho_{i,\lambda}^\nu,\rho_{j,\lambda}^{\nu'})}{d_\lambda^2-1} \Tr[[H_{k,\lambda},\widetilde{O}_{A,\lambda}][H_{k,\lambda},\widetilde{O}_{A,\lambda}]]\delta_{\lambda,\lambda'} \,.\label{eq:int-B}
\end{align}
Finally, using Eq.~\eqref{Eq_lemma2} leads to
\begin{equation}
    \int_{\mathbb{U}(d_\lambda)} d\mu(U_{A,\lambda})\Tr[[H_{k,\lambda},\widetilde{O}_{A,\lambda}][H_{k,\lambda},\widetilde{O}_{A,\lambda}]]=-\frac{2d_\lambda}{(d_\lambda^2-1)} \Delta(H_{k,\lambda})\Delta(O_{\lambda})\,.\label{eq:int-A}
\end{equation}
Combining Eqs.~\eqref{eq:var2}, \eqref{eq:int-A} and~\eqref{eq:int-B} leads to
\begin{equation}\label{eq:var-penultimate}
    \Var_{\thv}[ \partial_{\mu} \LC(\thv)]=\sum_{i,j=1}^M c_ic_j\sum_{\lambda}\frac{2d_\lambda}{(d_\lambda^2-1)^2}\Delta(H_{k,\lambda})\Delta(O_{\lambda})\sum_{\nu,\nu'=1}^{m_{\lambda}}\Delta(\rho_{i,\lambda}^\nu,\rho_{j,\lambda}^{\nu'})\,.
\end{equation}
From Eq.~\eqref{eq:var-penultimate} we can see that the variance of the loss function contains a term for each irrep. Let us analyze each of those terms more closely.
First, we recall that
\begin{equation}\label{eq:DeltaH-O-zero}
    \Delta(H_{k,\lambda})=D_{\textnormal{HS}}(H_{k,\lambda},\Tr[H_{k,\lambda}]\frac{\id}{d_\lambda})\,, \quad \quad 
    \Delta(O_{\lambda})=D_{\textnormal{HS}}(O_{\lambda},\Tr[O_{\lambda}]\frac{\id_\lambda}{d_\lambda})\,,
\end{equation}
meaning that within each irrep, the variance contains a term that quantifies how close the reduced measurement operator $O$ and the reduced generator $H_{k,\lambda}$ is to the (normalized) $d_\lambda\times d_\lambda$ identity $\id_\lambda$. As expected if $O$ of $H_\mu$ are trivial when projected into the $\lambda$ irrep, then its respective contribution to the variance will be insignificant.  

We now note that the summation over $i$ and $j$ can be pushed into $\Delta(\rho_{i,\lm}^\nu, \rho_{j,\lm}^{\nu'})$. Specifically,
\begin{align}
    \sum_{i,j=1}^M c_ic_j \sum_{\nu,\nu'=1}^{m_{\lm}}\Delta(\rho_{i,\lambda}^\nu,\rho_{j,\lambda}^{\nu'})
    & = \sum_{i,j=1}^Mc_ic_j \Delta(\sum_{\nu=1}^{m_{\lm}}\rho_{i,\lambda}^\nu, \sum_{\nu'=1}^{m_{\lm}}\rho_{j,\lambda}^{\nu'})\\
    & = \Delta(\sum_{i=1}^M c_i \sum_{\nu=1}^{m_{\lm}} \rho_{i,\lm}^\nu, \sum_{j=1}^M c_j \sum_{\nu'=1}^{m_{\lm}} \rho_{j,\lm}^{\nu'})\\
    & = \Delta(\sum_{i=1}^M c_i \sum_{\nu=1}^{m_{\lm}} \rho_{i,\lm}^\nu).
\end{align}
This motivates us to define $\sigma = \sum_{i=1}^M c_i \rho_i$, and we can now write the variance as 
\begin{equation}\label{eq:var_main}
    \Var_{\thv}[ \partial_{\mu} \hat{\LC}(\thv)]=\sum_{\lambda}\frac{2d_\lambda}{(d_\lambda^2-1)^2}\Delta(H_{\mu,\lambda})\Delta(O_{\lambda}) \Delta(\sum_{\nu = 1}^{m_\lm}\sigma_\lm^\nu)\,.
\end{equation}

Recall that $\Delta(\sum_{\nu=1}^{m_\lm} \sigma_\lm^\nu) = D_{\textnormal{HS}}(\sum_{\nu=1}^{m_\lm}\sigma_\lm^\nu , \frac{\Tr[\sum_{\nu=1}^{m_{\lm}}\sigma_\lm^\nu]}{d_\lm}\id_{d_\lm}) \geq 0$ and thus the variance is non-negative as expected. To evaluate the presence of barren plateaus, it remains to find $\Delta(H_\lm)$ and $\Delta(O_\lm)$.

\end{proof}

\section{Supplementary Methods 5: Proof of Theorem~\ref{thm:main2}}\label{si_sec_thm2}

The formula for the variance of the partial derivatives, given in Eq.~\eqref{eq:var_main}, involves the terms $\Delta(H_{k,\lambda})$ and $\Delta(O_\lambda)$, where the subscript $\lm$ denotes 'restricted to' $\HC_\lm^\nu$. We here provide a proof for Theorem~\ref{thm:main2}, yielding an exact expression for $\Delta(A_\lm)$ for several general classes of equivariant operators $A$ (which include both $S_n$-equivariant generators and measurements). Let us recall the result.

\begin{theorem}\label{thm:main2}
    Let $A$ be an $S_n$-equivariant operator.
\begin{equation}
        \begin{cases}
     \text{If } A=\sum_{j=1}^n \chi_j\,, \quad  \text{then } 
        \Delta(A_{\lm}) = 2\binom{d_\lm+1}{3}\,,\\
        \text{If }  A=\sum_{k<j} \chi_j\chi_k\,, \quad \text{then } 
        \Delta(A_{\lm}) = \frac{8}{3}\binom{d_\lm+2}{5}\,,\\
        \text{If } A=\prod_{j=1}^n \chi_j\,, \quad \text{then } \Delta(A_\lambda) = 
        \frac{d_\lm^2-1+n\,{\rm mod}2}{d_\lm}\,, 
    \end{cases}
\end{equation}
where $\chi=X,Y,Z$.
\end{theorem}
Because $A$ is equivariant, the restrictions are independent of the index $\nu$. That is, $A \cong \bigoplus_\lm \id_{m_\lm} \otimes A_\lm$. Crucially, we recall that $\Delta(B) = \Tr[B^2]-\frac{\Tr[B]^2}{\dim(B)}$ is equal to the variance of the eigenvalues of $B$. Therefore, the core of the problem is to find the eigenvalues of the restricted operators. First, note that the choice of basis is arbitrary; Without loss of generality, we will assume $\chi = Z$. Secondly, we note that the task can be broken down to two parts:
(i) Find the eigenvalues of $A$ in terms of the Hamming weights of the computational basis states, and (ii) find which Hamming weights are compatible with the irrep subspace $\HC_\lm^\nu$.

Therefore, let us first derive the following intermediate result solving (i).

\setcounter{lemma}{3}
\begin{lemma}[Eigenvalues in terms of Hamming weights]\label{lemma_intermediate}
Given $A$ an operator diagonal in the computational basis $\{\ket{\vec{z}}\}$, let  $e(\vec{z})$ denote the eigenvalue of $\ket{\vec{z}}$ and $w(\vec{z})$ the Hamming weight of $\vec{z}$. Then the following is satisfied.

\begin{equation}\label{eq_eigs}
\begin{cases}
    \text{If } A=\sum_{j=1}^n Z_j\,, \quad  \text{then }\,e(\vec{z})=n-2w(\vec{z}),\\
    \text{If } A=\sum_{k<j} Z_jZ_k\,, \quad \text{then }\,e(\vec{z})=\frac{n^2-n(4w(\textbf{z})+1)+4w(\textbf{z})^2}{2},\\
    \text{If } A=\prod_{j=1}^n Z_j\,, \quad \text{then }\,e(\vec{z})=(-1)^{w(\vec{z})}\,.
\end{cases}
\end{equation}

\end{lemma}

\begin{proof}
The first case is
\begin{equation}
    \sum_{i=1}^n \bra{z} Z_i \ket{z}= \sum_{i=1}^n (-1)^{z_i} = n-2w(z)\,.
\end{equation}
The second case is
\begin{align}
    \bra{\textbf{z}}\sum_{i < j} Z_iZ_j\ket{\textbf{z}} & = \sum_{i < j}(-1)^{z_i+z_j},\\
    & = \sum_{\substack{i < j\\ z_i = z_j = 0}}1 + \sum_{\substack{i < j\\ z_i = z_j = 1}}1 + \sum_{\substack{i < j\\ z_i \neq z_j}}(-1),\\
    & = \binom{w(\textbf{z})}{2}+\binom{n-w(\textbf{z})}{2} - w(\textbf{z})(n-w(\textbf{z})),\\
    & = \frac{n^2-n(4w(\textbf{z})+1)+4w(\textbf{z})^2}{2}\,.
\end{align}
The last case is
\begin{equation}
    \bra{\vec{z}} Z^{\otimes n} \ket{\vec{z}} = \bra{\vec{z}}(-1)^{w(\vec{z})}\ket{\vec{z}}\,.
\end{equation}
\end{proof}

Now that we have connected the eigenvalues of $A$ in terms of Hamming weights $w(\vec{z})$ (see Lemma~\ref{lemma_intermediate}), we can proceed to (ii): Investigate which of them are compatible with a given irrep subspace $\HC_\lm^\nu$.

\begin{proof}[Proof of Theorem~\ref{thm:main2}]
 We prove the theorem by using the connection between the irreps of $\mbb{U}(2)$, Pauli operators, and spin systems. The operator $A$ must take the block diagonal form $A\cong \bigoplus_\lm \id_{m_\lm}\otimes A_\lm$, where $A_\lm$ acts on the irrep of $\mbb{U}(2)$ indexed by $\lm$. Notably the irreps of $\mbb{U}(2)$ are spaces of fixed total spin \cite{bacon2007quantum, zheng2022super,sakurai1995modern,woit2017quantum}. That is, for $\lm = (n-m,m)$, the total spin of states in the subspace are fixed at $\frac{n-2m}{2}$. Thus, the local spin components must lie in the range $\{-\frac{n-2m}{2},-\frac{n-2m}{2}+1,\ldots,\frac{n-2m}{2}\}$. As Hamming weight corresponds to the number of spin down sites, for an irrep $\lm$ the weights are restricted to the range $w\in\{m,\ldots,n-m\}$. We can then use Eq.~\eqref{eq_eigs} to compute the eigenvalues and from there $\Delta(A_\lm)$.


Note: for convenience going forward we will adopt the notation $A^{(k)} = \sum_{\{j_1,\ldots,j_k\}}\otimes_{i=1}^{k}\chi_{j_i}$. Again, without loss of generality we can take $\chi = Z$. Thus, for example, $A^{(1)}=\sum_{i=1}^n Z_i$.

\noindent\textbf{Case 1.} Here we consider $A^{(1)}=\sum_{i=1}^n Z_i$ and from Eq.~\eqref{eq_eigs} we have $e(\vec{z})=n-2w(\vec{z})$. Therefore the eigenvalues of $A_\lm^{(1)}$ are $e\in \{n-2m,\cdots, -n+2m\}$. Clearly, $\Tr[A_\lm^{(1)}]=0$ for all $\lm$. Also,
\begin{equation}
    \Delta(A_\lm^{(1)})  = \Tr[(A_\lm^{(1)})^2]-\underbrace{\Tr[A_\lm^{(1)}]^2}_{0}=\sum_{i=0}^{n-2m} (n-2m-2i)^2 =    2\binom{d_\lm+1}{3}\,.
\end{equation}
Note that for even $n$ and $\lambda = (\frac{n}{2}, \frac{n}{2})$ the summation above simply evaluates to 0.

\noindent\textbf{Case 2.} In this case $A^{(2)}=\sum_{i<j} Z_iZ_j$, and from Eq.~\eqref{eq_eigs} we have $e(\vec{z})=\frac{n^2-n(4w(\textbf{z})+1)+4w(\textbf{z})^2}{2}$. Summing over these values yields 
\begin{align}
    \Tr[A_\lm^{(2)}] & = \frac{(n-2m+1)(n(n-1)+4m^2-4m(n+1))}{6},\\
    \Tr[A_\lm^{(2)}] & = \frac{(n-2m+1)(48 m^4 - 96 m^3 (1 + n) - 
   8 m (1 + n)^2 (-2 + 3 n) + 8 m^2 (4 + 13 n + 9 n^2) + 
   n (-8 + 3 n + 2 n^2 + 3 n^3))}{60}
\end{align}
Thus, the variance of the eigenvalues is
\begin{align}
    \Delta(A_\lm^{(2)}) &= \frac{8}{3}\binom{d_\lm + 2}{5}.
\end{align}
Note that for $m > \lfloor \frac{n}{2} \rfloor-1$ this evaluates to 0.

\noindent\textbf{Case 3.} 
In this case $A^{(n)}=\otimes_{i=1}^n Z_i$, and from Eq.~\ref{eq_eigs} we have $e(\vec{z}) = (-1)^{w(\vec{z})}$. Thus, $(A^{(n)}_\lm)^2 = \id_{d_\lm}$ and $\Tr[(A^{(n)}_\lm)^2] = d_\lm$. It is left to compute $\Tr[A_\lm^{(n)}]$.
\begin{align}
    \Tr[A_\lm^{(n)}] = \sum_{w=m}^{n-m}(-1)^w = (-1)^w \sum_{k=0}^{n-2m}(-1)^k = (-1)^w\begin{cases} 1 & n\equiv 0 \ (\text{mod }2)\\
    0 & n\equiv 1 \ (\text{mod }2)\end{cases}\,.
\end{align}
Thus, $\Tr[A_\lm^{(n)}] = (n + 1) \text{ mod }2$. Combining, 
\begin{align}
    \Delta(A_\lm^{(n)}) = \frac{d_\lm^2-1+ (n \text{ mod }2)}{d_\lm}\,.
\end{align}

\end{proof}

\section{Supplementary Methods 6: Generalization of Theorem~\ref{thm:main2}}\label{ap:sec:gen}
We here generalize the result in Theorem~\ref{thm:main2} to the case of arbitrary $k$-local Pauli strings.

Consider the operator $A^{(k)} = \sum_{\{j_1,\ldots,j_k\}} \otimes_{i=1}^{k}\chi_{j_i}$. Again, $\chi$ is arbitrary thus we assume $\chi = Z$. The eigenvalue corresponding to a bitstring $\ket{\vec{z}}$ is 
\begin{align}
    \sum_{\{j_1,\ldots,j_k\}} \otimes_{i=1}^{k}\Z_{j_i} \ket{\vec{z}} & = \sum_{l=0}^{k}\sum_{\sum_{i}z_{j_i} = l}(-1)^l \ket{\vec{z}} \\
    & = \sum_{l=0}^k \binom{n-w(\vec{z})}{k-l}\binom{w(\vec{z})}{l}(-1)^l \ket{\vec{z}}\\
    & =  \mathcal{K}_k(w(\vec{z});n)\ket{\vec{z}},
\end{align}
where $\mathcal{K}_k(x;n)$ are the binary Krawtchouk polynomials. Then we have that
\begin{align}
    \Tr[A^{(k)}_\lm] = \sum_{w=m}^{n-m}\mathcal{K}_k(w;n)\, ,\quad  \Tr[(A^{(k)}_\lm)^2] = \sum_{w=m}^{n-m} \mathcal{K}_k(w;n)^2\,,
\end{align}
and
\begin{equation}
    \Delta(A^{(k)}_\lm) = \sum_{w=m}^{n-m}\mathcal{K}_k(w;n)^2 - \frac{\sum_{w,w'=m}^{n-m}\mathcal{K}_k(w;n)\mathcal{K}_k(w';n)}{n-2m+1}.
\end{equation}

\section{Supplementary Methods 7: Trainability of States}\label{app:states}
We here consider certain families of states and argue for, or against, their trainability under an $S_n$-equivariant architecture.

As discussed in the main text, Theorem~\ref{thm:main1} asserts that there can be several sources of untrainability, those QNN and measurement related, and those dataset related. Since Corollary~1 of the main the text states that the former source cannot lead to exponentially vanishing gradients, the trainability of the model hinges on the behavior of the latter dataset-related source. We note that this \textit{dataset-dependent} trainability is not unique for $S_n$-equivariant QNNs, but rather present in all absence of barren plateaus results (see Refs.~\cite{cerezo2021cost,pesah2020absence,larocca2021diagnosing,thanasilp2021subtleties,liu2021presence}) as there always exist datasets for which an otherwise trainable model can be rendered untrainable. In this section we briefly discuss the trainability of $S_n$-equivariant models for different types of datasets.

First, we recall from Theorem~\ref{thm:main1} that the dataset will induce a barren plateau if, for every $\lm$, $\Delta(\sum_{\nu=1}^{m_{\lambda}}\sigma_\lm^\nu)$ is exponentially small. Here, $\sigma =\sum_{i=1}^M c_i\rho_i$ and $\sigma_\lm^\nu$ is the reduction of $\sigma$ to $\HC_\lambda^\nu$. This highlights a clear connection between the underlying 
block structure and the trainability condition for a given dataset: Because $\sigma$ is Hermitian, we can interpret $\Delta(\sum_{\nu=1}^{m_{\lambda}}\sigma_\lm^\nu)$ as measuring the variance of the eigenvalues of the multiplicity-averaged reduced operators $\sigma_\lm^\nu$, and to guarantee trainability we need at least one irrep $\lm$ where such eigenvalue variance is not exponentially vanishing. 

In the following, we will make some distinction between training \textit{single states} versus \textit{datasets}. The former is a natural framework for variational quantum algorithms~\cite{cerezo2020variationalreview,bharti2021noisy}. In this case, we need only to show that $\Delta(\sum_{\nu}\sigma_\lm^\nu) \in\Omega( \frac{1}{\poly(n)})$ for some $\lm \in \hat{S}_n$. For the latter, trainability also depends on the weights $c_i$. Let us clarify this with an example.

\medskip

\noindent\textbf{Example 1}. Consider the set $\{\ket{\psi}\in \mathcal{H}\,|\, R(\pi) \ket{\psi} = \ket{\psi},\ \forall \pi \in S_n\}$. This is commonly known as the \textit{symmetric subspace} \cite{harrow2013church} and corresponds to the irrep $\lm = (n,0)$ where $d_{(n,0)} = n+1$ and $m_{(n,0)} = 1$. Now, suppose our dataset is composed of $n+1$ pure states forming an orthonormal basis for this subspace, which we label as $\rho_i = \dya{v_i}$. Individually, $\Delta(\rho_{(n,0)})=1-\frac{1}{n+1}\in\Theta(1)$ and these states are trainable. However, this may not be the case when we train on the entire dataset. Assume that we are performing binary classification, for which $c_i := -\frac{1}{M}y_i$.  If all inputs have the same label, then $\sigma = \pm\frac{1}{n+1}P_\text{Sym}$, where $P_{\text{Sym}}$ is the projector into the symmetric subspace, and thus $\sigma_{(n,0)} =\frac{1}{n+1}\id_{(n,0)}$ implying $\Delta(\sigma_{(n,0)}) = 0$ (this follows from the fact that  $P_{\text{Sym}}$ is the identity matrix in $\HC_\lambda^\nu$, and hence its eigenvalues have zero variance). Now instead consider $c_1 = -1$ and $c_i = 1$ for $i > 1$. In this case, $\Delta(\sigma_{(n,0)}) = \frac{4n}{(n+1)^3}$ and there will not be barren plateaus. Note that the true average, $c_i = \frac{1}{N}$, is a sort of worst case as each state is treated the same with no dependence on $y_i$. Regardless, in Sec.~\ref{eq:dataset} we present numerics showing that the true average of certain families of states is trainable. 

\medskip

While the previous shows that $\sigma$ may not be trainable, even when each $\rho_i$ is, we can instead say that if each $\rho_i$ is \textit{not} trainable, then nor will $\sigma$ be. We formalize this with the following theorem:
\begin{proposition}\label{thm:delta_boundsm}
    Given an ensemble of states $\{\rho_i\}$ and their corresponding $\sigma = \sum_i c_i \rho_i$, then
    \begin{equation}
        \Delta(\sum_{\nu = 1}^{m_\lm} \sigma_\lm^\nu) \leq (\sum_{i=1}^M \lvert c_i \rvert (\Delta(\sum_{\nu = 1}^{m_\lm}\rho_{i,\lm}^\nu))^{1/2})^2.
    \end{equation}
\end{proposition}

\begin{proof}
    This result is easily proven by recalling that $\Delta(A)= D_{\textnormal{HS}}(A,\frac{\Tr[A]}{d}\id_d) = \lVert A - \frac{\Tr[A]}{d}\id_d \rVert_2^2$. Then,
\begin{align}
    \Delta(\sum_{\nu = 1}^{m_\lm} \sigma_\lm^\nu) & = \lVert \sum_{i=1}^M c_i\sum_{\nu=1}^{m_{\lm}} \rho_{i,\lm}^\nu - \frac{\sum_{i=1}^M c_i \sum_{\nu=1}^{m_{\lm}} \Tr[\rho_{i,\lm}^\nu]}{d_\lm}\id_{d_\lm}\rVert_2^2\\
    & = \lVert \sum_{i=1}^M c_i(\sum_{\nu=1}^{m_{\lm}}\rho_{i,\lm}^\nu - \frac{ \sum_{\nu=1}^{m_{\lm}} \Tr[\rho_{i,\lm}^\nu]}{d_\lm}\id_{d_\lm})\rVert_2^2\\
    & \leq (\sum_{i=1}^M \lvert c_i \rvert \lVert \sum_{\nu=1}^{m_{\lm}}\rho_{i,\lm}^\nu - \frac{ \sum_{\nu=1}^{m_{\lm}} \Tr[\rho_{i,\lm}^\nu]}{d_\lm}\id_{d_\lm}) \rVert_2)^2\\
    & = (\sum_{i=1}^M \lvert c_i \rvert (\Delta(\sum_{\nu = 1}^{m_\lm}\rho_{i,\lm}^\nu))^{1/2})^2\,.
\end{align}
\end{proof}

The utility of this result is that if $\Delta(\sum_{\nu = 1}^{m_\lm}\rho_{i,\lm}^\nu)$ is small for an ensemble of states $\{\rho_i\}$, then it will still be small for the weighted average $\sigma$.

\subsection{Analytical results for trainability of single states}\label{sec:trainable_analytical}

\subsubsection{Fixed Hamming-weight encoding}

Building on the intuition that states with an at most polynomially vanishing component in the symmetric subspace can have a $\Delta(\rho_{(n,0)})$ that is at most polynomially vanishing (see Example 1), let us devise an encoding scheme with trainability guarantees: Given some array of real values $\{x_i\}$, perhaps representing some high dimensional vector, such that $\sum_i x_i^2 = 1$, encode each $x_i$ as the weight of a unique bitstring $\textbf{z}$ of Hamming weight $k$, where $k$ is some fixed constant.  That is, prepare the quantum state
\begin{equation}
    \ket{\vec{x}} = \sum_{w(\vec{z}) = k} x_{\vec{z}} \ket{\vec{z}},
\end{equation}
where $w(\vec{z})$ is the Hamming weight and we are now indexing $\{x_i\}$ with $\vec{z}$. To analyze the trainability of this encoding we consider the projection into the symmetric subspace. From the analysis above we know that we need only show that $\Tr[P_\text{sym}\dya{\vec{x}}] \in\Omega( \frac{1}{\poly (n)})$. Recalling that $P_\text{sym} = \frac{1}{n!}\sum_{g\in S_n}R(g)$, we can evaluate this trace directly. First, we note that
\begin{align}
    P_\text{sym}\ket{\vec{z}} = \frac{\max(k!,(n-k)!)}{n!}\sum_{w(\vec{y}) = k}\ket{\vec{y}},
\end{align}
where again $k = w(\vec{z})$. Thus,
\begin{equation}
    P_\text{sym}\ket{\vec{x}} = \frac{\max(k!,(n-k)!)}{n!}(\sum_i x_i)\sum_{w(\vec{y}) = k}\ket{\vec{y}}\,,
\end{equation}
and
\begin{equation}
    \langle \vec{x} | P_\text{sym} | \vec{x} \rangle = \frac{\max(k!,(n-k)!)}{n!}(\sum_i x_i^2 + 2\sum_{i < j}x_ix_j)\,.
\end{equation}
Note that since $\sum_i x_i^2 =1$, we have
\begin{equation}
    \langle \vec{x} | P_\text{sym} | \vec{x} \rangle = \frac{\max(k!,(n-k)!)}{n!}(1 + 2\sum_{i < j}x_ix_j)\,.
\end{equation}
Here we consider two cases. First, that where $x_i \geq 0,\ \forall i$. For example, in encoding an image, the pixel values are all non-negative. In this case clearly $\langle \vec{x} | P_\text{sym} | \vec{x} \rangle \geq \frac{\max(k!,(n-k)!)}{n!}$. As a second scenario, consider drawing data $\{z_i\}_{i=1}^{\binom{n}{k}}$ from some process such that $\mathbb{E}[z_iz_j] = 0,\ i\neq j$.

It is not unreasonable to assume that $x_i$ is symmetric in distribution around 0. That is, $P(x_i \leq -a) = P(x_i \geq a)$. If we also assume that only the magnitudes of the data are correlated, and not the signs, then $\mathbb{E}[x_ix_j] = 0$ (for $i\neq j$). In this case, $\mathbb{E}[\langle \vec{x} | P_\text{sym} | \vec{x} \rangle] = \frac{\max(k!,(n-k)!)}{n!}$. Without loss of generality take $k \leq \frac{n}{2}$. Then, $\frac{(n-k)!}{n!}\in \Theta(n^{-k})$. For a fixed $k$, then the expected component in the symmetric subspace $\Delta(\rho_{(n,0)})$ is inverse-polynomial.

\subsubsection{Global Haar random state }

While the previous family of input states were trainable, we now provide an example leading to untrainability. Consider sending a single global Haar random pure state $\rho$ through the $S_n$-equivariant QNN.
That is, states of the form $\rho = V \dya{\psi} V^\dagger$ where $V$ is a Haar-random unitary and $\ket{\psi}$ is an arbitrary pure state. 

Consider the expectation value:
\begin{align}
    \mathbb{E}[\Delta(\rho_\lm^\nu,\rho_\lm^{\nu'})] = \mathbb{E}[\Tr[\rho_\lm^\nu\rho_\lm^{\nu'}]]-\frac{1}{d_\lm}\mathbb{E}[\Tr[\rho_\lm^\nu]\Tr[\rho_\lm^{\nu'}]].
\end{align}
To continue, recall that $\rho_\lm^\nu = Q_\lm^\nu \rho (Q_\lm^\nu)^\dagger$. We can then rewrite the first expectation value above as
\begin{align}
    \mathbb{E}[\Tr[\rho_\lm^\nu\rho_\lm^{\nu'}]] &= \int_{\mathbb{U}(d)} d\mu(V) \Tr[V\dya{\psi} V^\dagger (Q_\lm^{\nu})^\dagger Q_\lm^{\nu'}V\dya{\psi} V^\dagger (Q_\lm^{\nu'})^\dagger Q_\lm^{\nu}]\\
    & = \int_{\mathbb{U}(d)} d\mu(V) \Tr[V\dya{\psi} V^\dagger P_\lambda^{\nu,\nu'} V\dya{\psi} V^\dagger M^\dagger],
\end{align}
where we have defined $P_\lm^{\nu,\nu'}:=(Q_\lm^{\nu})^\dagger Q_\lm^{\nu'}$. We employ the following identity (Eq.~\eqref{Eq_lemma2})
\begin{align}\label{eq:haar_int_1}
    \int_{\mathbb{U}(d)}d\mu(V) \Tr[V A V^\dagger B V C V^\dagger D] =&  \frac{1}{d^2-1}(\Tr[A]\Tr[C]\Tr[BD]+\Tr[AC]\Tr[B]\Tr[D])\\ & -\frac{1}{d(d^2-1)}(\Tr[AC]\Tr[BD]+\Tr[A]\Tr[B]\Tr[C]\Tr[D]).
\end{align}
 If $\nu = \nu'$ then $Q_\lm^\nu (Q_\lm^\nu)^\dagger = \id_{d_\lm}$ and $\Tr[P_\lm^{\nu,\nu'}]=d_\lm$. If $\nu \neq \nu'$ then $\Tr[P_\lm^{\nu,\nu'}]=0$. Take $\{\ket{v_i^\nu}\}$ and $\{\ket{v_i^{\nu'}}\}$ to be the orthogonal sets of vectors in the full Hilbert space corresponding to the subspaces $\mathcal{H}_\lm^\nu$ and $\mathcal{H}_\lm^{\nu'}$. Then $\Tr[P_\lm^{\nu,\nu'}] = \sum_i \langle v_i^\nu|v_i^{\nu'} \rangle= 0$. It is left for us to compute $\Tr[P_\lm^{\nu,\nu'}(P_\lm^{\nu,\nu'})^\dagger]$. Note that $P_\lm^{\nu,\nu'}(P_\lm^{\nu,\nu'})^\dagger = (Q_\lm^\nu)^\dagger Q_\lm^{\nu'}(Q_\lm^{\nu'})^\dagger Q_\lm^{\nu} = (Q_\lm^{\nu})^\dagger Q_\lm^{\nu}$. Thus, $\Tr[P_\lm^{\nu,\nu'}(P_\lm^{\nu,\nu'})^\dagger] = d_\lm$. Using Eq.~\eqref{eq:haar_int_1} we then see that
 \begin{align}
     \mathbb{E}[\Tr[\rho_\lm^\nu \rho_\lm^{\nu'}]] = \frac{d_\lm+d_\lm^2\delta_{\nu\nu'}}{d(d+1)}\,.
 \end{align}
 
To evaluate the second expectation value we first rewrite the expression as 
\begin{align}
    \mathbb{E}[\Tr[\rho_\lm^\nu]\Tr[\rho_\lm^{\nu'}]] & = \int_{\mathbb{U}(d)} d\mu(V) \Tr[V \dya{\psi} V^\dagger (Q_\lm^\nu)^\dagger Q_\lm^\nu]\Tr[V \dya{\psi} V^\dagger (Q_\lm^{\nu'})^\dagger Q_\lm^{\nu'}]\\
    & = \int_{\mathbb{U}(d)} d\mu(V) \Tr[V \dya{\psi} V^\dagger P_\lm^{\nu,\nu}]\Tr[V \dya{\psi} V^\dagger P_\lm^{\nu',\nu'}],
\end{align}
We now employ the following identity (Eq.~\eqref{Eq_lemma3})
\begin{align}\label{eq:haar_int_2}
    \int_{\mathbb{U}(d)} d\mu(V) \Tr[V A V^\dagger B]\Tr[V C V^\dagger D] = & \frac{1}{d^2-1}(\Tr[A]\Tr[B]\Tr[C]\Tr[D]+\Tr[AC]\Tr[BD])\\- & \frac{1}{d(d^2-1)}(\Tr[AC]\Tr[B]\Tr[D]+\Tr[A]\Tr[C]\Tr[BD])\,.
\end{align}
This allows us to evaluate the expectation value
\begin{align}
    \mathbb{E}[\Tr[\rho_\lm^\nu]\Tr[\rho_\lm^{\nu'}]] = \frac{d_\lm(d_\lm+\delta_{\nu\nu'})}{d(d+1)}\,.
\end{align}
Combining the expectation values derived above yields
\begin{align}
    \mathbb{E}[\Delta(\rho_\lm^\nu,\rho_\lm^{\nu'})] = \begin{cases}
        \frac{d_\lm^2-1}{d(d+1)} & \nu = \nu'\\
        0 & \text{otherwise}
    \end{cases}\,.
\end{align}
We turn back to $\mathbb{E}[\Delta(\sum_{\nu=1}^{m_\lm}\rho_\lm^\nu)]$, which is now easy to evaluate.
\begin{align}
    \mathbb{E}[\Delta(\sum_{\nu=1}^{m_\lm}\rho_\lm^\nu)] = \frac{m_\lm(d_\lm^2-1)}{d(d+1)}.
\end{align}
Recall that, $m_\lm \in O(2^n)$, $d_\lm \in O(n)$ and $d= 2^n$. Thus, $\mathbb{E}[\Delta(\sum_{\nu=1}^{m_\lm}\rho_\lm^\nu)] \in O(\frac{n^2}{2^n})$ and, on average, Haar random states will not be trainable. Of course, this is not surprising as it is known that Haar random states are highly entangled and make poor computational resources~\cite{gross2009most}.

\section{Supplementary Methods 8: VC Dimension Bounds (For General Symmetries)}\label{app:vc}
Besides the covering number bound in the main text, here we show a bound on the VC dimension of equivariant QNNs.

\setcounter{theorem}{6}
\begin{theorem}\label{thm:vc}
    Let $\UC_{\thv}$ be an $S_n$-equivariant QNN with generators in $\GC$, and $O$ an $S_n$-equivariant measurement from $\MC$. The VC dimension of classifiers of the form $\Tr[U(\boldsymbol{\theta}) \rho U(\boldsymbol{\theta})\ad O] \geq b$ is less than or equal to $\sum_{\lm}d_\lm^2={\rm Te}_{n+1} \in \Theta(n^3)$.
\end{theorem}

We begin by giving a general framework for GQML. Consider a quantum neural network $\mathcal{N}_{\thv}: \mathcal{B}(\mathcal{H})\mapsto\mathcal{B}(\mathcal{K})$, where $\mathcal{H}$ is the input Hilbert space and $\mathcal{K}$ the output. Note that this needs not be a unitary mapping and can be a general quantum channel. Further, assume $\mathcal{N}$ is equivariant to a compact group $G$ with representation $R_1(g)$ on $\mathcal{H}$ and $R_2(g)$ on $\mathcal{K}$. That is,
\begin{equation}
    \mathcal{N}_{\thv}(R_1(g) \rho R_1(g)^\dagger) = R_2(g) \mathcal{N}_{\thv}(\rho) R_2(g)^\dagger,\ \forall g\in G.
\end{equation}

With this neural network we construct classifiers of the form
\begin{equation}
    h_{\thv}(\rho) = \Tr[O\mathcal{N}_{\thv}(\rho)] \geq c,
\end{equation}
where $O$ is some Hermitian operator such that $[O,R_2(g)] = 0,\ \forall g\in G$. Clearly, $h_{\thv}(R_1(g) \rho R_1^\dagger(g)) = h_{\thv}(\rho)$.

\begin{theorem}
    The VC dimension of classifiers of the form $\Tr[O\mathcal{N}_{\thv}(\rho)] \geq c$ is upper-bounded by the dimension of the commutant of $R_2$, which we denote as $\mf{comm}(R_2) = \{H\in \mathcal{B}(\HC)\,|\, [H,R_2(g)]=0 \ \forall g \in G\}$.
\end{theorem}
\begin{proof}
    First, note that $c$ can be absorbed into $O$ via setting $O' := O - c\id$. As $\id$ commutes with everything, $O'$ is then still in $\mf{comm}(R_2)$. We can then consider classifiers of the form $\Tr[O \mathcal{N}(\rho;\boldsymbol{\theta})] \geq 0$ without loss of generality.
    
    Recall that the twirling operation for a compact group is given by $\mathcal{T}_G[X] = \int_G d\mu(g) R(g) X R(g)^\dagger$, where $\mu$ is the Haar measure. In addition, $\mathcal{T}_G$ is known to be a  projector into $\mf{comm}(R)$~\cite{ragone2022representation}. Note that here we twirl with $R_2(g)$. Clearly, $\mathcal{T}_G[O] = O$ for the chosen measurement operator. We require one last property of twirling--it is self-adjoint. That is, $\langle \mathcal{T}_G[X],Y\rangle = \langle X, \mathcal{T}_G[Y] \rangle$, where here $\langle \cdot, \cdot \rangle$ is the HS inner product.
    \begin{align}
        h_{\thv}(\rho) & = \Tr[O\mathcal{N}_{\thv}(\rho)]\\
        & = \Tr[\mathcal{T}_G[O]\mathcal{N}_{\thv}(\rho)]\\
        & = \Tr[O \mathcal{T}_G[\mathcal{N}_{\thv}(\rho)]].
    \end{align}
    $\mathcal{T}_G[\mathcal{N}_{\thv}(\rho)]$ is some operator in $\mf{comm}(R_2)$. Thus, any classifier of the form $\Tr[O \mathcal{N}(\rho;\boldsymbol{\theta})] \geq 0$ 
    is a linear classifier in $\mf{comm}(R_2)$. The VC dimension of all linear classifiers is the dimension of the vector space \cite{Hajek2021Statistical}. As the space of classifiers obtained by varying $\boldsymbol{\theta}$ is a subset of all linear classifiers, then the VC dimension of $F_{\thv}$ is less than or equal to the dimension of the commutant of $R_2(g)$.
\end{proof}

For the architecture considered in this paper, $G=S_n$ and $R_1(g) = R_2(g)$. From the isotypic decomposition the dimension of the commutant will be the sum of the squares of the multiplicities. The multiplicity space of irrep $\lm$ is an irrep of $\mathbb{U}(2)$ of dimension $n-2m+1$. Summing over all irreps, we see that the total dimension of the commutant is $\sum_{m=0}^{\lfloor \frac{n}{2} \rfloor}(n-2m+1)^2$, proving Theorem~\ref{thm:vc}.


\section{Supplementary Methods 9: Absence of barren plateaus implies absence of narrow gorges}

In this section we will show that in the loss function does not exhibit a barren plateaus, and more importantly that it cannot have a narrow gorge.

First, let us recall the definition of a narrow gorge from Ref.~\cite{arrasmith2021equivalence}. Let the loss  function be defined as
\begin{equation}
    \ell_{\thv}(\rho)=\Tr[ \UC_{\thv}(\rho) O]\,.
\end{equation}
We say that $\ell_{\thv}(\rho)$ exhibits a narrow gorge if the following two conditions are met:
\begin{enumerate}
    \item There exist points in the landscape which are lower than the mean cost value by at least a a quantity $\Delta(n)$ with $\Delta(n)\in\Omega(1/\poly(n))$.
    \item  The probability that the cost function value at a point chosen from the uniform distribution over the parameter space differs from the mean by at least $\delta$ is bounded as
    \begin{equation}
        P(|\ell_{\thv}(\rho)-\langle \ell_{\thv}(\rho)\rangle_{\thv}|\geq \delta)\leq \frac{G(n)}{\delta^2}\,,
    \end{equation}
    with $G(n)\in\OC(1/b^n)$ for some $b>1$.
\end{enumerate}

Next, let us assume that $\ell_{\thv}(\rho)$ has no barren plateau. That is, we have that
\begin{equation}\label{eq:absence-of-BPs}
    \Var_{\thv}[\partial_\mu \ell_{\thv}(\rho)]\,, \Var_{\thv}[\ell_{\thv}(\rho)]= H(n)\,, \quad \text{with } H(n)\in\Omega(1/\poly(n))\,,
\end{equation}
where we have used the fact that absence of barren plateaus implies both gradients and loss function variances vanishing, at most, polynomially with $n$~\cite{arrasmith2021equivalence}. As we will see, Eq.~\eqref{eq:absence-of-BPs} allows us to  show that while the first condition is always met with overwhelming probability, the second one is never met, implying that there is no narrow gorge. We define the two complementary events  $\AC=\left\{\thv:|\ell_{\thv}(\rho)-\langle \ell_{\thv}(\rho)\rangle_{\thv}|> \delta\right\}$ and $\AC^{\perp}=\left\{\thv:|\ell_{\thv}(\rho)-\langle \ell_{\thv}(\rho)\rangle_{\thv}|\leq \delta\right\}$, such that
\begin{equation}
    \Pr[\AC]+\Pr[\AC^{\perp}]=1\,.
\end{equation}
Then, from Chebyshev's inequality, we know that 
\begin{equation}
    \Pr[\AC^{\perp}]\leq \frac{\Var_{\thv}[\ell_{\thv}(\rho)]}{\delta^2}\,,
\end{equation}
and hence the following inequality hold
\begin{align}\label{eq:probab-bound}
    \Pr[\AC]&=1-\Pr[\AC^{\perp}]\nonumber\\
    &\geq 1- \frac{\Var_{\thv}[\ell_{\thv}(\rho)]}{\delta^2}\,.
\end{align}
From here, we can ask the question. What is the probability that the loss differs from its mean by a quantity $\delta =(H(n))^{\kappa}$. In particular we choose  $\kappa<1/2$ and $\kappa\in\Omega(1)$ so that $\delta\in \Omega(1/\poly(n))$. Using Eqs.~\eqref{eq:absence-of-BPs} and~\eqref{eq:probab-bound} leads to 
\begin{align}\label{eq:probab-bound}
    \Pr[\AC]
    &\geq 1- (H(n))^{1-2\kappa}\underset{n\rightarrow \infty}{\rightarrow}1\,.
\end{align}
Hence, we can see that with overwhelming probability, the cost will differ from its mean by a quantity that is at least polynomially close to one. As such, the loss exhibits no narrow gorge. That is, the loss anti-concentrates.

\section{Supplementary Methods 10: Numerical results}\label{app:numerics}
In this section, we provide a comprehensive set of numerical results supporting the results and theorems presented in the main text. 
In Sec.~\ref{eq:num_grads_ovw}, we report scalings of the variances of the gradients for several families of random input states. These were used in Table.~1 of the main text. In Sec.~\ref{sec:assess_var}, we assess the gradient variances as given in Eq.~\eqref{eq:var_main}, and compare them to variances numerically estimated. In Sec.~\ref{sec:contribs}, we evaluate the contributions of the different irreps to the overall gradient variances. While results of Sec.\ref{eq:num_grads_ovw}, Sec.~\ref{sec:assess_var} and Sec.~\ref{sec:contribs} are obtained in the single state scenario ($M=1$ in Eq.~\eqref{eq:empirical-loss}), in Sec.~\ref{eq:dataset}, we probe the effect of multi-state datasets ($M=50$).

All the gradients results that are produced are based on a minimum of $50$ random initial states with $50$ random set of parameters per state (i.e., at least $2500$ parameter gradient values). If not otherwise stated, gradient results correspond to a choice of measurement operator $O= \frac{1}{n}\sum_{j=1}^n X_j$. Note that other choices of observables belonging to $\mathcal{M}$ (Eq.~10 of the main text) yield similar results but are not displayed here.

\subsection{Gradient scalings for different families of input states}\label{eq:num_grads_ovw}
We numerically probe the scaling of the gradients for different families of random initial states. All the results reported in  Fig.~\ref{fig:grads_diff_initial_states} are obtained for a generator $H_\mu = \frac{1}{n}\sum_{j=1}^n X_j$ in the middle of the $S_n$ circuits with $L=3n$ layers. Similar scalings are obtained for the other generators (Eq.~9 in the main text) but not displayed here.

\begin{figure}[b]
    \includegraphics[scale = 0.80]{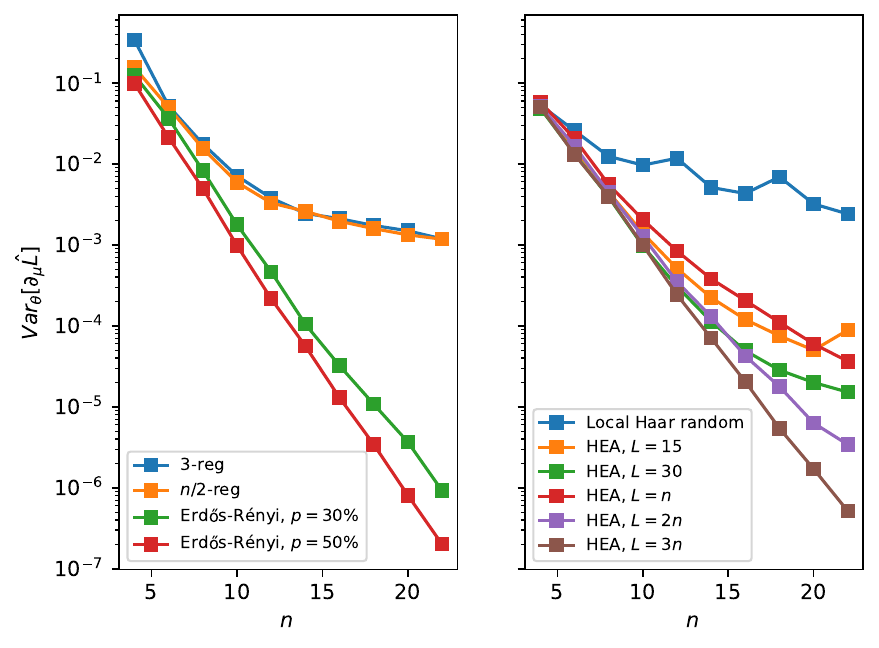}
    \caption{\textbf{Scaling of the gradients of $S_n$-equivariant circuits}. Left: the input states are graph states, which underlying graphs are drawn from different distributions (in legend). Right: Local Haar random states and states prepared with a random circuit taken to be a HEA circuit (described further in the main text) with a number $L$ of layers either constant or scaling linearly with the system sizes (given in legend).}
    \label{fig:grads_diff_initial_states}
\end{figure}

In Fig.~\ref{fig:grads_diff_initial_states}(a), we study \emph{graph states}~\cite{raussendorf2003measurement, hein2004multiparty} which, we recall, are obtained by applying a controlled-Z gate, $CZ= \dya{0} \otimes \id + \dya{1}\otimes Z$ 
for each edge $(a,b)$ belonging to an underlying graph, onto an initial state $|+\rangle^{\otimes n}$.
The underlying random graphs considered are either $k$-regular graphs (with $k=3$ or $k=n/2$), or they follow an Erd\"{o}s–R\'{e}nyi distribution (with a probability of any edge to be be included taken to be $p=30\%$ or $50\%$). As can be seen, gradients are found exponentially decreasing for the graphs drawn from a Erd\"{o}s–R\'{e}nyi distribution, but only polynomially decreasing for the $k$-regular graphs. 

Gradients for local Haar random input states and also for states produced by application of a random circuit to a fiducial state $|0\rangle^{\otimes n}$ are reported in Fig.~\ref{fig:grads_diff_initial_states}(b). 
For the random circuit preparation, we use a Hardware Efficient Ansatz (HEA)~\cite{kandala2017hardware} composed of layers of $\Y$ rotations applied to each qubit with random angles followed by controlled-NOT gates acting onto pairs of adjacent qubits.
Such circuits of HEA layers are either taken to be constant depth ($L=15$ or $30$), or with a number layers scaling linearly with the system size ($L=n$, $2n$ or $3$n). 
Only in the case of a linear number of layers that the gradients are found to vanish exponentially.

\subsection{Assessment of the analytical variances.}\label{sec:assess_var}

\begin{figure}[t]
    \includegraphics[scale = 0.80]{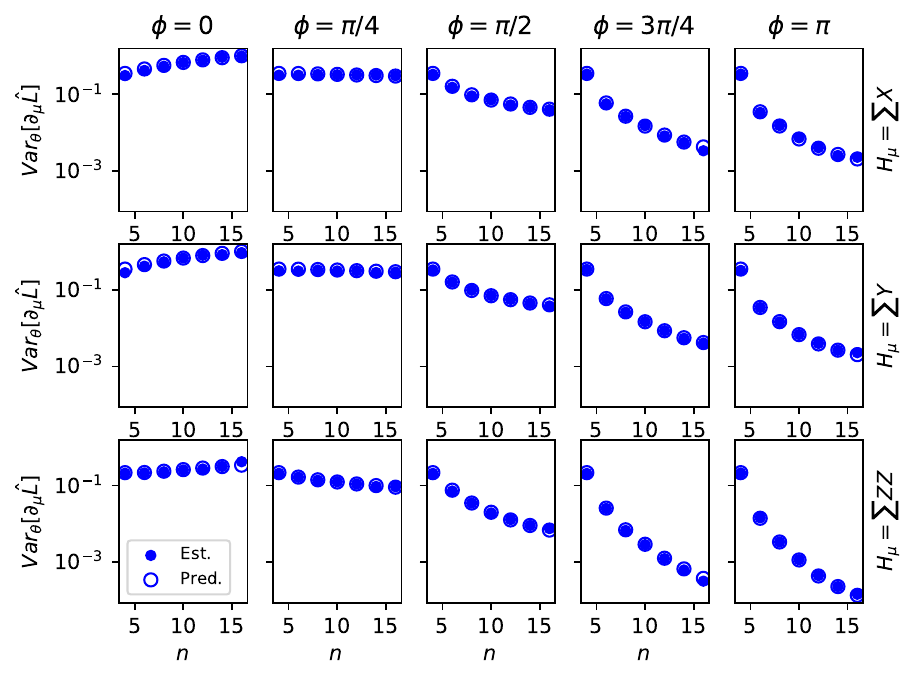}
    \caption{\textbf{Predicted vs actual gradients} Comparison between gradients estimated numerically (Est.) and our analytical expressions (Pred.) in Eq.~\eqref{eq:var_main}.
    All the data reported corresponds to generalized graph states (described in the main text) for $3$-regular graphs with number of nodes $n\in [4,16]$. Results are displayed for varied encoding angles $\phi$ (columns) and for each of the generators $H_\mu$ (rows) appearing in $\mathcal{G}$ (Eq.~9 of the main text).
}
    \label{fig:est_vs_predict_graph_states_3reg}
\end{figure}

In Fig.~\ref{fig:est_vs_predict_graph_states_3reg} we compare analytical expression of the gradient variances (as provided in Eq.~\eqref{eq:var_main}) to variances estimated numerically. 
Both are evaluated for a \emph{generalized graph states}, where edges are now encoded by means of a controlled-phase unitary $CP(\phi)=\dya{0} \otimes \id + \dya{1}\otimes P$, with $P=\dya{0} + \er^{\ir\phi}\dya{1}$ 
that depends on an global angle $\phi$. The case $\phi = \pi$ corresponds to the graph states discussed thus far.
The underlying graphs are $3$-regular graphs with number of nodes $n\in[4,16]$.

To evaluate Eq.~\eqref{eq:var_main}, we explicitly construct a basis that block diagonalizes $S_n$ in the form of the isotypic decomposition (Eq.~11 in the main text)
This so called Schur basis can be constructed recursively (for instance, see Refs.~\cite{bacon2007quantum,kirby2017practical}). In turn, this allows us to evaluate for any operator $B$ (including the input state $\sigma$, the measurement operators $O$ and the generators $H_{\mu}$) its restriction $B_\lambda^\nu$ as defined in Eq.~\eqref{Eq_red} and as required to compute Eq.~\eqref{eq:var_main}.  
As shown in Fig.~\ref{fig:est_vs_predict_graph_states_3reg}, for all the different scenarios studied, the variances yielded by Eq.~\eqref{eq:var_main} match closely the variance numerically estimated. Furthermore, in addition to the case with $\phi=\pi$ already discussed it can be seen that there exist choices of $\phi$ results in non-exponentially vanishing gradient variances.

\subsection{Contributions of the different irreps for $3$-regular graph states}\label{sec:contribs}
Given this ability to evaluate the different terms entering Eq.~\eqref{eq:var_main}, we can further analyze the contributions of the different irreps to the total variance, for any given initial state. 
The contribution of a given irrep is defined as the summand indexed by $\lambda$ in Eq.~\eqref{eq:var_main}, such that the sum of all irreps contributions equals the overall variance.

In Fig.~\ref{fig:est_vs_predict_graph_states_3reg}, we report these contributions for each of the irreps involved and labeled as $\lambda=(n-m, m)$ (given in legend). 
These data are obtained for graph states based on $3$-regular graphs and a number of nodes $n\in[4,16]$. While analytical results for trainable states in Sec.~\ref{sec:trainable_analytical} were based on symmetric component of the input state vanishing only polynomially, here we can see that the situation is different. 
Explicitly, despite the exponential vanishing of the contribution of $\lambda=(n,0)$ (the symmetric subspace in the top right sub-figure), the overall gradient variance only decays polynomially with $n$. As can be seen, each of the irreps plays a significant contribution to the overall variance (dashed line) at some point (i.e., for a specific value of $n$) before decreasing. Similar results are obtained for the case of local Haar random states (discussed in Sec.~\ref{eq:num_grads_ovw}) but not displayed here.    

\begin{figure}[b]
    \includegraphics[scale = 0.80]{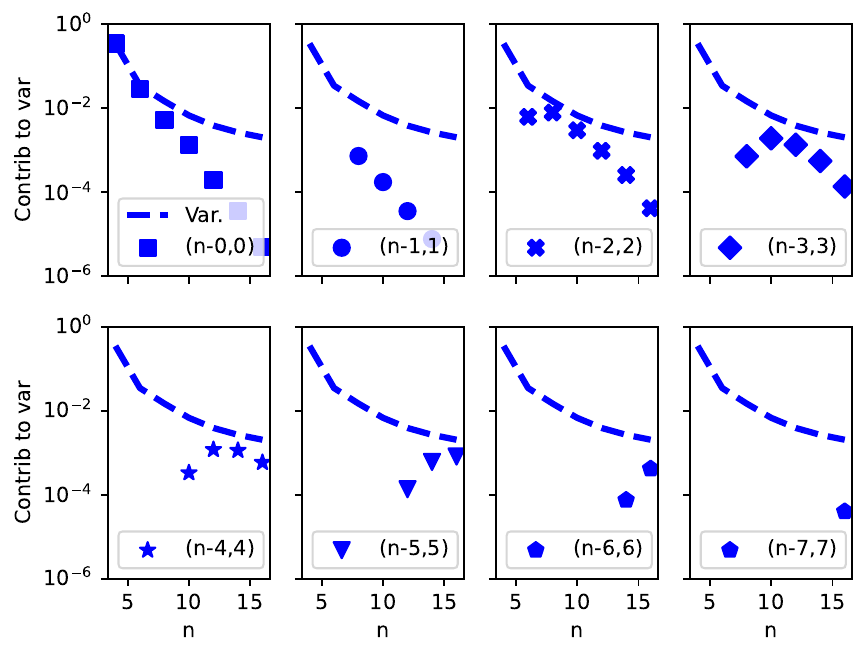}
    \caption{\textbf{Contributions of the different irreps to the overall variance of the gradients}. We report the contributions of the different irreps $\lambda=(n-m, m)$ (in legend) to the overall variance of the gradients (dashed lines). The contribution of an irrep is defined as the summand in Eq.~\eqref{eq:var_main}. All the results presented correspond to graph states for $3$-regular graphs with a number of nodes $n\in[4,16]$.
}
    \label{fig:num}
\end{figure}

\subsection{Variances for single states and dataset}\label{eq:dataset}
\begin{figure*}[h!t]
    \centering
    \includegraphics[width=0.80\linewidth]{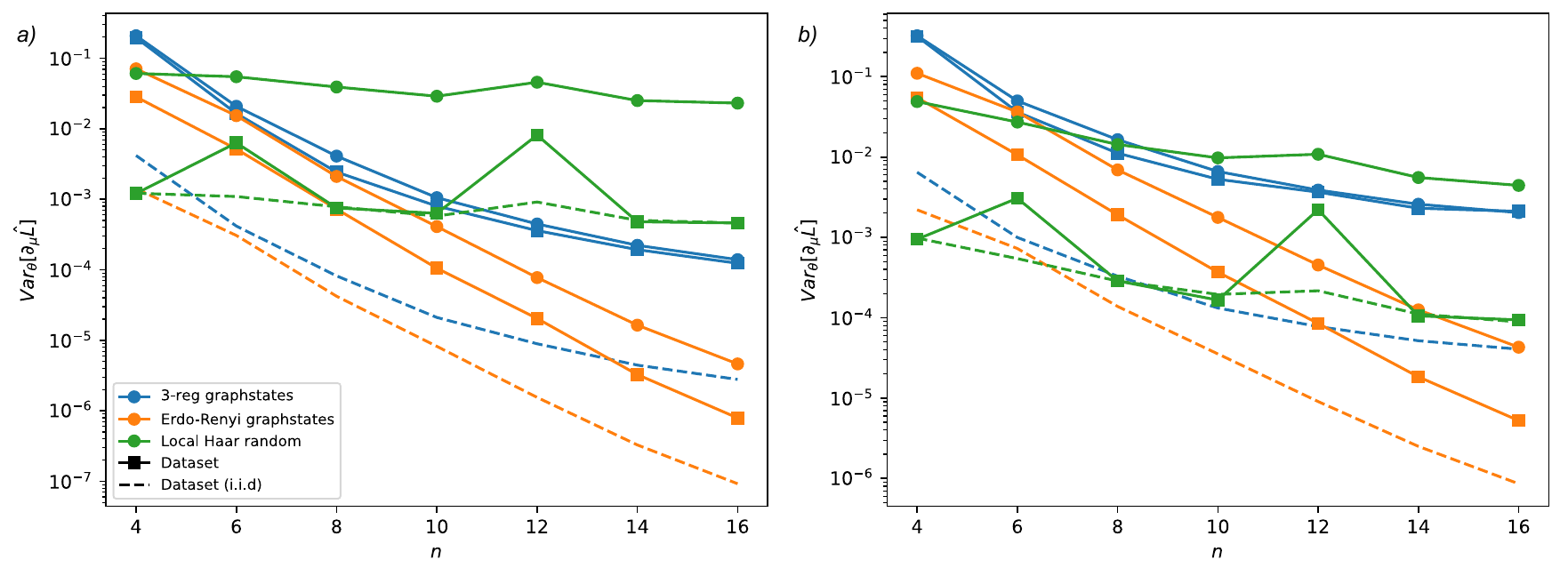}
    \caption{\textbf{Variance of the gradients for different sizes $N$ of the dataset.} 
    We report variances obtained for gradients of the empirical loss in Eq.~\eqref{eq:empirical-loss} for datasets of $M=50$ states with uniform weighting $c_i=1/M$ (square markers) and compare them to the single state, with $M=1$, case (circle marker).  Additionally we report variances for the dataset scenario under the assumption of i.i.d. gradients (dashed line, and discussed further in the main text). Variances are obtained for three distributions of input states (colors in legend) including local Haar random states and graph states based on $3$-regular graphs or drawn for an Erd\"{o}s–R\'{e}nyi distribution. Results are reported for two gate generators, namely $H_\mu = \sum_{k<j} Z_jZ_k$ in (a) and $H_\mu = \sum_{i=1}^n X_j$ in (b).
    }
    \label{fig:var_dataset}
\end{figure*}

While all the previous results were reported for the single state scenario (i.e., $M=1$), we now numerically probe the effect of averaging the empirical loss in Eq.~\eqref{eq:empirical-loss} over a dataset of $M=50$ states. 
To do so, we need to define the weights employed which, in practice, would depend on the specific details of the task at hand. Here, for simplicity we will resort to a simple average yielding weights $c_i:=1/M$.

In Fig.~\ref{fig:var_dataset} we report variances resulting from such average (square markers) and compare them to variances of the gradients for the single states (circle markers). Additionally, we report the variances that would arise for $M=50$, but, under the assumption that the gradients obtained for different states are independent and identically distributed (i.i.d. in legend). These are readily obtained by rescaling the variances of single states by a factor of $1 / \sqrt{M}$. We note that by comparing gradients for the dataset case to the i.i.d. one, one can gauge the effect of the correlations in-between gradients yielded from distinct random states. The greater the correlation, the more similar these gradients are.
As can be seen, while generally resulting in a decrease of the variances, the effect of averaging the empirical loss over a dataset does not change the overall scaling of the variances (exponentially decaying for the graph states corresponding to an Erd\"{o}s–R\'{e}nyi distribution, but only polynomially decaying otherwise).

\section{Supplementary Methods 11: Comparison between an $S_n$-equivariant QNN and a non-symmetry respecting QNN: Classifying Connected vs Disconnected Graphs }\label{ap:sec:comparison}

To further substantiate our claims on the efficacy of $S_n$-equivariant QNNs compared to non-equivariant models, we consider the problem of classifying connected vs disconnected graphs encoded as graph states ($\phi = \pi)$. While this problem can be solved readily with classical methods, it still serve as an instructive example as connectivity is clearly invariant under relabeling. The number of non-isomorphic graphs is exponential, but we train with only a polynomial number of random graphs to illustrate the small data requirements of $S_n$-equivariant QNNs. To reach overparametrization, we train models with $\Theta(n^3)$ layers. Our EQNN model follows the structure in Fig.~2 of the main text and ends with a measurement of $O=Z^{\otimes n}$. Then, the prediction of the classifier is given by $h_{\thv}(\rho) = \text{sgn}(\Tr[\mathcal{U}_{\thv}(\rho)Z^{\otimes n}])$.

As a comparison, we also train standard hardware-efficient ansatzes \cite{kandala2017hardware} with the same number of parameters. Here the layers are composed of local unitaries (on single qubits) followed by ladders of CNOTs on adjacent parties. At the end, $X_1+X_2$ is measured and again classification given by the sign of the expectation value.

In Fig.~\ref{fig:eqnn_vs_hea} we give an example for $n=7$. First, note that the EQNN converges in less than 100 epochs while the standard QNN requires $\sim 250$. Most importantly, the training and testing accuracies of the EQNN closely match. The standard QNN, on the other hand, achieves similar training accuracy, but severely overfits and performs poorly in testing. Within $1000$ epochs the QNN fails to achieve high training accuracy, unlike the EQNN.                                              

\begin{figure*}[t]
    \centering
    \includegraphics[width=0.8\linewidth]{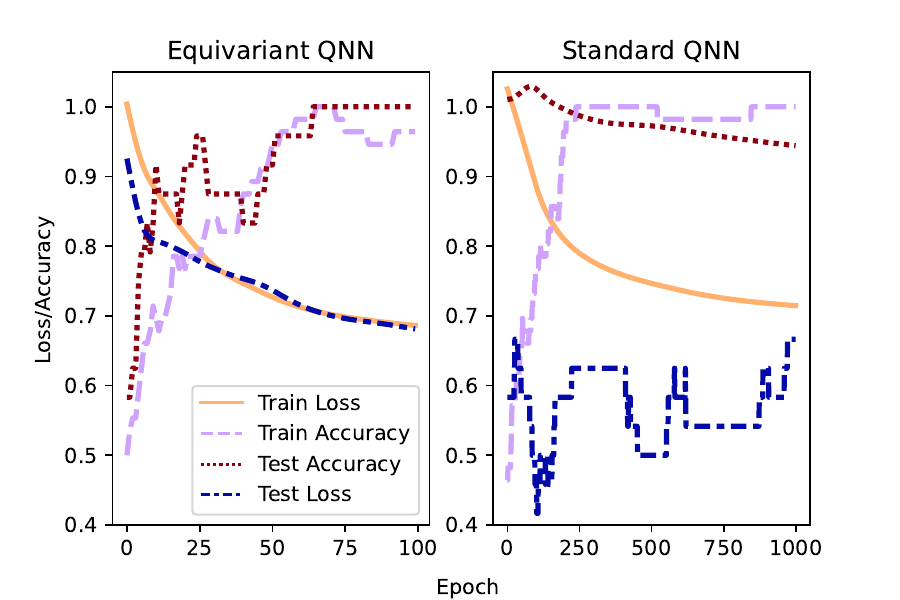}
    \caption{\textbf{Training and testing performance for a permutation-equivariant QNN versus a standard QNN learning to classify connected vs disconnected graphs of 7 nodes.} Each epoch (x axis) is a single step of optimization consisting of takes gradients over two batches of the data. Both models have $120$ free parameters, a count such that the equivariant model reaches overparametrization. Models were trained with a set of 14 random connected and 14 disconnected graphs. Testing was conducted with 6 new random connected and 6 disconnected graphs. For both architectures, the best model out of 15 random initial settings was selected. Note that the equivariant QNN converges within 100 epochs with small generalization error while the standard model requires several hundred and heavily overfits the data.}
    \label{fig:eqnn_vs_hea}
\end{figure*}

\section*{Supplementary Note}

Soon after our work was completed the preprint~\cite{anschuetz2022efficient} was posted on the arXiv. Therein, the authors claim that the loss function of $S_n$-equivariant QNNs can be classically simulated. In this section we address these claims.  In particular, we aim at making two important points:

\begin{itemize}
    \item For most relevant cases in QML, the relevant algorithm in~\cite{anschuetz2022efficient} is not fully classical, as it require access to a quantum computer to obtain a ``classical description'' of the input data.   
    \item  Even if one is given such ``classical description'', the ensuing algorithm that replaces the use of a quantum neural networks scales extremely poorly with the number of qubits, meaning that it is not clear whether it is truly favorable to replace the quantum circuit by classical post-processing. 
\end{itemize}

\medskip

We begin by recalling that in  Ref.~\cite{anschuetz2022efficient}, the authors argue that for symmetries restrictive enough (such as the permutation group $S_n$) some computational tasks become classically tractable provided certain conditions. These conditions are defined as "\emph{given certain classical descriptions of the input or the results of efficiently obtainable quantum measurements}".
This are indeed crucial assumptions, which we argue in the following do not hold in many tasks of practical importance, including quantum machine learning (QML) problems as envisioned in our manuscript. 
We stress that we do not claim that no interesting problem can be addressed classically, but rather that this is not the case for  most QML problems. For instance, we have no doubt that evaluating eigenenergies (and efficiently describing their eigenstates) of $S_n$-invariant Hamiltonians can indeed be performed classically as evidenced by the Theorems of Ref.~\cite{anschuetz2022efficient}.

\medskip

Let us thus henceforth focus on QML applications. In particular, we will discuss  Corollary 12  and Theorem 8 of Ref.~\cite{anschuetz2022efficient}, since these address directly our QML setup (as evidenced by a reference to our manuscript provided just before the corollary). 
Here, the authors argue that one can evaluate the loss function $\ell_{\thv}(\rho_i)$ provided one is given    "\emph{efficiently obtainable quantum measurements}" on the quantum state $\rho_i$.  In practice, this materializes as being able to prepare $\rho_i$ and  implementing on a quantum computer  the quantum Schur transform~\cite{harrow2005applications}, which is denoted in Ref.~\cite{anschuetz2022efficient} as  $V_{STO}$. For completeness, we recall that the quantum Schur transform (QST) is a quantum primitive that promises to greatly accelerate several tasks provided access to quantum computing resources (e.g., see a discussion of the many known applications in Ch.6 of Ref.~\cite{harrow2005applications} or the introduction of Ref.~\cite{krovi2019efficient}).    
The QST  is deemed \emph{efficiently} implementable on a quantum device in the sense that it requires a polynomial number of operations and a logarithmic number of additional ancillas (See, e.g., Refs~\cite{bacon2005quantum,bacon2006efficient,kirby2017practical,krovi2019efficient} for different possible implementations). While the previous is theoretically true, to our knowledge  the QST  has not been demonstrated on  current devices due to its  high requirements for implementation.

\medskip

Given that the simulability of a generic $S_n$-invariant QML pipeline, based on Corollary~12 in Ref.~\cite{anschuetz2022efficient}, effectively requires accessing a quantum state (i.e., the input data) in the first place, coherent processing (i.e., performing a QST) and additional non-trivial measurements (i.e., classical shadows as mentioned when introducing Th.8 in Ref.~\cite{anschuetz2022efficient}), we feel comfortable classifying the algorithm of Ref.~\cite{anschuetz2022efficient} as truly \emph{non-classical} (and also  non-near-term).

\medskip

Next, let us discuss the second scenario envisioned in Ref.~\cite{anschuetz2022efficient}, which assumes that one has been "\emph{given certain classical description}" of the input states. 
Already, this assumption is non-realistic and does  not hold in generic QML problems, whereby input states can be accessed rather than being "\emph{described}" in a preferable basis. 
Still, one could question the case when the input data is classical and encoded in a quantum state. 
In such a case, the learner effectively has \emph{some} description of the input state. 
Still, with the exception of most basic forms of encoding, such description is not readily transferable in a form suitable to the classical simulations of Ref.~\cite{anschuetz2022efficient}.
To exemplify such situation, the authors address in Ref.~\cite{anschuetz2022efficient} the case where the input state is a predefined \emph{computational basis state}. 
Indeed, it is known that description of such state can be efficiently mapped to a description in the preferable basis. 
This can readily be extended to superposition of polynomially many computational basis states.
Yet, already the graph state encoding that was used as an example in our numerical experiments would escape such description. Hence, even in this classical input data scenario, we believe that for the majority of encoding (except, the simplest ones) obtaining an efficient classical description of the data without a quantum computer  is not a viable option. 


\medskip
 Lastly,  let us assume that one is given  adequate measurements or access to an appropriate classical description of the input state. Then, while it is true that the loss $\ell_{\thv}(\rho_i)$ can be computed classically, this requires incurring very large computational requirements. Explicitly, as mentioned by the authors of Ref.~\cite{anschuetz2022efficient} (Appendix F), the classical computation involves routines scaling as $\OC(n^{10})$ and $\OC(n^{15})$.
While polynomial in the system size $n$, such complexity may well prohibit any reasonable attempt of scaling past $n=20$ qubits. 

\medskip

This discussion shows that  if one has access to a quantum computer (as required by the algorithm in Ref.~\cite{anschuetz2022efficient}) it is not entirely obvious that one should  use it to obtain a classical description of the data (followed by expensive post-processing), rather than using the device to run an $S_n$- equivariant quantum neural networks with trainability guarantees. One notes that the run-time of our model scales like $n^6$ ($O(n^3)$ layers each composed of $O(n^2)$ operations). Then, our results indicate that using a quantum computer with a run-time scaling as $n^6(\ll n^{15})$ leads to overparametrized models with no barren plateaus; and thus there is hope of achieving a quantum advantage in spite of having such large symmetries.

\clearpage



\section*{Supplementary References}